\documentclass[runningheads,table]{llncs}
\usepackage{amsmath,amssymb,xspace,cancel}
\usepackage{lmodern}
\usepackage{bussproofs,stackengine}
\usepackage{enumitem}
\usepackage{tcolorbox}
\usepackage{wrapfig}
\usepackage{tikz}
\usetikzlibrary{arrows,arrows.meta,shapes,positioning,decorations.pathmorphing,patterns,backgrounds,calc,patterns.meta}
\usepackage{dot2texi}
\usepackage[disable]{todonotes}
\usepackage{thm-restate}
\usepackage{hyperref}
\usepackage{cleveref}
\usepackage{mdframed}
\usepackage{multirow}
\usepackage{rotating}
\usepackage{pifont}
\usepackage{bm}
\usepackage{hyperref}
\usepackage[outline]{contour}
\contourlength{2pt}
\usepackage{array}
\usepackage{arydshln}
\usepackage{bbding}

\newcommand\review[1]{#1}
\newcommand\cameraready[1]{}
\newcommand\fullversion[1]{#1}

\newcommand{\tods}{{TODs}}
\newcommand{\tod}{{TOD}}
\newcommand{\vampire}{\textsc{Vampire}}
\newcommand{\setof}[1]{\{#1\}}
\newcommand{\QEDsymbol}{\text{\ding{111}}}
\newcommand{\QED}{\hspace*{\stretch{1}}\QEDsymbol}

\def\orcidID#1{\href{http://orcid.org/#1}{\raisebox{-1.25pt}{\includegraphics{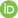}}}}

\newcommand\tsf[1]{\textsf{#1}}
\newcommand*{\defeq}{\stackrel{\scriptscriptstyle\tsf{def}}{=}}
\newcommand\eqs{\approx}

\newcommand\kbo{\tsf{kbo}}
\newcommand\lpo{\tsf{lpo}}

\newcommand\off{{\tt off}}
\newcommand\on{{\tt on}}
\newcommand\shared{{\tt shared}}

\newcommand\OD{\mathcal{T}}
\newcommand\ODprime{\mathcal{T}'}

\newcommand\comp{\mathrel{>^{\hspace{-1pt}?}\hspace{-1pt}}}

\newcommand\pos{\mathrel{\geq^{\hspace{-1pt}?}\hspace{-1pt}}}

\makeatletter
\newcommand\notsotiny{\@setfontsize\notsotiny\@vipt\@viipt}
\makeatother

\makeatletter
\def\rightedgefill@{\arrowfill@\relbar\relbar\rightarrow}
\def\leftedgefill@{\arrowfill@\leftarrow\relbar\relbar}
\newcommand\xedge[2][]{\ext@arrow 0359\rightedgefill@{#1}{#2}}
\makeatother

\newcommand\Vars{\mathcal{V}}
\newcommand\Funs{\mathcal{F}}

\newcommand\Poly{\mathcal{P}(\Vars)}

\newcommand\pathhighlightcolor{blue!30!white}
\newcommand\visitedcolor{red!50}

\tikzset{
    > = Stealth,
    node/.style={draw=black,fill=white,label={[font=\notsotiny,label distance=1pt]#1}},
    visitedA/.pic = {\begin{pgfinterruptboundingbox}
      \node[] (over) {\large \todSymbol};
    \end{pgfinterruptboundingbox}},
    processed/.style={
        fill=white,
        postaction={
            pattern={Lines[angle=45,distance={2pt},line width=.5pt]},
            pattern color=\visitedcolor
        }
    },
    unprocessed/.style={},
    genericnode/.style={node,circle,inner sep=2pt},
    datanode/.style={node=#1,rounded rectangle,rounded rectangle arc length=90,inner sep=2pt,font=\small},
    sourcenode/.style={node=#1,circle,processed,inner sep=2.5pt},
    sinknode/.style={node=#1,circle,fill=black,inner sep=2.5pt},
    termnode/.style={node=#1,rectangle},
    polynode/.style={node=#1,chamfered rectangle,chamfered rectangle corners={north west, south east},inner sep=1pt},
    currtermnode/.style={termnode,fill=\pathhighlightcolor,draw=none,inner sep=6pt},
    currpolynode/.style={polynode,fill=\pathhighlightcolor,draw=none,inner sep=3pt},
    currexitnode/.style={sinknode,fill=\pathhighlightcolor,draw=none,inner sep=4pt},
    currsuccessnode/.style={datanode,fill=\pathhighlightcolor,draw=none,inner sep=5pt},
    traceedge/.style={line width=5pt,line cap=rect,\pathhighlightcolor},
    edgenode/.style={pos=.45,fill=white,draw,circle,inner sep=1pt},
    gtedgenode/.style={edgenode,node contents={\tiny $>$}},
    geqedgenode/.style={edgenode,node contents={\tiny $\geq$},inner sep=.7pt},
    eqedgenode/.style={edgenode,node contents={\tiny $=$},inner sep=1.4pt},
    ngeqedgenode/.style={edgenode,node contents={\tiny $\ngeq$},inner sep=.2pt},
    canvassed/.style={execute at end picture={
        \begin{scope}[on background layer]
        \draw[draw=none,fill=black!3,rounded corners=1ex] ($(current bounding box.south west) - (.1,.1)$) rectangle ($(current bounding box.north east) + (.1,.1)$);
        \end{scope}
    }},
    transformedge/.style={-{Classical TikZ Rightarrow},double},
    labelnode/.style={draw,yshift=-1.2em,xshift=-.5em,fill=white},
}

\spnewtheorem{Definition}{Definition}{\bfseries}{\normalfont}

\newcommand\PostOrdering{{post-ordering}\xspace}
\newcommand\EqSet{\mathcal{E}\xspace}

\renewcommand{\paragraph}[1]{\par\smallskip\noindent\textbf{#1}}

\title{Term Ordering Diagrams}

\author{M\'arton Hajdu\inst{1}\textsuperscript{(\Envelope)}\orcidID{0000-0002-8273-2613}
\and
Robin Coutelier\inst{1}\orcidID{0009-0002-4735-5215}
\and
Laura Kov\'acs\inst{1}\orcidID{0000-0002-8299-2714}
\and
Andrei Voronkov\inst{2,3}
}
\institute{TU Wien, Vienna, Austria\\
\email{marton.hajdu@tuwien.ac.at}
\and University of Manchester, Manchester, UK \and EasyChair, Manchester, UK}
\date{}

\authorrunning{M\'arton Hajdu, Robin Coutelier, Laura Kov\'acs and Andrei Voronkov}

\begin{document}

\maketitle

\begin{abstract}
The superposition calculus for reasoning in first-order logic with equality relies on simplification orderings on terms. Modern saturation provers use the Knuth-Bendix order (KBO) and the lexicographic path order (LPO) for discovering redundant clauses and inferences. Implementing term orderings is however challenging. While KBO comparisons can be performed in linear time and LPO checks in quadratic time, using the best known algorithms for these orders is not enough. Indeed, our experiments show that for some examples term ordering checks may use about 98\% of the overall proving time. The reason for this is that some equalities that cannot be ordered can become ordered after applying a substitution (post-ordered), and we have to check for post-ordering repeatedly for the same equalities.
In this paper, we show how to improve post-ordering checks by introducing a new data structure called \emph{term ordering diagrams}, in short \tods{}, which creates an index for these checks. We achieve efficiency by lazy modifications of the index and by storing and reusing information from previously performed checks to speed up subsequent checks. Our experiments demonstrate efficiency of TODs.
\end{abstract}

\section{Introduction}
%
Superposition-based theorem provers commonly use simplification orderings on terms to restrict their search space~\cite{NieuwenhuisRubio:HandbookAR:paramodulation:2001}. All top performing provers from the last CASC competitions~\cite{Sut16} -- \vampire{}~\cite{CAV13}, iProver~\cite{iProver}, E~\cite{E19} and Zipperposition~\cite{zipperposition} -- use the Knuth-Bendix order (KBO)~\cite{KBO} and some also use the lexicographic path order (LPO)~\cite{LPOOriginal}. There is a linear time KBO algorithm~\cite{ThingsToKnowWhenImplementingKBO} and a quadratic time LPO algorithm~\cite{LPOConstraintSolving}, so implementing ordering comparisons does not seem to be challenging.

We denote by $\succ$ the simplification order used by the superposition calculus, and we call an \emph{ordering comparison} the operation of checking whether $s\succ t$ holds for two terms $s$ and $t$. Although there are very efficient algorithms for ordering comparison, surprisingly, it can easily take up a significant part of a theorem prover's running time. For example, ordering comparison is required in the following frequently used inference rule called \emph{demodulation}:
\begin{center}
\begin{tabular}{c >{\centering}m{.1\linewidth} l}
\multirow{2}{*}{
\AxiomC{$l\eqs r$}
\AxiomC{$C[l\sigma]$}
\BinaryInfC{$C[r\sigma]$}
\DisplayProof}
& \multirow{2}{*}{where} & (1) $l\sigma\succ r\sigma$,\\
& & (2) $C[l\sigma]\succ (l\simeq r)\sigma$,\\
\end{tabular}
\end{center}
between a unit clause consisting of an equality $l\eqs r$ and a clause $C[l\sigma]$ containing a subterm $l\sigma$, where $\sigma$ is a substitution. We can remove the right premise $C[l\sigma]$ after applying this rule, so demodulation is very valuable for keeping search space smaller.

\paragraph{Post-ordering checks with unordered equalities.} Demodulation is applied only when the side condition $l \sigma \succ r\sigma$ is satisfied. Based on properties of simplification orders, this side condition always holds when $l \succ r$. A more challenging case is when $l \not\succ r$ and $r \not\succ l$ (we say that $l \eqs r$ is \emph{unordered}) and we repeatedly apply demodulation with the same left premise $l \eqs r$ and different right premises. We call the comparison of $l \sigma$ and $r\sigma$ a \emph{post-ordering check} to emphasize that $l \eqs r$ is not pre-ordered. 
The number of required post-ordering checks can be very large. For example, if we have a unit equality problem (which is not unusual in algebraic reasoning) and generate $10^6$ clauses, the number of demodulation inferences
can be of the order of $10^{12}$. Even if only 1\% of all equalities is unordered, we still can have $10^{10}$ post-ordering comparisons. Further, while KBO has a linear time algorithm~\cite{ThingsToKnowWhenImplementingKBO}, it still takes a significant time compared to other algorithms used by a theorem prover (such as matching or unification). Improving post-ordering checks could thus further improve equational reasoning. 


%

\paragraph{Equality retrieval with post-ordering checks.} This paper focuses on retrieving equalities that become ordered after applying a substitution. This is used in the superposition and demodulation rules. Namely, \review{our goal is to efficiently solve} the following problem.

\begin{mdframed}[frametitle={\fcolorbox{black}{white}{~The Post-Ordering Problem~}}, innertopmargin=0pt, innerbottommargin=8pt, frametitleaboveskip=-5pt, frametitlealignment=\center, frametitlefont=\scshape]
\label{todProblem}
Given a finite set $\EqSet$ of unordered equalities $l \eqs r_1, \ldots, l \eqs r_n$ with the same left-hand side and a substitution $\sigma$, \emph{retrieve} from $\EqSet$ equalities that satisfy $l\sigma \succ r_i\sigma$.
\end{mdframed}
We can be interested in retrieving one, some, or all equalities. When this problem is used in demodulation, one normally first uses an index that retrieves a term $l$, and only then unordered equalities $l \eqs r_1, \ldots, l \eqs r_n$ that form the set $\EqSet$. In this case all equalities in $\EqSet$ have the same left-hand side.
We assume that the same set $\EqSet$  of equalities can be repeatedly used for retrieval, interleaved with operations of adding new equalities to $\EqSet$ or removing existing equalities from it. As such, the post-ordering problem becomes a \emph{term indexing problem}~\cite{TermIndexing}, where the substitution $\sigma$ is a \emph{query substitution}.

\paragraph{{\bf Motivating Example.}}\label{page:motivating} Let us illustrate the \PostOrdering{} problem using a KBO with a constant weight function. Let $l$ and $r_1$ denote the terms $f(x,y)$ and $f(y,x)$, respectively, and let $l\eqs r_1$ be an equality. Given a substitution $\sigma$, checking whether $l\sigma\succ r_1\sigma$ holds involves computing term weights and comparing symbol precedences of $f(x,y)\sigma$ and $f(y,x)\sigma$ recursively. Some of these operations are, however, independent of the substitution $\sigma$ applied. For example, the weight of $f(x,y)\sigma$ is the same as the weight of $f(y,x)\sigma$ independently of $\sigma$, so the weight check can be dropped. Simple analysis shows that checking  $l\sigma\succ r_1\sigma$ can be simplified to an \emph{equivalent} ordering check $x\sigma\succ y\sigma$. Note that this analysis has to be performed only once, so all consequent checks of $f(x,y)\sigma' \succ f(y,x)\sigma'$ for any substitution $\sigma'$ can be immediately reduced to checking $x\sigma' \succ y\sigma'$. 

We can do even better when we perform several consecutive checks. For example, suppose that we retrieve equalities in the same order and the next equality after $f(x,y) \eqs f(y,x)$ is $f(x,y) \eqs f(x,x)$. Then, if during the first check we established $y\sigma \succ x\sigma$, we can immediately conclude $f(x,y)\sigma \succ f(x,x)\sigma$.



%
\paragraph{Our contributions.} 
This paper is based on ideas explained in the example above: (i) simplify ordering checks using the definition of KBO or LPO, and (ii) use the previous computation history to get rid of redundant checks. We bring the following contributions. 
\begin{itemize}
\item We introduce a new data structure, called the \emph{term ordering diagram (\tod{})}, in  Section~\ref{sec:term_ordering_diagrams}. 
\item We describe \tod{} transformations in Section~\ref{sec:transformations}.  These transformations use KBO/LPO properties and  produce equivalent, yet more efficient TODs.\fullversion{\footnote{Detailed proofs can be found in the Appendix.}}
\cameraready{\footnote{Detailed proofs can be found in the extended version of this paper~\cite{FullPaper}.}}
\item We propose an equality retrieval algorithm in Section~\ref{sec:retrieval} for solving the \PostOrdering{} problem using TODs. To the best of our knowledge, our work provides the first algorithm for efficient solving of \PostOrdering{} checks. It is also the first paper discussing runtime specialization of LPO checks, previously used for other operations in~\cite{PartiallyAdaptiveCodeTrees,EfficientCheckingOfTermOrderingConstraints}.
\item We evaluate our work by implementing term ordering diagrams in \vampire{} and report on our results in Section~\ref{sec:evaluation}. The use of TODs in forward demodulation, that is, demodulation applied to newly generated clauses from previously stored unit equalities~\cite{CAV13}, has significantly contributed to \vampire{}'s success in the unit equality division (UEQ) of the CASC competition~\cite{Sut16} in 2024.
\end{itemize}


%

\section{Preliminaries}
We work with a fixed signature $\Funs$ consisting of a finite set of \emph{function symbols} with associated arities and consider an alphabet of \emph{variables} $\Vars$; variables are not part of the signature.
\emph{Function symbols} are denoted by $f, g$, and variables by $x, y$, possibly with indices. \emph{Terms} are defined in the standard way over $\Funs\cup\Vars$;  variable-free terms are called \emph{ground}. A \textit{substitution} $\sigma$ is a mapping from variables to terms, such that the set  $\{x\mid \sigma(x)\neq x\}$ of variables is finite.
We call a \emph{simplification ordering} any ordering $\succ$ on terms that is
\begin{enumerate}[label=(\arabic*)]
\item \emph{well-founded}: there is no infinite decreasing chain of terms $s_1\succ s_2\succ \ldots$,
\item \emph{stable under substitutions}: if $s\succ t$ then $s\sigma\succ t\sigma$, 
\item \emph{monotonic}: for  $f\in\Funs$ and  terms $s_1,\ldots,s_n,s,t$ such that $s \succ t$, we have $f(s_1,\dots,s_{i-1},s,s_{i},\dots,s_n)\succ f(s_1,\dots,s_{i-1},t,s_{i},\dots,s_n)$, 
\item satisfies the \emph{subterm property}: if $t$ is a proper subterm of $s$, then $s \succ t$.
\end{enumerate}
%
A \emph{precedence relation}, denoted by $\gg$, is a total order on the signature $\Funs$. A \emph{weight function} is a function $w$ from $\Funs$ to non-negative integers. We will refer to $w(f)$ as the \emph{weight} of $f$. We denote by $w_0$ the smallest weight of constants.
%
%
For $p\in\Funs\cup\Vars$, we write $|t|_p$ to denote  the \emph{number of occurrences} of $p$ in 
term $t$. \review{For example, $|f(x,x)|_f=1$, $|f(x,x)|_x=2$ and $f(x,x)_y=0$.} Let $\Poly$ be the set of linear expressions over $\Vars$ with integer coefficients. The \emph{weight of a term} $t$, denoted by $|t|$, is a linear expression in $\Poly$ defined as:
$$|t|\defeq\sum_{f\in\Funs}|t|_f\cdot w(f) + \sum_{x\in\Vars}|t|_x\cdot x$$
A substitution $\sigma$ can also be considered as mapping from linear expressions to linear expressions, as follows: 
\[
\sigma(\alpha_0 + \alpha_1\cdot x_1 + \ldots + \alpha_n\cdot x_n)\defeq \alpha_0 + \alpha_1\cdot |x_1\sigma| + \ldots + \alpha_n\cdot |x_n\sigma|.
\]
\review{For example, if $w(f)=2$, $w(a)=1$ and $\sigma=\{x\mapsto a\}$, then $|f(x,x)|=2\cdot x+2$ and $\sigma(|f(x,x)|)=\sigma(2\cdot x+2)=2\cdot|x\sigma|+2=4$.}
It is not hard to argue that $|t\sigma|$ = $\sigma(|t|)$. 
Let $e\in\Poly$ be a linear expression. We call a substitution $\sigma$ \emph{grounding} for $e$ if $\sigma(e)$ does not contain variables \review{and $|x\sigma|\geq w_0$ for all $x\in\Vars$ such that $|x\sigma|$ is a constant}.
We write $e > 0$ if $\sigma(e) > 0$ for all grounding substitutions $\sigma$ for $e$. We write $e \gtrsim 0$ if $\sigma(e) \geq 0$ for all grounding substitution $\sigma$ for $e$.
The \emph{Knuth-Bendix order (KBO)}, denoted by $\succ_\kbo$, is parameterized by a precedence relation $\gg$ and a weight function $w$. For terms $s,t$, we have  $s \succ_\kbo t$ if:
\begin{enumerate}[label=(K\arabic*),itemsep=2pt,leftmargin=2.5em]
\item\label{prop:kbo1} $|s|-|t|>0$, or
\item\label{prop:kbo2} $|s|-|t|\gtrsim 0$, $s = f(s_1,...,s_n)$, $t = g(t_1,...,t_m)$ and $f\gg g$, or
\item\label{prop:kbo3} $|s|-|t|\gtrsim 0$, $s = f(s_1,...,s_n)$, $t = f(t_1,...,t_n)$ and there exists $1\le i\le n$ such that $s_i\succ_\kbo t_i$ and
$s_j=t_j$ for all $1\le j<i$.
\end{enumerate}
The \emph{lexicographic path order (LPO)}, denoted by $\succ_\lpo$, is  parameterized by a precedence relation $\gg$. Let   $s,t$ be terms with  $s=f(s_1,...,s_n)$. We write  $ s \succ_\lpo t$ if:
\begin{enumerate}[label=(L\arabic*),itemsep=2pt,leftmargin=2.5em]
\item\label{prop:lpo1} $t=f(t_1,...,t_n)$ and there exists $1\le i\le n$ s.t. $s_j=t_j$ for  $1\le j<i$, $s_i\succ_\lpo t_i$, and $s\succ_\lpo t_k$ for  $i<k\le n$, or
\item\label{prop:lpo2} $t=g(t_1,...,t_m)$, $f\gg g$ and $s\succ_\lpo t_i$ for  $1\le i\le m$, or
\item\label{prop:lpo3} $s_i\succeq_\lpo t$ for some $1\le i\le n$.
\end{enumerate}
\review{It is known that LPO is a simplification order; and that KBO is a simplification order for any precedence relation $\gg$ and weight function $w$, if $w_0>0$ and $f\gg g$ for all $g\in\Funs$ different from $f$, for all unary $f\in\Funs$ with $w(f)=0$.}

\section{Term Ordering Diagram -- \tod}
\label{sec:term_ordering_diagrams}

To solve the \PostOrdering{} problem, 
we introduce the \emph{term ordering diagram (TOD)} data structure  (Definition~\ref{def:tod}). For retrieving equalities using TODs we will need the following two operations on substitutions.



\begin{enumerate}[itemsep=2pt,leftmargin=1.3em]
\item Given two terms $s$ and $t$, compare $s\sigma$ and $t\sigma$ using $\succ$. We call the expression $s \comp t$ a \emph{term comparison}. We consider a substitution $\sigma$ as a mapping from term comparisons to the set of values $\{ >, =, \ngeq\}$, as follows:

\[
  \sigma(s \comp t) \defeq 
    \left\{
      \begin{array}{rl}
         >,  & \text{~~ if } s\sigma \succ t\sigma, \\
         =, & \text{~~ if } s\sigma = t\sigma, \\
         \ngeq, & \text{~~ otherwise. }
      \end{array}
    \right.
\]


\item Given a linear expression $e$, check the sign of the linear expression $\sigma(e)$. We call the expression $e\pos 0$ a \emph{positivity check}. We consider a substitution $\sigma$ as a mapping from positivity checks to the set of values $\{ >, \geq, \ngeq\}$, as follows:

\[
  \sigma(e \pos 0) \defeq 
    \left\{
      \begin{array}{rl}
         >,  & \text{~~ if } \sigma(e)>0, \\
         \geq, & \text{~~ if } \sigma(e)\gtrsim 0,\\
         \ngeq, & \text{~~ otherwise. }
      \end{array}
    \right.
\]
Note that positivity checks are only defined for KBO.

\end{enumerate}
Before defining \tods{}, we illustrate how we solve the \PostOrdering{} problem.
\begin{example}\label{ex:tod1}
Let $\sigma$ be a query substitution and  $l\defeq f(x,y)$, $r_1\defeq f(y,x)$ be terms. The rooted directed graph of Figure~\ref{fig:example1}(a) illustrates the two key operations for retrieving the equality $l\eqs r_1$ if $f(x,y)\sigma\succ f(y,x)\sigma$ holds. We first evaluate the \emph{term comparison} $f(x,y)\comp f(y,x)$ in the top node. If $\sigma(f(x,y)\comp f(y,x))$ evaluates to $>$, the computation proceeds to the bottom node, which contains the equality $l\eqs r_1$, which is then included in the result of the retrieval.

The evaluation of $\sigma(f(x,y)\comp f(y,x))$ in Figure~\ref{fig:example1}(a)  uses  KBO, which in turn  computes the linear expression $|f(x,y)|-|f(y,x)|$, 
performs a \emph{positivity check} and proceeds with one of the subcases~\ref{prop:kbo1}--\ref{prop:kbo3}. The linear expression is 0 as both terms have the same number of $f$, $x$ and $y$ symbols. Hence, the positivity check always results in $\geq$, which rules out subcase~\ref{prop:kbo1}. Similarly, \ref{prop:kbo2} is violated, as $f(x,y)$ and $f(y,x)$ both have $f$ as top-most symbol.

Note that \ref{prop:kbo1} and \ref{prop:kbo2} are not applicable \emph{regardless of the query substitution $\sigma$}, hence we can simplify $f(x,y)\sigma\succ f(y,x)\sigma$ into $(x\succ y\lor (x=y\land y\succ x))\sigma$. As $(x=y\land y\succ x)\sigma$ is false, this further simplifies into $(x\succ y)\sigma$. The term comparison $f(x,y)\comp f(y,x)$ of Figure~\ref{fig:example1}(a) is thus turned into the (cheaper) term comparison $x\comp y$ shown in Figure~\ref{fig:example1}(b).

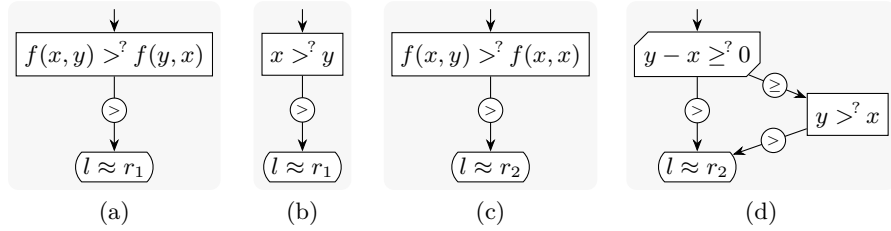
\begin{figure}[t]
\begin{minipage}{.25\linewidth}

\begin{center}
\begin{tikzpicture}[canvassed]

\coordinate (source) at (0,.6) {};
\node[termnode] (c1) at (0,0) {$f(x,y)\comp f(y,x)$};
\node[datanode] (s1) at (0,-1.5) {$l\eqs r_1$};

\draw[->] (source) -- (c1);
\draw[->] (c1) -- node[midway,gtedgenode]{} (s1);

\end{tikzpicture}\\
(a)
\end{center}
\end{minipage}\begin{minipage}{.16\linewidth}

\begin{center}
\begin{tikzpicture}[canvassed]
\coordinate (source) at (0,.6) {};
\node[termnode] (c1) at (0,0) {$x\comp y$};
\node[datanode] (s1) at (0,-1.5) {$l\eqs r_1$};

\draw[->] (source) -- (c1);
\draw[->] (c1) -- node[midway,gtedgenode]{} (s1);

\end{tikzpicture}\\
(b)
\end{center}
\end{minipage}\begin{minipage}{.25\linewidth}

\begin{center}
\begin{tikzpicture}[canvassed]
\coordinate (source) at (0,.6) {};
\node[termnode] (c1) at (0,0) {$f(x,y)\comp f(x,x)$};
\node[datanode] (s1) at (0,-1.5) {$l\eqs r_2$};

\draw[->] (source) -- (c1);
\draw[->] (c1) -- node[midway,gtedgenode]{} (s1);

\end{tikzpicture}\\
(c)
\end{center}
\end{minipage}\begin{minipage}{.34\linewidth}

\begin{center}
\begin{tikzpicture}[canvassed]
\coordinate (source) at (0,.6) {};
\node[polynode] (c1) at (0,0) {$y-x\pos 0$};
\node[datanode] (s1) at (0,-1.5) {$l\eqs r_2$};
\node[termnode] (c2) at (2,-.8) {$y\comp x$};

\draw[->] (source) -- (c1);
\draw[->] (c1) --  node[midway,gtedgenode]{} (s1);
\draw[->] (c1) -- node[midway,geqedgenode]{} (c2);
\draw[->] (c2) -- node[midway,gtedgenode]{} (s1);

\end{tikzpicture}\\
(d)
\end{center}
\end{minipage}
\caption{Equality retrievals, where $l\defeq f(x,y)$, $r_1\defeq f(y,x)$ and $r_2\defeq f(x,x)$. (a) and (b) show retrieval of $l\eqs r_1$, where (a) contains a term comparison $l\comp r_1$ and (b) is a simplified version of (a) with term comparison $x\comp y$. Further,  (c) and (d) show retrieval of $l\eqs r_2$, where (c) contains term comparison $l\comp r_2$ and (d) is a simplified version of (c) with a positivity check $y-x\pos 0$ and term comparison $y\comp x$.  }
\label{fig:example1}
\end{figure}

Let us now consider an additional equality $l\eqs r_2$, with $r_2\defeq f(x,x)$. Figure~\ref{fig:example1}(c) displays the steps of  retrieving $l\eqs r_2$ if $f(x,y)\sigma\succ f(x,x)\sigma$ holds. Similarly to Figure~\ref{fig:example1}(a), the \emph{term comparison} $f(x,y)\comp f(x,x)$ is simplified using KBO properties into $(y-x>0\lor (y-x\gtrsim 0\land y\succ x))\sigma$.

The simplified computation is depicted in Figure~\ref{fig:example1}(d), where we first perform a \emph{positivity check} $y-x\pos 0$, corresponding to deciding which of $\sigma(y-x)>0$ or $\sigma(y-x)\gtrsim 0$ hold in the above formula, then performing at most one more term comparison.
The computation of Figure~\ref{fig:example1}(d) is more efficient than the one in Figure~\ref{fig:example1}(c) because we avoid the computation of the linear expression $|f(x,y)|-|f(x,x)|$ and as well as some intermediate term comparisons.\QED

\end{example}

Example~\ref{ex:tod1} shows that interleaving term comparisons and positivity checks can simplify, and hence speed up, term ordering checks. Let us now give formal definitions. We will use standard graph-theoretic notions related to directed acyclic graphs (dags), trees, and paths in a graph.

\begin{Definition}[Term Ordering Diagram -- TOD]
\label{def:tod}
A \emph{term ordering diagram} (TOD) is a directed acyclic graph $\OD$ which contains five kinds of nodes:

\begin{enumerate}[label=(\arabic*),itemsep=2pt,leftmargin=1.7em]
    \item A single \emph{root node}, so that every node is reachable from the root node.
    The root node has 
    one outgoing edge.
    \item A single \emph{exit node}, reachable from every node.
    \item 
    A \emph{term comparison node} is labeled by a term comparison $s \comp t$ and has 
    three outgoing edges, labeled by $>$, $=$ and $\ngeq$. 
    \item  
    A \emph{positivity check node} is labeled by a positivity check $e \pos 0$ and has 
    three outgoing edges, labeled by $>$, $\geq$, and $\ngeq$.
    \item 
    A  \emph{success node} is labeled by an equality $l \eqs r$ and has 
    one outgoing edge.\QED
\end{enumerate}
\end{Definition}
We  collectively refer to term comparison nodes and positivity check nodes as \emph{evaluation nodes}. The idea is that we evaluate the substitution in these nodes and follow the node label resulting from the evaluation. Note that the exit node and its incoming edges are usually omitted when displaying a TOD (as in~\Cref{fig:example1}).
To retrieve equalities towards solving the \PostOrdering{} problem, we traverse a \tod{} as follows. 
\begin{Definition}[TOD Traversal]\label{def:tod:traversal:foces}
Let $\sigma$ be a query substitution and $\OD$ a \tod. 
The \emph{$\OD$ traversal for $\sigma$} is the path in $\OD$ from its root  to its exit node so that:  
\begin{enumerate}[label=(\arabic*)]
\item If the path contains a term comparison node labeled by $s \comp t$, then the path also contains the outgoing edge of this node labeled by $\sigma(s \comp t)$.
\item If the path contains a positivity check node labeled by $e \pos 0$, then the path also contains the outgoing edge of this node labeled by $\sigma(e \pos 0)$.
\end{enumerate}
The \emph{success set of the $\OD$  traversal for $\sigma$} is the set of labels of all success nodes on this path; we also refer to it as \emph{the success set of $\sigma$ in $\OD$}.
For any non-exit node in a traversal, we refer to its  \emph{next edge} and \emph{next node}.

For a  node $n$ labeled by $c$, where $c$ is either a term comparison or a positivity check, we say that the node $n$ \emph{forces an edge label} $\ell$, if for every query substitution $\sigma$, if the $\OD$ traversal for $\sigma$ reaches $n$, then $\sigma(c) = \ell$.\QED
\end{Definition}

We use a \tod{} to perform term ordering checks only when it is necessary. Key to this are \tod{} transformations (Section~\ref{sec:transformations}): given a \tod{} $\OD_1$, we transform $\OD_1$ into an equivalent \tod{} $\OD_2$, so that 
$\OD_2$ can be traversed faster, making ordering checks cheaper. The efficiency gain in $\OD_2$ traversal comes  by replacing $\OD_1$ checks by less expensive ones or removing redundant checks of $\OD_1$. 
Importantly, our \tod{} transformations on $\OD_1$ are performed while we traverse $\OD_1$ for a specific query substitution, as shown  in 
Example~\ref{ex:tod2}. 

%
%

\begin{example}\label{ex:tod2}
Consider again 
$l\defeq f(x,y)$, $r_1\defeq f(y,x)$ and $r_2\defeq f(x,x)$. Let $\sigma$ be a query substitution. 
Figure~\ref{fig:example2}(a) shows the sequence of steps for retrieving $l\eqs r_1$ if $(l\succ r_1)\sigma$, and then retrieving $l\eqs r_2$ if $(l\succ r_2)\sigma$. 
Similarly to Figures~\ref{fig:example1}(a) and (c), we can modify Figure~\ref{fig:example2}(a) into its more efficient version displayed in Figure~\ref{fig:example2}(b). However, in Figure~\ref{fig:example2}(b) we  implement stronger \tod{} modifications:  in the (second) retrieval of $l\eqs r_2$ we  use information from the (first) retrieval of $l\eqs r_1$.
\review{For example, if $x\sigma=y\sigma$, the evaluation of $x\comp y$ in~\Cref{fig:example2}(b) is sufficient to conclude that neither $l\eqs r_1$ nor $l\eqs r_2$ are retrieved; saving computation compared to its equivalent version in~\Cref{fig:example2}(a) where both $f(x,y)\comp f(y,x)$ and $f(x,y)\comp f(x,x)$ have to be evaluated to reach this conclusion.}
\QED
\begin{figure}[t]

\begin{minipage}{.52\linewidth}

\begin{center}
\scalebox{.94}{
\begin{tikzpicture}[canvassed]
\coordinate (source) at (0,.6) {};
\node[termnode] (c1) at (0,0) {$f(x,y)\comp f(y,x)$};
\node[datanode] (s1) at (0,-1.5) {$l\eqs r_1$};
\node[termnode] (c2) at (3.5,0) {$f(x,y)\comp f(x,x)$};
\node[datanode] (s2) at (3.5,-1.5) {$l\eqs r_2$};

\draw[->] (source) -- (c1);
\draw[->] (c1) -- node[midway,gtedgenode]{} (s1);
\draw[->] (c1) edge[bend left=10] node[midway,eqedgenode]{} (c2);
\draw[->] (c1) edge[bend right=10] node[midway,ngeqedgenode]{} (c2);
\draw[->] (s1) -- (c2);
\draw[->] (c2) -- node[midway,gtedgenode]{} (s2);

\end{tikzpicture}}\\
(a)
\end{center}

\end{minipage}\begin{minipage}{.48\linewidth}

\begin{center}
\scalebox{.94}{
\begin{tikzpicture}[canvassed]
\coordinate (source) at (0,.6) {};
\node[termnode] (c1) at (0,0) {$x\comp y$};
\node[datanode] (s1) at (0,-1.5) {$l\eqs r_1$};
\node[polynode] (c2) at (2.5,0) {$y-x\pos 0$};
\node[datanode] (s2) at (2.5,-1.5) {$l\eqs r_2$};
\node[termnode] (c3) at (4.5,-.8) {$y\comp x$};

\draw[->] (source) -- (c1);
\draw[->] (c1) -- node[midway,gtedgenode]{} (s1);
\draw[->] (c1) -- (c2) node[midway,ngeqedgenode]{} (c2);
\draw[->] (s1) -- (c2);
\draw[->] (c2) -- node[midway,gtedgenode]{} (s2);
\draw[->] (c2) -- node[midway,geqedgenode]{} (c3);
\draw[->] (c3) -- node[midway,gtedgenode]{} (s2);

\end{tikzpicture}}\\
(b)
\end{center}

\end{minipage}
\caption{Retrieving multiple equalities, where 
$l\defeq f(x,y)$, $r_1\defeq f(y,x)$ and $r_2\defeq f(x,x)$. }
\label{fig:example2}
\end{figure}
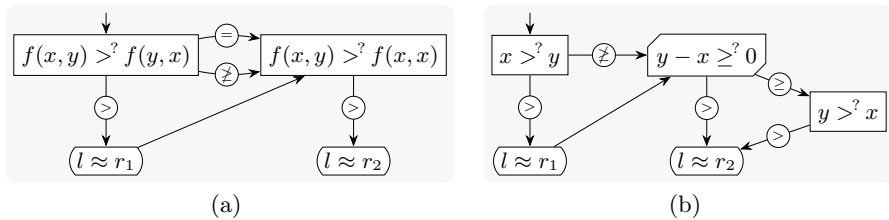
\end{example}

Examples~\ref{ex:tod1}--\ref{ex:tod2} show that \emph{we transform  a \tod{} while we are traversing it}. This is done for two reasons. (i) First, term comparisons may involve expensive computations, for example, when we deal with linear expressions (as shown in Example~\ref{ex:tod1}). When we compute a linear expression we modify the \tod{} to have a corresponding positivity check node, so that on the next traversal, if we reach the same node, we do not compute the linear expression again. 
(ii) 
Second, our transformations are \emph{lazy} -- we only modify nodes in a \tod{} when we reach them during traversal. The reason for this is that modifying \tods{} is expensive due to KBO/LPO checks; hence, these modifications should not be performed for nodes that will never be visited.
Despite modifying \tods, we  do not change their semantics:  
our \tod{} transformations (Section~\ref{sec:transformations}) yield equivalent \tods.
\begin{Definition}[TOD Equivalence\label{def:tod-equiv}]
Two \tods{} $\OD_1$ and $\OD_2$ are \emph{equivalent} if for every substitution $\sigma$, the success sets of $\sigma$ in $\OD_1$ and $\OD_2$ coincide.\QED
\end{Definition}
%

\newsavebox\redundantNodeOne

\begin{lrbox}{\redundantNodeOne}
    \begin{tikzpicture}[canvassed]
    \node[genericnode] (n1) at (0,0) {$n_1$};
    \node[genericnode,inner sep=3.5pt] (n2) at (-.8,-1.5) {$n$};
    \node[genericnode] (n3) at (.8,-1.5) {$n_2$};

    \draw[->] (n1) -- (n2) node[edgenode,inner sep=.5pt]{\scriptsize $\ell'$};
    \draw[->] (n2) -- (n3) node[edgenode]{\scriptsize $\ell$};
    \end{tikzpicture}
\end{lrbox}

\newsavebox\redundantNodeTwo

\begin{lrbox}{\redundantNodeTwo}
    \begin{tikzpicture}[canvassed]
    \node[genericnode] (n1) at (0,0) {$n_1$};
    \node[genericnode,inner sep=3.5pt] (n2) at (-.8,-1.5) {$n$};
    \node[genericnode] (n3) at (.8,-1.5) {$n_2$};

    \draw[->] (n1) -- (n3) node[edgenode,inner sep=.5pt]{\scriptsize $\ell'$};
    \draw[->] (n2) -- (n3) node[edgenode]{\scriptsize $\ell$};
    \end{tikzpicture}
\end{lrbox}

\newsavebox\nodeReplicationOne

\begin{lrbox}{\nodeReplicationOne}
    \begin{tikzpicture}[canvassed]
    \node[inner sep=0pt] (a) at (-.8,-1) {};
    \node[genericnode] (n1) at (0,0) {$n_1$};
    \node[genericnode,inner sep=3.5pt] (n2) at (0,-1.5) {$n$};

    \draw[->] (a) -- (n2);
    \draw[->] (n1) -- (n2) node[edgenode]{\scriptsize $\ell$};
    \end{tikzpicture}
\end{lrbox}

\newsavebox\nodeReplicationTwo

\begin{lrbox}{\nodeReplicationTwo}
    \begin{tikzpicture}[canvassed]
    \node[inner sep=0pt] (a) at (-1.3,-1) {};
    \node[genericnode] (n1) at (0,0) {$n_1$};
    \node[genericnode,inner sep=3.5pt] (n2) at (-.5,-1.5) {$n$};
    \node[genericnode] (n3) at (1,-1.5) {$n_2$};

    \draw[->] (a) -- (n2);
    \draw[->] (n1) -- (n3) node[edgenode]{\scriptsize $\ell$};
    \end{tikzpicture}
\end{lrbox}

\newsavebox\insertionOne

\begin{lrbox}{\insertionOne}
    \begin{tikzpicture}[canvassed]
    \node[inner sep=0pt] (source1) at (-.6,.6) {};
    \node[inner sep=0pt] (sourcen) at (.6,.6) {};
    \node at (0,.6) {$\ldots$};
    \node[sinknode] (sink) at (0,0) {};

    \draw[->] (source1) -- (sink);
    \draw[->] (sourcen) -- (sink);
    \end{tikzpicture}
\end{lrbox}

\newsavebox\insertionTwo

\begin{lrbox}{\insertionTwo}
    \begin{tikzpicture}[canvassed]
    \node[inner sep=0pt] (source1) at (-.6,.6) {};
    \node[inner sep=0pt] (sourcen) at (.6,.6) {};
    \node at (0,.6) {$\ldots$};
    \node[termnode] (c1) at (0,0) {$l\comp r$};
    \node[datanode] (s1) at (0,-1.5) {$l\eqs r$};
    \node[sinknode] (sink) at (1.5,-.75) {};

    \draw[->] (source1) -- (c1);
    \draw[->] (sourcen) -- (c1);

    \draw[->] (c1) -- node[gtedgenode]{} (s1);
    \draw[->] (c1) edge[bend right=20] node[eqedgenode]{} (sink);
    \draw[->] (c1) edge[bend left=20] node[ngeqedgenode]{} (sink);
    \draw[->] (s1) -- (sink);
    \end{tikzpicture}
\end{lrbox}

\newsavebox\maincase

\begin{lrbox}{\maincase}
    \begin{tikzpicture}[canvassed]
    \node[termnode] (n) at (0,0) {$f(\bar{s})\comp g(\bar{t})$};
    \node[genericnode] (n1) at (-1.5,-1.5) {$n_1$};
    \node[genericnode] (n2) at (0,-1.5) {$n_2$};
    \node[genericnode] (n3) at (1.5,-1.5) {$n_3$};
    
    \draw[->] (n) -- (n1) node[midway,gtedgenode]{};
    \draw[->] (n) -- (n2) node[midway,eqedgenode]{};
    \draw[->] (n) -- (n3) node[midway,ngeqedgenode]{};
    \end{tikzpicture}
\end{lrbox}

\newsavebox\kboGreater

\begin{lrbox}{\kboGreater}
    \begin{tikzpicture}[canvassed]
    \node[polynode] (n) at (0,0) {$|f(\bar{s})|-|g(\bar{t})|\pos 0$};
    \node[genericnode] (n1) at (-1.5,-2) {$n_1$};
    \node[genericnode] (n2) at (0,-2) {$n_2$};
    \node[genericnode] (n3) at (1.5,-2) {$n_3$};
    
    \draw[->] (n) -- node[midway,gtedgenode]{} (n1);
    \draw[->] (n) -- (0,-1) -- node[geqedgenode,pos=0]{} (n1);
    \draw[->] (n) -- node[midway,ngeqedgenode]{} (n3);
    \end{tikzpicture}
\end{lrbox}

\newsavebox\kboLess

\begin{lrbox}{\kboLess}
    \begin{tikzpicture}[canvassed]
    \node[polynode] (n) at (0,0) {$|f(\bar{s})|-|g(\bar{t})|\pos 0$};
    \node[genericnode] (n1) at (-1.5,-2) {$n_1$};
    \node[genericnode] (n2) at (0,-2) {$n_2$};
    \node[genericnode] (n3) at (1.5,-2) {$n_3$};
    
    \draw[->] (n) -- node[midway,gtedgenode]{} (n1);
    \draw[->] (n) -- (0,-1) -- node[geqedgenode,pos=0]{} (n3);
    \draw[->] (n) -- node[midway,ngeqedgenode]{} (n3);
    \end{tikzpicture}
\end{lrbox}

\newsavebox\kboEqual

\begin{lrbox}{\kboEqual}
    \begin{tikzpicture}[canvassed]
    \node[polynode] (n) at (0,0) {$|f(\bar{s})|-|g(\bar{t})|\pos 0$};
    \node[termnode] (arg1) at (0,-1.3) {$s_1\comp t_1$};
    \node[] (argi) at (0,-2.4) {$\ldots$};
    \node[termnode] (argk) at (0,-3.5) {$s_k\comp t_k$};
    \node[genericnode] (n1) at (-2,-4.8) {$n_1$};
    \node[genericnode] (n2) at (0,-4.8) {$n_2$};
    \node[genericnode] (n3) at (2,-4.8) {$n_3$};
    
    \draw[->] (n) -- (-1,-1) -- node[gtedgenode,pos=0]{} (n1);
    \draw[->] (n) -- node[geqedgenode,pos=0.4]{} (arg1);
    \draw[->] (n) -- (1,-1) -- node[ngeqedgenode,pos=0]{} (n3);
    
    \draw[->] (arg1) -- (-1,-2.4) -- node[gtedgenode,pos=0]{} (n1);
    \draw[->] (arg1) -- node[eqedgenode,pos=0.4]{} (argi);
    \draw[->] (arg1) -- (1,-2.4) -- node[ngeqedgenode,pos=0]{} (n3);
    
    \draw[->] (argi) -- (-1,-3.3) -- node[gtedgenode,pos=0]{} (n1);
    \draw[->] (argi) -- node[eqedgenode,pos=0.4]{} (argk);
    \draw[->] (argi) -- (1,-3.3) -- node[ngeqedgenode,pos=0]{} (n3);
    
    \draw[->] (argk) -- node[gtedgenode]{} (n1);
    \draw[->] (argk) -- node[eqedgenode,pos=0.4]{} (n2);
    \draw[->] (argk) -- node[ngeqedgenode]{} (n3);
    \end{tikzpicture}
\end{lrbox}

\newsavebox\lpoGreaterBase

\begin{lrbox}{\lpoGreaterBase}
    \begin{tikzpicture}[canvassed]
    \node[node,circle] {$n_1$};
    \end{tikzpicture}
\end{lrbox}

\newsavebox\lpoGreater

\begin{lrbox}{\lpoGreater}
    \begin{tikzpicture}[canvassed]
    \node[termnode] (arg1) at (0,0) {$f(\bar{s})\comp t_1$};
    \node[] (argi) at (0,-1.5) {$\ldots$};
    \node[termnode] (argk) at (0,-3) {$f(\bar{s})\comp t_m$};
    \node[genericnode] (n1) at (0,-4.5) {$n_1$};
    \node[genericnode] (n2) at (1.5,-4.5) {$n_2$};
    \node[genericnode] (n3) at (3,-4.5) {$n_3$};

    \draw[->] (arg1) -- node[gtedgenode]{} (argi);
    \draw[->] (arg1) -- (1.5,-.8) -- node[eqedgenode,pos=0]{} (n3);
    \draw[->] (arg1) -- (2,-.7) -- node[ngeqedgenode,pos=0]{} (n3);

    \draw[->] (argi) -- node[gtedgenode]{} (argk);
    \draw[->] (argi) -- (1,-2.2) -- node[eqedgenode,pos=0]{} (n3);
    \draw[->] (argi) -- (1.5,-2.1) -- node[ngeqedgenode,pos=0]{} (n3);

    \draw[->] (argk) -- node[gtedgenode]{} (n1);
    \draw[->] (argk) -- (1.2,-3.6) -- node[eqedgenode,pos=0]{} (n3);
    \draw[->] (argk) -- (1.7,-3.5) -- node[ngeqedgenode,pos=0]{} (n3);
    \end{tikzpicture}
\end{lrbox}

\newsavebox\lpoLessBase

\begin{lrbox}{\lpoLessBase}
    \begin{tikzpicture}[canvassed]
    \node[node,circle] {$n_3$};
    \end{tikzpicture}
\end{lrbox}

\newsavebox\lpoLess

\begin{lrbox}{\lpoLess}
    \begin{tikzpicture}[canvassed]
    \node[termnode] (arg1) at (3,0) {$s_1\comp g(\bar{t})$};
    \node[] (argi) at (3,-1.5) {$\ldots$};
    \node[termnode] (argk) at (3,-3) {$s_k\comp g(\bar{t})$};
    \node[genericnode] (n1) at (0,-4.5) {$n_1$};
    \node[genericnode] (n2) at (1.5,-4.5) {$n_2$};
    \node[genericnode] (n3) at (3,-4.5) {$n_3$};
    
    \draw[->] (arg1) -- (1,-.7) -- node[gtedgenode,pos=0]{} (n1);
    \draw[->] (arg1) -- (1.5,-.8) -- node[eqedgenode,pos=0]{} (n1);
    \draw[->] (arg1) -- node[midway,ngeqedgenode]{} (argi);
    
    \draw[->] (argi) -- (1.5,-2.1) -- node[gtedgenode,pos=0]{} (n1);
    \draw[->] (argi) -- (2,-2.2) -- node[eqedgenode,pos=0]{} (n1);
    \draw[->] (argi) -- node[midway,ngeqedgenode]{} (argk);
    
    \draw[->] (argk) -- (1.3,-3.5) -- node[gtedgenode,pos=0]{} (n1);
    \draw[->] (argk) -- (1.8,-3.6) -- node[eqedgenode,pos=0]{} (n1);
    \draw[->] (argk) -- node[midway,ngeqedgenode]{} (n3);
    \end{tikzpicture}
\end{lrbox}

\newsavebox\lpoEqualBase

\begin{lrbox}{\lpoEqualBase}
    \begin{tikzpicture}[canvassed]
    \node[node,circle] {$n_2$};
    \end{tikzpicture}
\end{lrbox}

\newsavebox\lpoEqual

\begin{lrbox}{\lpoEqual}
    \begin{tikzpicture}[canvassed]
    \node[termnode] (arg1) at (0,0) {$s_1\comp t_1$};
    \node[termnode] (allgt2) at (-2,-1.5) {$f(\bar{s})\comp t_2$};
    \node[termnode] (arg2) at (0,-1.5) {$s_2\comp t_2$};
    \node[termnode] (alllt2) at (2,-1.5) {$s_2\comp g(\bar{t})$};
    \node[] (allgti) at (-2,-3) {$\ldots$};
    \node[] (argi) at (0,-3) {$\ldots$};
    \node[] (alllti) at (2,-3) {$\ldots$};
    \node[termnode] (allgtk) at (-2,-4.6) {$f(\bar{s})\comp t_k$};
    \node[termnode] (argk) at (0,-4.6) {$s_k\comp t_k$};
    \node[termnode] (allltk) at (2,-4.6) {$s_k\comp g(\bar{t})$};
    \node[genericnode] (n1) at (4,-6.5) {$n_1$};
    \node[genericnode] (n2) at (0,-6.5) {$n_2$};
    \node[genericnode] (n3) at (-4,-6.5) {$n_3$};
    
    
    \draw[->] (arg1) -- node[midway,gtedgenode]{} (allgt2);
    \draw[->] (arg1) -- node[midway,eqedgenode]{} (arg2);
    \draw[->] (arg1) -- node[midway,ngeqedgenode]{} (alllt2);
    
    
    \draw[->] (allgt2) -- node[midway,gtedgenode]{} (allgti);
    \draw[->] (allgt2) -- (-3.8,-2.4) -- node[eqedgenode,pos=0]{} (n3);
    \draw[->] (allgt2) -- (-3.3,-2.5) -- node[ngeqedgenode,pos=0]{} (n3);
    
    \draw[->] (arg2) -- node[midway,gtedgenode]{} (allgti);
    \draw[->] (arg2) -- node[midway,eqedgenode]{} (argi);
    \draw[->] (arg2) -- node[midway,ngeqedgenode]{} (alllti);
    
    \draw[->] (alllt2) -- (3.3,-2.5) -- node[gtedgenode,pos=0]{} (n1);
    \draw[->] (alllt2) -- (3.8,-2.4) -- node[eqedgenode,pos=0]{} (n1);
    \draw[->] (alllt2) -- node[midway,ngeqedgenode]{} (alllti);
    
    
    \draw[->] (allgti) -- node[midway,gtedgenode]{} (allgtk);
    \draw[->] (allgti) -- (-3.2,-3.6) -- node[eqedgenode,pos=0]{} (n3);
    \draw[->] (allgti) -- (-2.7,-3.7) -- node[ngeqedgenode,pos=0]{} (n3);
    
    \draw[->] (argi) -- node[midway,gtedgenode]{} (allgtk);
    \draw[->] (argi) -- node[midway,eqedgenode]{} (argk);
    \draw[->] (argi) -- node[midway,ngeqedgenode]{} (allltk);
    
    \draw[->] (alllti) -- (2.7,-3.7) -- node[gtedgenode,pos=0]{} (n1);
    \draw[->] (alllti) -- (3.2,-3.6) -- node[eqedgenode,pos=0]{} (n1);
    \draw[->] (alllti) -- node[midway,ngeqedgenode]{} (allltk);
    
    
    \draw[->] (allgtk) -- node[midway,gtedgenode]{} (n1);
    \draw[->] (allgtk) -- (-3.3,-5.5) -- node[eqedgenode,pos=0]{} (n3);
    \draw[->] (allgtk) -- (-2.8,-5.5) -- node[ngeqedgenode,pos=0]{} (n3);
    
    \draw[->] (argk) -- node[midway,gtedgenode]{} (n1);
    \draw[->] (argk) -- node[midway,eqedgenode]{} (n2);
    \draw[->] (argk) -- node[midway,ngeqedgenode]{} (n3);
    
    \draw[->] (allltk) -- (2.8,-5.5) -- node[gtedgenode,pos=0]{} (n1);
    \draw[->] (allltk) -- (3.3,-5.5) -- node[eqedgenode,pos=0]{} (n1);
    \draw[->] (allltk) -- node[midway,ngeqedgenode]{} (n3);
    \end{tikzpicture}
\end{lrbox}

\section{\tod{} Transformations}
\label{sec:transformations}
%
%
We explain our \tod{} transformations via a sequence of subgraph replacements and show how \tod{} equivalences are preserved. Replacements are performed while visiting specific nodes $n$ in the \tod, depending on the node label. As such, some replacements are \emph{specific to KBO/LPO} constraints (Figures~\ref{fig:kbo-transformations}--\ref{fig:lpo-transformations}), whereas others are \emph{generic} (Figure~\ref{fig:universal-transformations}), independent of the ordering.  
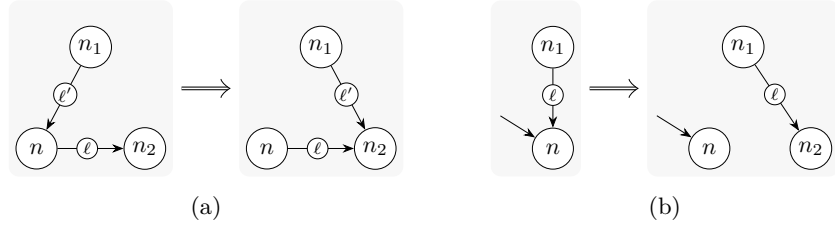
\begin{figure}[t]
    \centering
    \begin{minipage}{.5\linewidth}
    \begin{center}
        \begin{tikzpicture}[node distance=2em]
        \node (one) {\scalebox{0.9}{\usebox\redundantNodeOne}};
        \node[right=of one] (two) {\scalebox{0.9}{\usebox\redundantNodeTwo}};
        \draw[transformedge] (one) -- (two);
        \end{tikzpicture}\\
        (a)
    \end{center}
    \end{minipage}\begin{minipage}{.5\linewidth}
    \begin{center}
        \begin{tikzpicture}[node distance=2em]
        \node (one) {\scalebox{0.9}{\usebox\nodeReplicationOne}};
        \node[right=of one] (two) {\scalebox{0.9}{\usebox\nodeReplicationTwo}};
        \draw[transformedge] (one) -- (two);
        \end{tikzpicture}\\
        (b)
    \end{center}
    \end{minipage}
    \caption{Generic transformations on \tods: (a)  \emph{redundant node removal}, where node $n$ forces label $\ell$; and (b)  \emph{node replication}, where 
    $n_2$ is an exact copy of a non-exit node $n$ that has multiple incoming edges. }
    \label{fig:universal-transformations}
\end{figure}
Figure~\ref{fig:universal-transformations} summarizes our generic \tod{} transformations, as detailed below: 
\begin{itemize}[leftmargin=1.1em]
\item Figure~\ref{fig:universal-transformations}(a) corresponds to a \emph{redundant node removal}, where $n$ forces $\ell$. Simply removing nodes might violate \tod{} properties: as $n$ may have no incoming edges, the existence of a single root node in a \tod{} is not fulfilled. In such cases, we also remove $n$ from the \tod{}. Generally, if a redundant node removal  introduces a node with no incoming edges, we repeatedly remove such nodes. 
\item  Figure~\ref{fig:universal-transformations}(b) shows a \emph{node replication}, where $n$ is a non-exit node with multiple incoming edges and $n_2$ is the exact copy of $n$ with the same outgoing edges. \review{We use node replication to make nodes reachable via one path, simplifying reasoning for redundant node removal (\Cref{sec:retrieval}).}
\end{itemize}
Our \emph{KBO and LPO transformations} on \tods{} are shown Figures~\ref{fig:kbo-transformations} and~\ref{fig:lpo-transformations}, respectively. Here, we denote by $s$ a term $f(s_1,\ldots,s_k)$, also written as $f(\bar{s}$), with $k\geq 0$. Similarly,  
$t$ is a term $g(t_1,\ldots,t_m)$, also written as $g(\bar{t})$, with $m \geq 0$.
Note that in Figure~\ref{fig:lpo-transformations}, if $f\gg g$ and $m = 0$, the node $n_1$ replaces the top node; conversely, in case of $g\gg f$ and $k=0$, the node $n_3$ replaces the top node; finally, if $f=g$ and $k=0$, then the node $n_2$ replaces the top node.
%

\newsavebox\kboTransformations

\begin{lrbox}{\kboTransformations}
    \begin{tikzpicture}[node distance=3.5em and 6em]
    \node (center) {};
    \node[above=of center] (main) {\scalebox{0.8}{\usebox\maincase}};
    \node[left=of center] (case1) {\scalebox{0.8}{\usebox\kboGreater}};
    \node[right=of center] (case2) {\scalebox{0.8}{\usebox\kboLess}};
    \node[yshift=-5.5em] (case3) at (center) {\scalebox{0.8}{\usebox\kboEqual}};

    \draw[transformedge] (main) -| node[near end,fill=white,inner sep=1pt]{\scriptsize\begin{tabular}{c}Case 1: $f\gg g$\end{tabular}} (case1);
    \draw[transformedge] (main) -| node[near end,fill=white,inner sep=1pt]{\scriptsize\begin{tabular}{c}Case 3: $g\gg f$\end{tabular}} (case2);
    \draw[transformedge] (main) -- node[fill=white,inner sep=1pt]{\scriptsize\begin{tabular}{c}Case 2: $f=g$\end{tabular}} (case3);
    \end{tikzpicture}
\end{lrbox}

\newsavebox\lpoTransformations

\begin{lrbox}{\lpoTransformations}
    \begin{tikzpicture}[node distance=4.5em and 6.5em]
    \node (center) {};
    \node[above=of center] (main) {\scalebox{0.8}{\usebox\maincase}};
    \node[below=of center,yshift=6em] (case3) {\scalebox{0.8}{\usebox\lpoEqual}};
    \node[left=of center,yshift=2.5em,draw=white,line width=6pt,inner sep=0pt,rounded corners=1ex] (case1) {\scalebox{0.75}{\usebox\lpoGreater}};
    \node[right=of center,yshift=2.5em,draw=white,line width=6pt,inner sep=0pt,rounded corners=1ex] (case2) {\scalebox{0.75}{\usebox\lpoLess}};

    \draw[transformedge] ($(main.west) + (0,.8)$) -| node[near start,fill=white,inner sep=2pt]{\scriptsize\begin{tabular}{c}Case 1: $f\gg g$\end{tabular}} ($(case1.north) - (.5,0)$);
    \draw[transformedge] ($(main.east) + (0,.8)$) -| node[near start,fill=white,inner sep=2pt]{\scriptsize\begin{tabular}{c}Case 3: $g\gg f$\end{tabular}} ($(case2.north) + (.5,0)$);
    \draw[transformedge] (main) -- node[fill=white,inner sep=2pt]{\scriptsize\begin{tabular}{c}Case 2: $f=g$\end{tabular}} (case3);
    \end{tikzpicture}
\end{lrbox}

\begin{figure}[t]
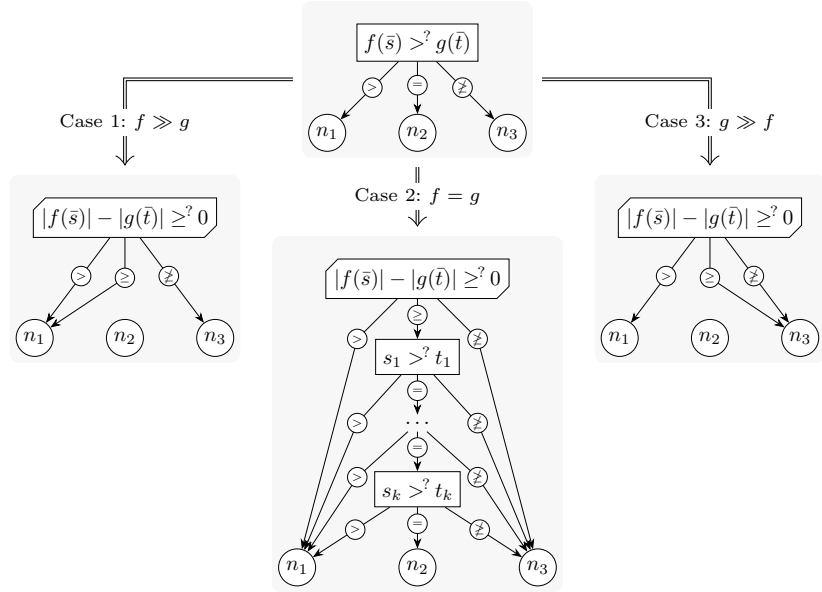

    \centering
    \usebox\kboTransformations
    \caption{KBO transformations on \tods.}
    \label{fig:kbo-transformations}
\end{figure}
\begin{figure}[t]
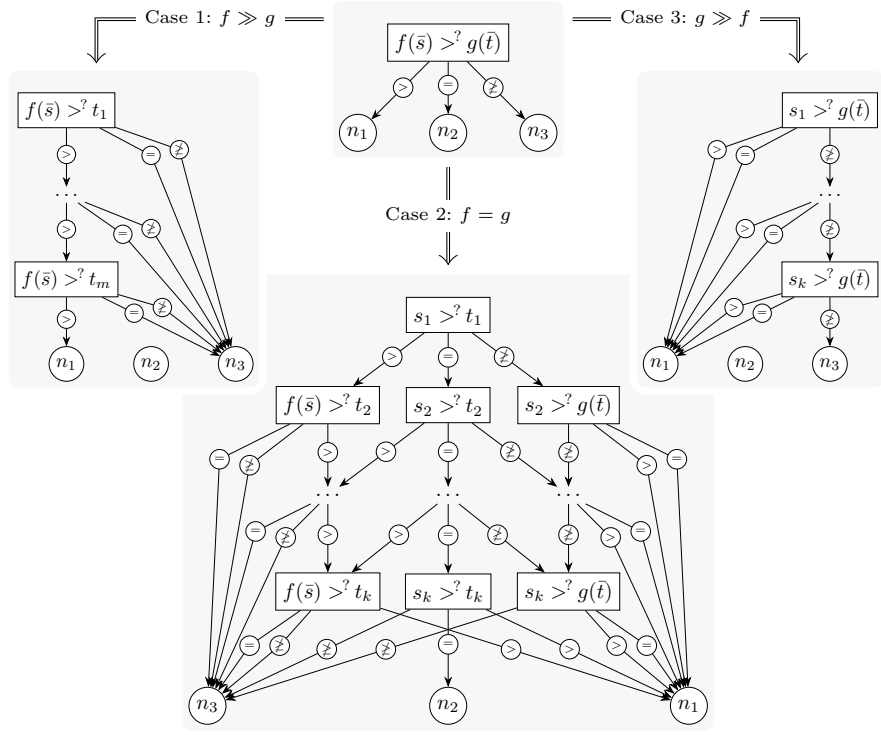

    \centering
    \usebox\lpoTransformations
    \caption{LPO transformations on \tods.}
    \label{fig:lpo-transformations}
\end{figure}


\begin{restatable}[Correctness of \tod{} Transformations]{theorem}{correctness}
(1) Every sequence of \tod{} transformations terminates. (2) Every \tod{} transformation of Figures~\ref{fig:universal-transformations}--\ref{fig:lpo-transformations} preserves \tod{} equivalence.\label{thm:tranformations}
\end{restatable}
\begin{proof}
(1) For proving 
\emph{termination}, we  introduce a well-founded order on TODs and show that every transformation replaces a TOD by a smaller one. In the proof, we use (well-founded) finite multiset extensions of (well-founded) orders.

We first introduce a mapping $\mu$ from nodes to finite multisets of terms as:
\begin{enumerate}
    \item[(i)] If $n$ is a term comparison node $s\comp t$, then $\mu(n)$ is the multiset 
    $\setof{s,t}$.
    \item[(ii)] For any other node $n$, we define $\mu(n)$ to be the empty multiset.
\end{enumerate}
Let us also define an order $>_\mu$ on nodes as follows: $n_2 >_\mu n_1$ if $\mu(n_2)$ is greater than $\mu(n_1)$ in the multiset extension of the order $\succ$ on terms. Note that $>_\mu$ is well-founded, since $\mu$ embeds it in the multiset extension of a well-founded order.

For every path $\pi$ in a \tod, we denote by $\mu(\pi)$ the multiset consisting of elements $\mu(n)$ for all nodes $n$ in $\pi$. We define an ordering, also denoted by $>_\mu$, on paths as follows: $\pi_2 >_\mu \pi_1$ if $\mu(\pi_2)$ is greater than $\mu(\pi_1)$ in the multiset extension of the order $>_\mu$ on nodes. Again,  $>_\mu$ is well-founded, since it can be embedded to the multiset extension of a well-founded order.

Finally, for a \tod{} $\OD$, we denote by $\mu(\OD)$ the multiset consisting of all multisets $\mu(\pi)$, where $\pi$ is a path from the root node in $\OD$. We also define an ordering $>_\mu$ on \tods{} by letting $\OD_2 >_\mu \OD_1$ if $\mu(\OD_2) > \mu(\OD_1)$. Using the same arguments as before, we conclude that $>_\mu$ on TODs is well-founded, too.

For every transformation apart from node replication that changes a \tod{} $\OD$ to a 
\tod{} $\OD'$, we have that $\OD >_\mu \OD'$, which implies termination. The proof is by routine 
inspection of transformations. For example for the case $f = g$ of the LPO transformations 
(last transformation of Figure~\ref{fig:lpo-transformations}), we replace a path with a term comparison node $f(\bar{s}) \comp f(\bar{t})$ by a 
finite number of paths, so that every new term comparison node on these paths contains a 
comparison of a pair of terms strictly smaller in the multiset order than the multiset $\setof{f(\bar{s}),f(\bar{t})}$. 
Another example is the redundant node removal of Figure~\ref{fig:universal-transformations}. 
This transformation replaces on some paths nodes $n_1,n,n_2$ by $n_1,n_2$, which results in a smaller multiset.

Finally, we note that node replication results in a \tod{} containing exactly the same multiset of paths, but it can only be applied a finite number of times. 

\noindent{\emph{(2) \tod{} equivalence.}} 
The \tod{} transformations of Figure~\ref{fig:kbo-transformations}--\ref{fig:lpo-transformations} preserve \tod{} equivalence by properties of KBO/LPO. 
Redundant node removal in Figure~\ref{fig:universal-transformations}(a)  ensures equivalence by the ``forces'' relation (Definition~\ref{def:tod:traversal:foces}). Node replication in Figures~\ref{fig:universal-transformations}(b) preserves equivalence since the set of \tod{} paths does not change. \QED
\end{proof}

\section{TOD Retrieval and Maintenance}
\label{sec:retrieval}

TOD is a data structure intended to be an index~\cite{TermIndexing} for retrieval of equalities, which become ordered by the query substitution. In this section we describe the three main operations on the \tod{} index: \emph{insertion, deletion and retrieval}. An interesting feature of TODs is that the main modifications of them occur not during insertion but during retrieval. As mentioned in Section~\ref{sec:term_ordering_diagrams}, we modify the TOD $\OD$ lazily, when we traverse it for a query substitution $\sigma$, resulting in an equivalent TOD $\OD'$. The new TOD $\OD'$ is less expensive to traverse, since it either removes checks performed at nodes (redundant node removal), or replaces them by simpler ones (all transformations of Figures~\ref{fig:kbo-transformations}--\ref{fig:lpo-transformations}). Subsequent retrievals are then performed on the TOD $\OD'$ instead of $\OD$.

\begin{wrapfigure}[8]{r}[0pt]{0.45\textwidth}
\vspace{-3em}
\centering
\begin{tikzpicture}[node distance=2em]
\node (one) {\scalebox{0.9}{\usebox\insertionOne}};
\node[right=of one] (two) {\scalebox{0.9}{\usebox\insertionTwo}};
\draw[transformedge] (one) -- (two);
\end{tikzpicture}
\vspace{-1em}
\caption{Insertion of $l\eqs r$}
\label{fig:insertion}
\end{wrapfigure}

\paragraph{Index operations.} The \emph{insertion} of an equality $l \eqs r$ is simple: we insert a term comparison node just before the exit node, as shown in Figure~\ref{fig:insertion}.

The \emph{deletion} of an equality $l \eqs r$ is not performed. Instead of doing the deletion, we simply memorize that $l \eqs r$ is to be deleted. In practice, we do not even have to do this, since the clause containing $l \eqs r$ will be marked as deleted anyhow. 

A \emph{retrieval} from a TOD consists of traversal, possibly interleaved with TOD transformations (Section~\ref{sec:transformations}). We next describe our \tod{} {retrieval} algorithm. 

\paragraph{\tod{} retrieval algorithm.}
We  introduce a notion meant to capture sufficient conditions for the (expensive) ``forced" relation
\cameraready{\footnote{We discuss forcing functions implemented in \vampire{} in the extended version of this paper~\cite{FullPaper}.}}
\fullversion{\footnote{We discuss forcing functions implemented in \vampire{} in the Appendix.}}
, in extension of Definition~\ref{def:tod:traversal:foces}. 

\begin{Definition}[Forcing function\label{def:forcing:function}]
Let 
$\sigma$ be a substitution. We call a path $\pi$ a \emph{$\sigma$-path}, if for every edge from a node $n$ to a node $n'$ in it labeled by $\ell$, (1) if $n$ is a term comparison node $s \comp t$, then $\sigma(s \comp t) = \ell$, and (2) if $n$ is a positivity check node $e \pos 0$, then $\sigma(e \pos 0) = \ell$.
We call a partial function $F$ from paths to edge labels a \emph{forcing function} if it has the following property. For every substitution $\sigma$ and every $\sigma$-path $\pi = n_0,\ldots,n_i,n_{i+1}$ from the root, if $F(n_0,\ldots,n_i) = \ell$, then the edge in $\pi$ from $n_i$ to $n_{i+1}$ is labeled with $\ell$.\QED
\end{Definition}

Our retrieval algorithm is parametrized by a forcing function $F$. 
During the retrieval, we mark some nodes \emph{visited}. Once a node $n$ is marked visited, the path from the root to $n$ will never change.  
Initially, all nodes in a TOD are unvisited. Let $\OD$ be a TOD and $\sigma$ a substitution. The traversal starts in the successor of the root node and repeatedly applies steps shown in Figure~\ref{fig:retrieval-steps}. When we describe the steps, we assume that during the traversal we already followed a path $\pi = n_0,n_1,\ldots, n_k,n$ leading to the current node $n$.

\begin{figure}[t]

\noindent\hrulefill

\begin{enumerate}[leftmargin=1.3em]
    \item If $n$ is the exit node, we terminate and return the success set of $\sigma$ for $\OD$ (the set of all non-deleted equalities in success nodes visited during the traversal).\label{case:exit}

    \item If $n$ is a success node, we set $n$ to the successor node of $n$.
    Alternatively, if we are interested in only one candidate, we can terminate once we reached the first success node with a non-deleted equality.\label{case:success}

    \item Let $n$ be an unvisited evaluation node. If $n$ has more than one incoming edge, we first apply node replication so that the only incoming edge to $n$ is from $n_k$. Then, \label{case:visited}

    \begin{enumerate}
        \item If $F(\pi) = \ell$, we apply the transformation of Figure~\ref{fig:universal-transformations} 
        (a).\label{case:redundant-node}
        
       \item If any of the transformations of Figures~\ref{fig:kbo-transformations} and \ref{fig:lpo-transformations} applies to $n$, we apply this transformation.\label{case:non-visited}

       \item Otherwise, we mark $n$ visited.\label{case:mark-visited}
    \end{enumerate}

    \item If $n$ is a visited term comparison node containing $s \comp t$, we follow the edge $\sigma(s \comp t)$ from $n$. \label{case:follow-term-node}
    
    \item If $n$ is a visited positivity check node containing $e \pos 0$, we follow the edge $\sigma(e \pos 0)$ from $n$.\label{case:follow-positivity}

\end{enumerate}

    \noindent\hrulefill

    \caption{Retrieval algorithm steps.}
    \label{fig:retrieval-steps}
\end{figure}

Let us now establish some properties of the algorithm,  proving its correctness. 

\begin{lemma}\label{lem:visited}%
    If $n$ is a visited node, there is a single path from the root to $n$.
\end{lemma}
\begin{proof}
    When we make a node visited at step~\ref{case:mark-visited} of Figure~\ref{fig:retrieval-steps}, by induction we can assume that $n_k$ already has this property. Since after this step there is only one edge to $n$, and this edge is from $n_k$, 
    $n$  satisfies this property.
    None of the other steps adds an incoming edge to a visited node or changes the content of a visited node, so $n_0,\ldots,n_k,n$ remains the only path to $n$. \QED
\end{proof}

\begin{lemma}\label{lem:forcing:function}
    Step \ref{case:redundant-node} of Figure~\ref{fig:retrieval-steps} transforms the TOD into an equivalent one. 
\end{lemma}
\begin{proof}
    By Lemma~\ref{lem:visited}, there is only one path from the root to $n$. By Definition~\ref{def:forcing:function}, it follows that $n$ forces $\ell$, so this step is a special case of redundant mode removal from Figure~\ref{fig:universal-transformations}(a), which preserves \tod{} equivalence by Theorem~\ref{thm:tranformations}. \QED
\end{proof}

\begin{lemma}[Termination]
    The retrieval algorithm terminates.
\end{lemma}
\begin{proof}
    Straightforward by Lemma~\ref{lem:forcing:function} and Theorem~\ref{thm:tranformations}. Indeed, all transformations made during the retrieval are special cases of TOD transformations, so we can only make a finite number of them. All other steps of the algorithm either change the current node to its successor in the TOD, or mark an unvisited node as visited. Since any TOD is a dag, we can only make a finite number of such steps. \QED
\end{proof}

\begin{lemma}[Correctness]
    The TOD $\OD'$ resulting from the retrieval is equivalent to the TOD $\OD$ before the retrieval.
\end{lemma}
\begin{proof}
    By Lemma~\ref{lem:forcing:function}, step~\ref{case:redundant-node} of Figure~\ref{fig:retrieval-steps} is a special case of redundant node removal.  All other transformations in Figure~\ref{fig:retrieval-steps} are special cases of TOD transformations, which preserve equivalence by Theorem~\ref{thm:tranformations}. \QED
\end{proof}

\newsavebox\retrievalExampleOne

\begin{lrbox}{\retrievalExampleOne}
    \begin{tikzpicture}[canvassed]
    \coordinate (sourceC) at (0,.6);
    \coordinate (c1C) at (0,0);

    \draw[traceedge,line cap=round] (sourceC) edge (c1C);
    \node[currtermnode] at (c1C) {$f(x,y)\comp f(y,x)$};

    \node[termnode,unprocessed] (c1) at (c1C) {$f(x,y)\comp f(y,x)$};
    \node[datanode,unprocessed] (s1) at (0,-1.4) {$f(x,y)\eqs f(y,x)$};

    \draw[->] (sourceC) -- (c1);
    \draw[->] (c1) -- node[midway,gtedgenode]{} (s1);
    \end{tikzpicture}
\end{lrbox}

\newsavebox\retrievalExampleTwo

\begin{lrbox}{\retrievalExampleTwo}
    \begin{tikzpicture}[canvassed]
    \coordinate (sourceC) at (0,.6);
    \coordinate (c1C) at (0,0);

    \draw[traceedge,line cap=round] (sourceC) edge (c1C);
    \node[currpolynode] at (c1C) {$0\pos 0$};

    \node[polynode,unprocessed] (c1) at (c1C) {$0\pos 0$};
    \node[termnode,unprocessed] (c2) at (2,0) {$x\comp y$};
    \node[termnode,unprocessed] (c3) at (4,0) {$y\comp x$};
    \node[datanode,unprocessed] (s1) at (2,-1.4) {$f(x,y)\eqs f(y,x)$};

    \draw[->] (sourceC) -- (c1);
    \draw[->] (c1) -- node[gtedgenode]{} (s1);
    \draw[->] (c1) -- node[geqedgenode]{} (c2);

    \draw[->] (c2) -- node[gtedgenode]{} (s1);
    \draw[->] (c2) -- node[eqedgenode]{} (c3);
    
    \draw[->] (c3) -- node[gtedgenode]{} (s1);
    
    \end{tikzpicture}
\end{lrbox}

\newsavebox\retrievalExampleThree

\begin{lrbox}{\retrievalExampleThree}
    \begin{tikzpicture}[canvassed]
    \coordinate (sourceC) at (0,.6);
    \coordinate (c1C) at (0,0);
    
    \draw[traceedge,line cap=round] (sourceC) edge (c1C);
    \node[currtermnode] at (c1C) {$x\comp y$};

    \node[termnode,unprocessed] (c1) at (c1C) {$x\comp y$};
    \node[termnode,unprocessed] (c2) at (2,0) {$y\comp x$};
    \node[datanode,unprocessed] (s1) at (1,-1.2) {$f(x,y)\eqs f(y,x)$};

    \draw[->] (sourceC) -- (c1);

    \draw[->] (c1) -- node[gtedgenode]{} (s1);
    \draw[->] (c1) -- node[eqedgenode]{} (c2);

    \draw[->] (c2) -- node[gtedgenode]{} (s1);
    
    \end{tikzpicture}
\end{lrbox}

\newsavebox\retrievalExampleThreePrime

\begin{lrbox}{\retrievalExampleThreePrime}
    \begin{tikzpicture}[canvassed]
    \coordinate (sourceC) at (0,.6);
    \coordinate (c1C) at (0,0);
    
    \draw[traceedge,line cap=round] (sourceC) edge (c1C);
    \node[currtermnode] at (c1C) {$x\comp y$};

    \node[termnode,processed] (c1) at (c1C) {\contour{white}{$x\comp y$}};
    \node[termnode,unprocessed] (c2) at (2,0) {$y\comp x$};
    \node[datanode,unprocessed] (s1) at (1,-1.2) {$f(x,y)\eqs f(y,x)$};

    \draw[->] (sourceC) -- (c1);

    \draw[->] (c1) -- node[gtedgenode]{} (s1);
    \draw[->] (c1) -- node[eqedgenode]{} (c2);

    \draw[->] (c2) -- node[gtedgenode]{} (s1);
    
    \end{tikzpicture}
\end{lrbox}

\newsavebox\retrievalExampleFour

\begin{lrbox}{\retrievalExampleFour}
    \begin{tikzpicture}[canvassed]
    \coordinate (sourceC) at (0,.6);
    \coordinate (c1C) at (0,0);
    \coordinate (c2C) at (2,0);

    \draw[traceedge,line cap=round] (sourceC) edge (c1C);
    \node[currtermnode] at (c1C) {$x\comp y$};
    \draw[traceedge,line cap=round] (c1C) edge (c2C);
    \node[currtermnode] at (c2C) {$y\comp x$};

    \node[termnode,processed] (c1) at (c1C) {\contour{white}{$x\comp y$}};
    \node[termnode,unprocessed] (c2) at (c2C) {$y\comp x$};
    \node[datanode,unprocessed] (s1) at (1,-1.2) {$f(x,y)\eqs f(y,x)$};

    \draw[->] (sourceC) -- (c1);

    \draw[->] (c1) -- node[gtedgenode]{} (s1);
    \draw[->] (c1) -- node[eqedgenode]{} (c2);

    \draw[->] (c2) -- node[gtedgenode]{} (s1);
    \end{tikzpicture}
\end{lrbox}

\newsavebox\retrievalExampleFive

\begin{lrbox}{\retrievalExampleFive}
    \begin{tikzpicture}[canvassed]
    \coordinate (sourceC) at (0,.6);
    \coordinate (c1C) at (0,0);
    \coordinate (exitC) at (2.3,0);

    \draw[traceedge,line cap=round] (sourceC) edge (c1C);
    \node[currtermnode] at (c1C) {$x\comp y$};
    \draw[traceedge,line cap=round] (c1C) edge (exitC);
    \node[currexitnode] at (exitC) {};

    \node[termnode,processed] (c1) at (c1C) {\contour{white}{$x\comp y$}};
    \node[sinknode] (exit) at (exitC) {};
    \node[datanode,unprocessed] (s1) at (1,-1.2) {$f(x,y)\eqs f(y,x)$};

    \draw[->] (sourceC) -- (c1);

    \draw[->] (c1) -- node[eqedgenode]{} (exit);
    \draw[->] (c1) -- node[gtedgenode]{} (s1);
    \end{tikzpicture}
\end{lrbox}

\newsavebox\retrievalExampleSix

\begin{lrbox}{\retrievalExampleSix}
    \begin{tikzpicture}[canvassed]
    \coordinate (sourceC) at (0,.6);
    \coordinate (c1C) at (0,0);
    \coordinate (s1C) at (.6,-1.3);
    \coordinate (c2C) at (3.3,0);

    \draw[traceedge,black!5,line cap=round] (sourceC) edge (c1C);
    \node[currtermnode,black!5] at (c1C) {$x\comp y$};
    \draw[traceedge,black!5,line cap=round] (c1C) edge (s1C);
    \node[currsuccessnode,black!5] at (s1C) {$f(x,y)\eqs f(y,x)$};
    \draw[traceedge,black!5,line cap=round] (s1C) edge (c2C);
    \node[currtermnode,black!5] at (c2C) {$f(x,y)\comp f(x,x)$};

    \node[termnode,processed] (c1) at (c1C) {\contour{white}{$x\comp y$}};
    \node[datanode,unprocessed] (s1) at (s1C) {$f(x,y)\eqs f(y,x)$};
    \node[termnode,unprocessed] (c2) at (c2C) {$f(x,y)\comp f(x,x)$};
    \node[datanode,unprocessed] (s2) at (3.5,-1.3) {$f(x,y)\eqs f(x,x)$};

    \draw[->] (sourceC) -- (c1);
    \draw[->] (c1) -- node[midway,gtedgenode]{} (s1);
    \draw[->] (c1) edge[bend left=10] node[midway,eqedgenode]{} (c2);
    \draw[->] (c1) edge[bend right=10] node[midway,ngeqedgenode]{} (c2);
    \draw[->] (s1) -- (c2);
    \draw[->] (c2) -- node[midway,gtedgenode]{} (s2);
    
    \end{tikzpicture}
\end{lrbox}

\newsavebox\retrievalExampleSeven

\begin{lrbox}{\retrievalExampleSeven}
    \begin{tikzpicture}[canvassed]
    \coordinate (sourceC) at (0,.6);
    \coordinate (c1C) at (0,0);
    \coordinate (s1C) at (.6,-1.3);
    \coordinate (c2C) at (3.3,0);

    \draw[traceedge,line cap=round] (sourceC) edge (c1C);
    \node[currtermnode] at (c1C) {$x\comp y$};
    \draw[traceedge,line cap=round] (c1C) edge (s1C);
    \node[currsuccessnode] at (s1C) {$f(x,y)\eqs f(y,x)$};
    \draw[traceedge,line cap=round] (s1C) edge (c2C);
    \node[currtermnode] at (c2C) {$f(x,y)\comp f(x,x)$};

    \node[termnode,processed] (c1) at (c1C) {\contour{white}{$x\comp y$}};
    \node[datanode,processed] (s1) at (s1C) {\contour{white}{$f(x,y)\eqs f(y,x)$}};
    \node[termnode,unprocessed] (c2) at (c2C) {$f(x,y)\comp f(x,x)$};
    \node[datanode,unprocessed] (s2) at (3.5,-1.3) {$f(x,y)\eqs f(x,x)$};

    \draw[->] (sourceC) -- (c1);
    \draw[->] (c1) -- node[midway,gtedgenode]{} (s1);
    \draw[->] (c1) edge[bend left=10] node[midway,eqedgenode]{} (c2);
    \draw[->] (c1) edge[bend right=10] node[midway,ngeqedgenode]{} (c2);
    \draw[->] (s1) -- (c2);
    \draw[->] (c2) -- node[midway,gtedgenode]{} (s2);
    
    \end{tikzpicture}
\end{lrbox}

\newsavebox\retrievalExampleEight

\begin{lrbox}{\retrievalExampleEight}
    \begin{tikzpicture}[canvassed]
    \coordinate (sourceC) at (0,.6);
    \coordinate (c1C) at (0,0);
    \coordinate (s1C) at (0,-1.4);
    \coordinate (c2C) at (3,0);

    \draw[traceedge,line cap=round] (sourceC) edge (c1C);
    \node[currtermnode] at (c1C) {$x\comp y$};
    \draw[traceedge,line cap=round] (c1C) edge (s1C);
    \node[currsuccessnode] at (s1C) {$f(x,y)\eqs f(y,x)$};
    \draw[traceedge,line cap=round] (s1C) edge (c2C);
    \node[currpolynode] at (c2C) {$y-x\pos 0$};

    \node[termnode,processed] (c1) at (c1C) {\contour{white}{$x\comp y$}};
    \node[datanode,processed] (s1) at (s1C) {\contour{white}{$f(x,y)\eqs f(y,x)$}};
    \node[polynode,processed] (c2) at (c2C) {\contour{white}{$y-x\pos 0$}};
    \node[termnode,unprocessed] (c3) at (6,0) {$x\comp x$};
    \node[termnode,unprocessed] (c4) at (6,-1.4) {$y\comp x$};
    \node[datanode,unprocessed] (s2) at (3,-1.4) {$f(x,y)\eqs f(x,x)$};

    \draw[->] (sourceC) -- (c1);
    \draw[->] (c1) -- node[midway,gtedgenode]{} (s1);
    \draw[->] (c1) edge[bend left=10] node[midway,eqedgenode]{} (c2);
    \draw[->] (c1) edge[bend right=10] node[midway,ngeqedgenode]{} (c2);
    \draw[->] (s1) -- (c2);
    \draw[->] (c2) -- node[midway,gtedgenode]{} (s2);
    \draw[->] (c2) -- node[midway,geqedgenode]{} (c3);
    \draw[->] (c3) -- node[midway,gtedgenode]{} (s2);
    \draw[->] (c3) -- node[midway,eqedgenode]{} (c4);
    \draw[->] (c4) -- node[midway,gtedgenode]{} (s2);
    
    \end{tikzpicture}
\end{lrbox}
%
\makeatletter
\newcommand{\subdiag}[2]{%
   \protected@write \@auxout {}{\string \newlabel {#1}{{#2}{\thepage}{#2}{#1}{}} }%
   \hypertarget{#1}{#2}
}
\makeatother
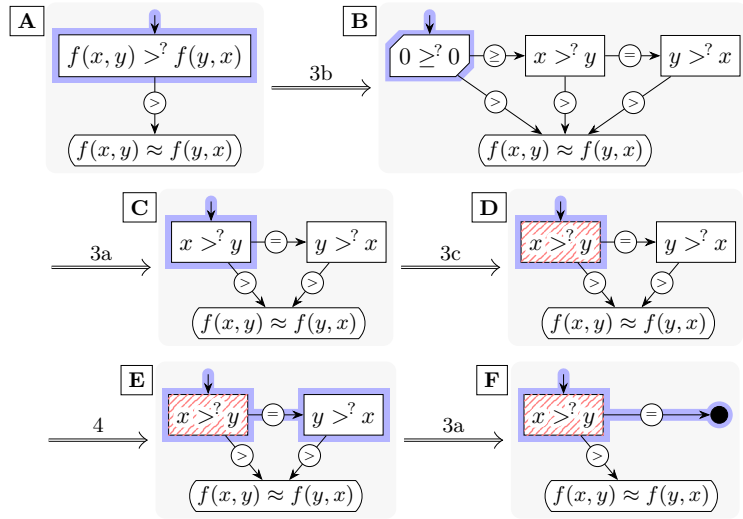
\begin{figure}[t]
    \centering
    \scalebox{.9}{
    \begin{tikzpicture}[node distance=0em and 4.5em]
    \node[] (n1) {\scalebox{1}{\usebox\retrievalExampleOne}};
    \node[right=of n1] (n2) {\scalebox{1}{\usebox\retrievalExampleTwo}};
    \node[below=of n1,xshift=5.5em] (n3) {\scalebox{1}{\usebox\retrievalExampleThree}};
    \node[draw=none,left=of n3] (n3p) {};
    \node[right=of n3] (n4) {\scalebox{1}{\usebox\retrievalExampleThreePrime}};
    \node[below=of n3] (n5) {\scalebox{1}{\usebox\retrievalExampleFour}};
    \node[draw=none,left=of n5] (n5p) {};
    \node[right=of n5] (n6) {\scalebox{1}{\usebox\retrievalExampleFive}};

    \node[labelnode] at (n1.north west) {\subdiag{diagA}{\bfseries A}};
    \node[labelnode] at (n2.north west) {\subdiag{diagB}{\bfseries B}};
    \node[labelnode] at (n3.north west) {\subdiag{diagC}{\bfseries C}};
    \node[labelnode] at (n4.north west) {\subdiag{diagD}{\bfseries D}};
    \node[labelnode] at (n5.north west) {\subdiag{diagE}{\bfseries E}};
    \node[labelnode] at (n6.north west) {\subdiag{diagF}{\bfseries F}};

    \draw[transformedge] (n1) -- node[above]{\small 3b} (n2);
    \draw[transformedge] (n3p) -- node[above]{\small 3a} (n3);
    \draw[transformedge] (n3) -- node[above]{\small 3c} (n4);
    \draw[transformedge] (n5p) -- node[above]{\small 4} (n5);
    \draw[transformedge] (n5) -- node[above]{\small 3a} (n6);
    \end{tikzpicture}}

    \caption{Retrieval with a query substitution $\sigma$ such that $x\sigma=y\sigma$.}
    \label{fig:retrieval-example1}
\end{figure}
\begin{example}
Consider the TOD in Figure~\ref{fig:example1}(a).
We perform retrieval on this TOD using a KBO with constant weight function  and a query substitution $\sigma$ such that $x\sigma=y\sigma$. The retrieval steps are shown in Figure~\ref{fig:retrieval-example1}. Initially, all nodes in the TOD are unvisited. We highlight the path up to the current node in each sub-diagram with blue, and denote visited nodes with a red striped background. Note that the last node on the blue path is the current node $n$ from the algorithm.

\begin{enumerate}[leftmargin=1.1em]
\item Starting with sub-diagram \ref{diagA} of Figure~\ref{fig:retrieval-example1}, we use step~\ref{case:non-visited}, applying case~2 of the KBO transformations (see Figure~\ref{fig:kbo-transformations}).
\item In sub-diagram \ref{diagB}, for the current path $\pi$ we have that $F(\pi)$ equals $\geq$, so we apply redundant node removal (step~\ref{case:redundant-node}).
\item In sub-diagram \ref{diagC}, the current node does not admit any transformations, so we mark it visited (step~\ref{case:mark-visited}).
\item In sub-diagram \ref{diagD}, $\sigma(x\comp y)$ is $=$, so we follow the edge labeled $=$ (step~\ref{case:follow-term-node}).
\item In sub-diagram \ref{diagE}, for the current path $\pi$ we get that $F(\pi)$ equals $=$, due to $\sigma(x\comp y)$ being $=$. We apply redundant node removal and get to the exit node denoted explicitly with a black disc (step~\ref{case:redundant-node}).
\item In sub-diagram \ref{diagF}, we return from the retrieval with no equalities (step~\ref{case:exit}).
\end{enumerate}

\begin{figure}[t]
    \centering
    \begin{tikzpicture}[node distance=0em and 2.4em]
    \node[] (n1) {\scalebox{.9}{\usebox\retrievalExampleSix}};
    \node[right=of n1] (n2) {\scalebox{.9}{\usebox\retrievalExampleSeven}};
    \node[below=of n1,xshift=11em] (n3) {\scalebox{.9}{\usebox\retrievalExampleEight}};
    \node[draw=none,left=of n3,xshift=-1em] (n3p) {};

    \node[labelnode] at (n1.north west) {\subdiag{diagG}{\bfseries G}};
    \node[labelnode] at (n2.north west) {\subdiag{diagH}{\bfseries H}};
    \node[labelnode] at (n3.north west) {\subdiag{diagI}{\bfseries I}};

    \draw[transformedge] (n1) -- node[above]{\scriptsize 4,3c,2} (n2);
    \draw[transformedge] (n3p) -- node[above]{\scriptsize  3b,3c} (n3);
    \end{tikzpicture}

    \caption{Result of insertion of $f(x,y)\eqs f(x,x)$ into the TOD of in sub-diagram~\ref{diagF} of Figure~\ref{fig:retrieval-example1}, and retrieval steps from top-left TOD with query substitution $\sigma$ s.t. $x\sigma\succ y\sigma$.}
    \label{fig:retrieval-example2}
\end{figure}
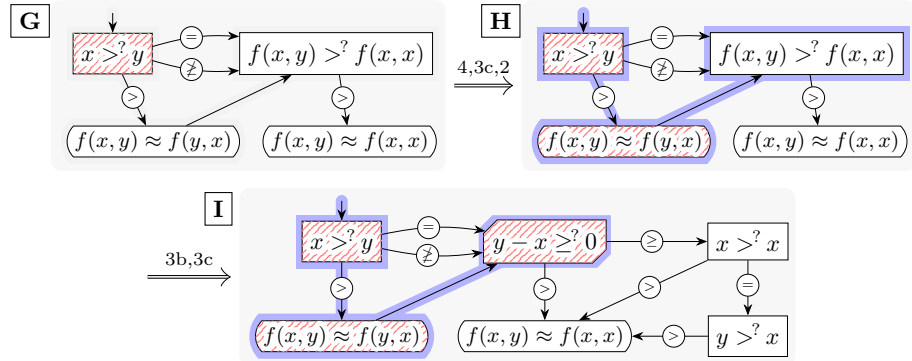

Suppose that, before further traversals, we insert the equality $f(x,y)\eqs f(x,x)$ into the TOD in sub-diagram~\ref{diagF} of Figure~\ref{fig:retrieval-example1}.
The resulting TOD is shown in sub-diagram~\ref{diagG} of Figure~\ref{fig:retrieval-example2}, with a new term comparison node labeled $f(x,y)\comp f(x,x)$ and a new success node labeled $f(x,y)\eqs f(x,x)$. The $>$ edge of the new term comparison node is connected to the new success node. The $=$ and $\ngeq$ edges of the new term comparison node, and the outgoing edge of the success node are connected to a new exit node (not shown).
%
%
We traverse the TOD in sub-diagram~\ref{diagG} with a query substitution $\sigma$ such that $x\sigma\succ y\sigma$.
\begin{enumerate}[leftmargin=1.1em]
\item Starting from sub-diagram~\ref{diagG} of Figure~\ref{fig:retrieval-example2}, we follow the edge $>$ (step~\ref{case:follow-term-node}), mark the success node labeled $f(x,y)\eqs f(y,x)$ visited (step~\ref{case:mark-visited}) and follow its outgoing edge (step~\ref{case:success}).
\item In sub-diagram~\ref{diagH}, we apply Case~2 of the KBO transformations (step~\ref{case:non-visited}) and mark the current node labeled $y-x\pos 0$ visited (step~\ref{case:mark-visited}).
\item In sub-diagram~\ref{diagI}, we follow the $\ngeq$ edge of the current node (step~\ref{case:follow-positivity}) and get to the exit node.
\item We exit with a single equality $f(x,y)\eqs f(y,x)$ (step~\ref{case:exit}).\QED
\end{enumerate}

\end{example}

\section{Evaluation}
\label{sec:evaluation}

We implemented our equality retrieval approach using \tod{} in  \vampire{}. 

\paragraph{Prover setup.}
We considered twelve configurations of \vampire{} options for reasoning using \tod{}, by 
taking Discount~\cite{Discount} or Otter~\cite{Otter} as 
saturation algorithm ({\tt -sa discount/otter}), KBO or LPO as term order ({\tt -to kbo/lpo}), and \tod{} variants using the new option {\tt -fdtod} with values {\tt off/on/shared}
\footnote{The option values {\tt on} and {\tt shared} are implemented in the {\tt master} and {\tt term-ordering-diagrams} branches of \vampire{}, respectively.}
for using \tods{} in forward demodulation, where:
%
%
\begin{enumerate}[label=(\arabic*)]
    \item \off{} does not use TODs, 
    \item \on{} uses a separate \tod{} for each forward demodulator (see e.g. \Cref{fig:example1}),
    \item \shared{} uses a shared \tod{} for each set of forward demodulators with the same left-hand side  (see e.g.~\Cref{fig:example2}).
\end{enumerate}

\paragraph{Benchmarks.} \review{We used the TPTP 8.2.0 library~\cite{TPTP} which contains 25474 problems. We ignored all problems where the TOD implementation is never called. Overall, we evaluated our work on 9310 TPTP problems.}

\paragraph{Hardware.} We used compute nodes  with AMD Epyc 7502 2.5GHz processors and 1TB RAM. Each benchmark run relied on a single core and 16GB of memory.
\begin{table}[t]
    \caption{Number of problems solved by \vampire{} within 60s. The numbers in parentheses show the number of problems lost and won by \tod{} variants compared to \texttt{off}.}
    \centering
    \setlength{\tabcolsep}{1em}
    \begin{tabular}{|c c| c | c | c |}
        \hline
                 &     & \off & \on & \shared \\
        \hline
        \multirow{2}{1cm}{Otter}
                 & KBO & 3250 (-0, +0) & 3258 (-0, +8) &	3261 (-0, +11)	\\
                 & LPO & 3064 (-0, +0) &	3083 (-1, +20) &	3092 (-2, +30)\\\hdashline
        \multirow{2}{1cm}{Discount}
                 & KBO & 3182 (-0, +0) &	3194 (-1, +13) &	3196 (-1, +15)\\
                 & LPO & 3016 (-0, +0) &	3044 (-1, +29) &	3058 (-1, +43)\\
        \hline
    \end{tabular}
    \label{tab:vampire-vanilla}
\end{table}

\paragraph{Experimental summary.} 
\Cref{tab:vampire-vanilla} shows the number of problems solved by \vampire{} with a 60 seconds timeout. Our approach using TODs is  better, regardless of the term ordering ({\tt to}) or the saturation algorithm ({\tt sa}). Further, \vampire{} \shared{} proves theorems significantly faster than \off{} (see \Cref{app:runtimes}).

Experiments using \tods{} seem to reach the memory limit in more instances than without \tods{}. For example in {\tt -sa otter -to lpo}, 44 problems reach a memory limit with \shared, whereas only 14 do so when \tods{} are \off. However, \vampire{} {\tt -fdtod off} hits some difficult ordering checks and is unable to generate huge clauses that consume a lot of memory 
(see \Cref{app:emprical-eval}).

\Cref{tab:instr-count}\todo{The titles of columns are a bit inconsistent}  shows the proportion of machine instructions spent on forward demodulation and on entire runs, showcasing that we 
reduce the solving time of the \PostOrdering{} problem. 
We used an instruction limit of $200\times\ 10^9$ (approximately 60 seconds).
\review{We highlight sub-results from \Cref{tab:instr-count} for problems in the UEQ (unit equality) category of TPTP in \Cref{tab:ueq-res}. Post-ordering checks problem in this category take up a more significant proportion of instructions, and even dominate the number of instructions when using LPO without TODs.}
\begin{table}[t]
    \caption{Number of instructions spent on \PostOrdering{} checks and forward demodulation compared to total instruction count. The first column describes the options, and the next three display the number of instructions performed on parts of \vampire. Columns 5--6 list the percentage of instructions spent on \PostOrdering{} checks and forward demodulation and the last column gives the number of forward demodulations.}
    \centering
        \renewcommand*\arraystretch{1.1}
    \setlength{\tabcolsep}{0.3em}
        \begin{tabular}{|l l l| c c c | c c | c |}
        \hline
        \multicolumn{3}{|c}{Options} & \multicolumn{3}{|c|}{\# of instructions $\times 10^{12}$} & \multicolumn{2}{c|}{Ratio to total} & \# of $\times 10^{6}$ \\
        \hdashline
                 &     &              & PostOrd.     & Fw dem.        & Total            & PostOrd. & Fw dem. & Fw dem.\\
        \hline
\multirow{6}{0.5cm}{\rotatebox{90}{Otter}}    & \multirow{3}{0.5cm}{\rotatebox{90}{KBO}}
        
               & \off    & \cellcolor{teal!5.8}73  & \cellcolor{teal!13.9}174 & 1258
                         & \cellcolor{teal!5.8}5.8\% & \cellcolor{teal!13.9}13.9\% & 2598 \\
         &     & \on     & \cellcolor{teal!2.7}34   & \cellcolor{teal!11.3}142	& 1256
                         & \cellcolor{teal!2.7}2.7\% & \cellcolor{teal!11.3}11.3\% & 2714 \\
         &     & \shared & \cellcolor{teal!1.4}18   & \cellcolor{teal!9.2}116 & 1255
                         & \cellcolor{teal!1.4}1.4\% & \cellcolor{teal!9.2}9.2\% & 2813 \\

\cdashline{2-9}
         & \multirow{3}{0.5cm}{\rotatebox{90}{LPO}}
        
               & \off    & \cellcolor{teal!20.0}259  & \cellcolor{teal!28.0}362 & 1292
                         & \cellcolor{teal!20.0}20.0\% & \cellcolor{teal!28.0}28.0\% & 1922 \\
         &     & \on     & \cellcolor{teal!12.0}154   & \cellcolor{teal!21.7}279	& 1286
                         & \cellcolor{teal!12.0}12.0\% & \cellcolor{teal!21.7}21.7\% & 2064 \\
         &     & \shared & \cellcolor{teal!7.7}98   & \cellcolor{teal!16.6}213 & 1284
                         & \cellcolor{teal!7.7}7.7\% & \cellcolor{teal!16.6}16.6\% & 2238 \\

\cdashline{1-9}
\multirow{6}{0.5cm}{\rotatebox{90}{Discount}}    & \multirow{3}{0.5cm}{\rotatebox{90}{KBO}}
        
               & \off    & \cellcolor{teal!7.2}91  & \cellcolor{teal!24.7}314 & 1272
                         & \cellcolor{teal!7.2}7.2\% & \cellcolor{teal!24.7}24.7\% & 3044 \\
         &     & \on     & \cellcolor{teal!3.6}45   & \cellcolor{teal!22.4}281	& 1258
                         & \cellcolor{teal!3.6}3.6\% & \cellcolor{teal!22.4}22.4\% & 3279 \\
         &     & \shared & \cellcolor{teal!3.0}37   & \cellcolor{teal!21.4}272 & 1269
                         & \cellcolor{teal!3.0}3.0\% & \cellcolor{teal!21.4}21.4\% & 3520 \\

\cdashline{2-9}
         & \multirow{3}{0.5cm}{\rotatebox{90}{LPO}}
        
               & \off    & \cellcolor{teal!16.5}214  & \cellcolor{teal!32.1}417 & 1301
                         & \cellcolor{teal!16.5}16.5\% & \cellcolor{teal!32.1}32.1\% & 2376 \\
         &     & \on     & \cellcolor{teal!8.9}115   & \cellcolor{teal!27.6}354	& 1283
                         & \cellcolor{teal!8.9}8.9\% & \cellcolor{teal!27.6}27.6\% & 2768 \\
         &     & \shared & \cellcolor{teal!6.7}86   & \cellcolor{teal!25.5}330 & 1294
                         & \cellcolor{teal!6.7}6.7\% & \cellcolor{teal!25.5}25.5\% & 3228 \\

\hline
    \end{tabular}
    \label{tab:instr-count}
\end{table}

\begin{table}[t]
    \caption{Number of instructions spent on \PostOrdering{} checks and forward demodulation compared to total instruction count, within UEQ problems (1456 problems).}
    \centering
    \renewcommand*\arraystretch{1.1}
    \setlength{\tabcolsep}{0.3em}
    \begin{tabular}{|l l l| c c c | c c | c |}
        \hline
        \multicolumn{3}{|c}{Options} & \multicolumn{3}{|c|}{\# of instructions $\times 10^{12}$} & \multicolumn{2}{c|}{Ratio to total} & \# of $\times 10^{6}$ \\
        \hdashline
                 &     &              & PostOrd.     & Fw dem.        & Total            & PostOrd. & Fw dem. & Fw dem.\\
        \hline
\multirow{6}{0.5cm}{\rotatebox{90}{Otter}}    & \multirow{3}{0.5cm}{\rotatebox{90}{KBO}}
        
               & \off    & \cellcolor{teal!23.4}23.6  & \cellcolor{teal!37.6}37.9 & 101.0
                         & \cellcolor{teal!23.4}23.4\% & \cellcolor{teal!37.6}37.6\% & 295.3 \\
         &     & \on     & \cellcolor{teal!10.8}10.8   & \cellcolor{teal!28.0}28.1	& 100.3
                         & \cellcolor{teal!10.8}10.8\% & \cellcolor{teal!28.0}28.0\% & 342.4 \\
         &     & \shared & \cellcolor{teal!5.1}5.1   & \cellcolor{teal!18.8}18.8 & 99.9
                         & \cellcolor{teal!5.1}5.1\% & \cellcolor{teal!18.8}18.8\% & 382.9 \\

\cdashline{2-9}
         & \multirow{3}{0.5cm}{\rotatebox{90}{LPO}}
        
               & \off    & \cellcolor{teal!51.6}58.4  & \cellcolor{teal!60.9}69.0 & 113.2
                         & \cellcolor{teal!51.6}51.6\% & \cellcolor{teal!60.9}60.9\% & 130.3 \\
         &     & \on     & \cellcolor{teal!30.8}33.9   & \cellcolor{teal!46.0}50.6	& 110.0
                         & \cellcolor{teal!30.8}30.8\% & \cellcolor{teal!46.0}46.0\% & 158.2 \\
         &     & \shared & \cellcolor{teal!18.0}20.0   & \cellcolor{teal!30.6}34.1 & 111.2
                         & \cellcolor{teal!18.0}18.0\% & \cellcolor{teal!30.6}30.6\% & 215.5 \\

\cdashline{1-9}
\multirow{6}{0.5cm}{\rotatebox{90}{Discount}}    & \multirow{3}{0.5cm}{\rotatebox{90}{KBO}}
        
               & \off    & \cellcolor{teal!23.8}26.8  & \cellcolor{teal!62.5}70.3 & 112.6
                         & \cellcolor{teal!23.8}23.8\% & \cellcolor{teal!62.5}62.5\% & 726.5 \\
         &     & \on     & \cellcolor{teal!13.8}15.4   & \cellcolor{teal!58.3}65.1	& 111.8
                         & \cellcolor{teal!13.8}13.8\% & \cellcolor{teal!58.3}58.3\% & 815.1 \\
         &     & \shared & \cellcolor{teal!10.4}11.6   & \cellcolor{teal!54.0}60.2 & 111.7
                         & \cellcolor{teal!10.4}10.4\% & \cellcolor{teal!54.0}54.0\% & 919.5 \\

\cdashline{2-9}
         & \multirow{3}{0.5cm}{\rotatebox{90}{LPO}}
        
               & \off    & \cellcolor{teal!51.8}63.0  & \cellcolor{teal!78.4}95.2 & 121.5
                         & \cellcolor{teal!51.8}51.8\% & \cellcolor{teal!78.4}78.4\% & 420.3 \\
         &     & \on     & \cellcolor{teal!32.6}39.0   & \cellcolor{teal!71.4}85.3	& 119.5
                         & \cellcolor{teal!32.6}32.6\% & \cellcolor{teal!71.4}71.4\% & 567.1 \\
         &     & \shared & \cellcolor{teal!24.2}28.9   & \cellcolor{teal!64.8}77.4 & 119.4
                         & \cellcolor{teal!24.2}24.2\% & \cellcolor{teal!64.8}64.8\% & 732.0 \\

\hline
    \end{tabular}
    \label{tab:ueq-res}
\end{table}

\section{Related Work and Conclusion}
We introduce the term ordering diagram (TOD) index to improve efficiency of term ordering checks. 
Our evaluation using \vampire{} shows significant improvement over  naive term ordering checks.

Decidability and complexity of problems related to KBO and LPO are studied in~\cite{LPOConstraintSolving,KBOConstraintSolvingNPComplete,OrientingKBO,PathOrderingConstraints}, complemented by efficient implementations~\cite{ThingsToKnowWhenImplementingKBO,EfficientCheckingOfTermOrderingConstraints,ThingsToKnowWhenImplementingLPO}. \review{In particular, in~\cite{EfficientCheckingOfTermOrderingConstraints}, each KBO ordering check is runtime specialized individually in one step. We improve upon this by preprocessing ordering checks completely lazily. Additionally, TODs allow algorithm specialization for an arbitrary number of sequential ordering checks, and also for LPO.}
The open problem~9 of~\cite{RewriteBasedDeductionAndSymbolicConstraints} is a generalization of the post-ordering problem. {Confluence trees}, developed to decide ground confluence of equational systems~\cite{ConfluenceTrees}, are similar to TODs. Confluence trees have also been adapted to ground reducibility checking~\cite{GroundReducibility}. {Shared rewriting} avoids the post-ordering problem by caching the results of demodulations in shared terms~\cite{SharedRewriting}. \review{Similarly to term indexing structures, the efficiency of term ordering diagrams is based on sharing information between indexed elements. However, insertion into and evaluation of term ordering diagrams is done lazily, necessitated by the more expensive ordering comparison operations they support.} To the best of our knowledge, our \tod{} approach gives the first algorithmic solution to efficient post-ordering checks.

Further work includes applying TODs to, for example, backward demodulation, constrained superposition, subsumption demodulation~\cite{SubsumptionDemodulation}, ground reducibility~\cite{GroundReducibility} or ground  joinability~\cite{Duarte2022}. Extending our  framework to further   simplification orders, such as the Weighted Path Order (WPO)~\cite{WeightedPathOrders}, is another task for the future. 


\paragraph{Acknowledgements.} This research was funded in whole or in part by the  ERC Consolidator Grant ARTIST 101002685, the ERC Proof of Concept Grant LEARN 101213411, the TU Wien Doctoral College SecInt, the FWF SpyCoDe Grant 10.55776/F85,  the WWTF grant ForSmart   10.47379/ICT22007, and the Amazon Research Award 2023 QuAT. 

\bibliographystyle{splncs04}
\bibliography{bibliography}

\newpage

\appendix

\chapter*{Appendix}

\section{Correctness of TOD Transformations}

In this section, we present a detailed proof of our main result.

\correctness*
\begin{proof}

(1) For proving 
\emph{termination}, we  introduce a well-founded order on TODs and show that every transformation replaces a TOD by a smaller one. In the proof, we use finite multiset extensions of orders and the fact that the multiset extension of a well-founded order is also well-founded.

We first introduce a mapping $\mu$ from nodes to finite multisets of terms as:
\begin{enumerate}
    \item[(i)] If node $n$ is a term comparison node $s\comp t$, then $\mu(n)$ is the multiset 
    $\setof{s,t}$.
    \item[(ii)] For any other node $n$, we set $\mu(n)$ to be  the empty multiset.
\end{enumerate}
Let us also define an order $>_\mu$ on nodes as follows: $n_2 >_\mu n_1$ if $\mu(n_2)$ is greater than $\mu(n_1)$ in the multiset extension of the order $\succ$ on terms. 

For every path $\pi$ in a \tod, we denote by $\mu(\pi)$ the multiset consisting of elements $\mu(n)$ for all nodes $n$ in $\pi$. We define an ordering, also denoted by $>_\mu$, on paths as follows: $\pi_2 >_\mu \pi_1$ if $\mu(\pi_2)$ is greater than $\mu(\pi_1)$ in the multiset extension of the order $>_\mu$ on nodes. 

Finally, for a \tod{} $\OD$, we denote by $\mu(\OD)$ the multiset consisting of all multisets $\mu(\pi)$, where $\pi$ is a path from the root node in $\OD$. We also define an ordering $>_\mu$ on \tods{} by letting $\OD_2 >_\mu \OD_1$ if $\mu(\OD_2) > \mu(\OD_1)$. 

We claim that for every transformation apart from node replication that changes a \tod{} $\OD$ to a 
\tod{} $\OD'$, we have that $\OD >_\mu \OD'$, which implies termination. The proof is by routine 
inspection of transformations. For example for the case $f = g$ of the LPO transformations 
(last transformation of Figure~\ref{fig:lpo-transformations}), we replace a path with a term comparison node $f(\bar{s}) \comp f(\bar{t})$ by a 
finite number of paths, so that every new term comparison node on these paths contains a 
comparison of a pair of terms strictly smaller in the multiset order than the multiset $\setof{f(\bar{s}),f(\bar{t})}$. 
Another example is the redundant node removal of Figure~\ref{fig:universal-transformations}. 
This transformation replaces on some paths nodes $n_1,n_2,n_3$ by $n_1,n_3$, which results in a smaller multiset.

Finally, we note that node replication results in a \tod{} containing exactly the same multiset of paths, but it can only be applied a finite number of times (it terminates when the \tod{} excluding the exit node becomes a tree). This concludes, (1) termination.

\noindent{\emph{(2) \tod{} equivalence.}} Every \tod{} transformation of Figures~\ref{fig:universal-transformations}--\ref{fig:lpo-transformations} preserves \tod{} equivalence. Redundant node removal preserves equivalence by the definition of the ``forces" relation (Definition~\ref{def:tod:traversal:foces}). Node replication preserves equivalence because the set of paths in the \tod{} does not change.

Let $n$ be a term comparison node labeled with $s\comp t$ in $\OD$ and let $s=f(s_1,\ldots,s_k)$ and $t=g(t_1,\ldots,t_m)$. Assume that we apply a KBO or LPO transformation to this node resulting in a TOD $\OD'$. It is sufficient to prove that for each query substitution s.t. the $\OD$ traversal for $\sigma$ contains $n$ and either $n_1$, $n_2$ or $n_3$, then the $\OD'$ traversal for $\sigma$ also contains $n_1$, $n_2$ or $n_3$, respectively.

\begin{figure}[t]
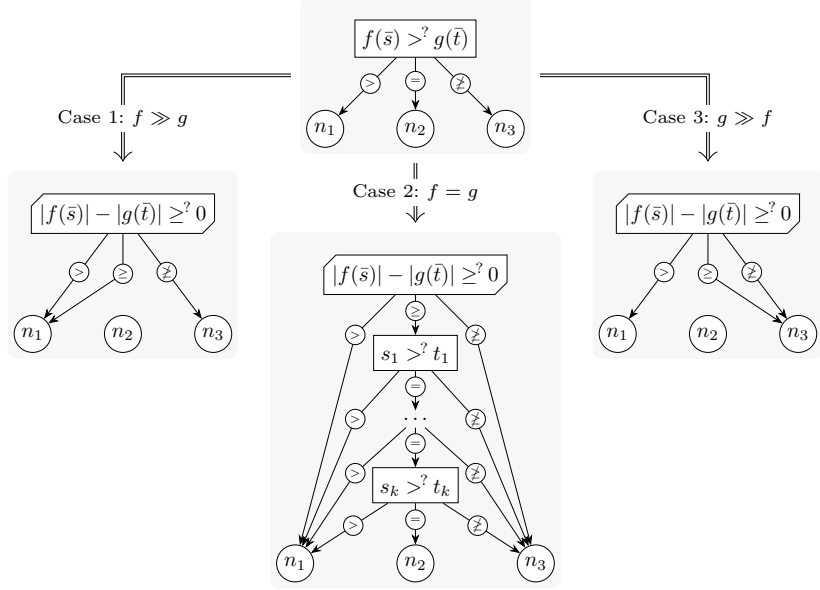

    \centering
    \usebox\kboTransformations
    \caption{KBO transformations on \tods{} (repeated).}
    \label{fig:kbo-transformations2}
\end{figure}
For KBO, we transform $n$ into a positivity check node labeled $|s|-|t|\pos 0$ in all cases. We repeat the KBO transformations for convenience in \Cref{fig:kbo-transformations2}. We consider the following subcases based on the definition of KBO:

\begin{enumerate}[label=(\arabic*),itemsep=.5em]

\item $\sigma(|s|-|t|)>0$. Then by~\ref{prop:kbo1}, $\sigma(s\comp t)$ equals $>$, so the $\OD$ traversal for $\sigma$ contains $n_1$. Further, $\sigma(s\comp t)$ equals $>$, so the $\ODprime$ traversal for $\sigma$ contains $n_1$ too.

\item $\sigma(|s|-|t|)\gtrsim 0$ and $f\gg g$. Then by~\ref{prop:kbo2}, $\sigma(s\comp t)$ equals $>$, so the $\OD$ traversal for $\sigma$ continues $n_1$. Further, $\sigma(|s|-|t|\pos 0)$ equals $\geq$ by assumption, so the $\ODprime$ traversal for $\sigma$ contains $n_1$ too.

\item $\sigma(|s|-|t|)\gtrsim 0$ and $g\gg f$. Then none of~\ref{prop:kbo1}--\ref{prop:kbo3} applies, so $\sigma(s\comp t)$ equals $\ngeq$, and the $\OD$ traversal for $\sigma$ contains $n_3$. Further, $\sigma(|s|-|t|\pos 0)$ equals $\geq$ and the $\ODprime$ traversal for $\sigma$ contains $n_3$ too.

\item $\sigma(|s|-|t|)\gtrsim 0$ and $f=g$. Then we also add $k$ new nodes $n'_1,\ldots,n'_k$ to $\ODprime$. We consider the following subcases:

\begin{enumerate}[label=(4.\arabic*),leftmargin=.7em]

    \item There is some $1\le i\le k$ s.t. $s_j\sigma=t_j\sigma$ for all $1\le j < i$ and $s_i\sigma\succ_\kbo t_i\sigma$. Then by~\ref{prop:kbo3}, $\sigma(s\comp t)$ equals $>$, so the $\OD$ traversal for $\sigma$ contains $n_1$. By assumption, $\sigma(|s|-|t|\pos 0)$ equals $\geq$. Also, $k\geq i>0$, so the $\ODprime$ traversal for $\sigma$ contains $n'_1$ with label $s_1\comp t_1$. For each $1\le j<i$, $n'_j$ is labeled with $s_j\comp t_j$, and $\sigma(s_j\comp t_j)$ equals $=$ by assumption, so we get that the $\ODprime$ traversal for $\sigma$ contains $n'_i$. This node is labeled with $s_i\comp t_i$, and by assumption, $\sigma(s_i\comp t_i)$ equals $>$, hence the $\ODprime$ traversal for $\sigma$ contains $n_1$ too.

    \item There is some $1\le i\le k$ s.t. $s_j\sigma=t_j\sigma$ for all $1\le j < i$ and $s_i\sigma\nsucceq_\kbo t_i\sigma$. Then none of~\ref{prop:kbo1}--\ref{prop:kbo3} applies, and $\sigma(s\comp t)$ equals $\ngeq$, so the $\OD$ traversal for $\sigma$ contains $n_3$. By assumption, $\sigma(|s|-|t|\pos 0)$ equals $\geq$. Also, $k\geq i>0$, so the $\ODprime$ traversal for $\sigma$ contains $n'_1$ with label $s_1\comp t_1$. Similarly to the previous case, we get that the $\ODprime$ traversal for $\sigma$ contains $n'_i$. This node is labeled with $s_i\comp t_i$, and by assumption, $\sigma(s_i\comp t_i)$ equals $\ngeq$, hence the $\ODprime$ traversal for $\sigma$ contains $n_3$ too.

    \item Otherwise, $s_i\sigma=t_i\sigma$ for all $1\le i\le k$. Then $s\sigma=t\sigma$, so $\sigma(s\comp t)$ equals $=$, so the $\OD$ traversal for $\sigma$ contains $n_2$. Also, $\sigma(|s|-|t|\pos 0)$ equals $\geq$. If $k=0$, then the $\ODprime$ traversal for $\sigma$ contains $n_2$. Otherwise, it contains $n'_1$. For each $1\le j\le k$, $n'_j$ is labeled with $s_j\comp t_j$, and $\sigma(s_j\comp t_j)$ equals $=$ by assumption, so we get that the $\ODprime$ traversal for $\sigma$ contains $n_2$ too.

\end{enumerate}

\item Otherwise, $\sigma(|s|-|t|)\ngtr 0$ and $\sigma(|s|-|t|)\not\gtrsim 0$ and none of~\ref{prop:kbo1}--\ref{prop:kbo3} applies, so $\sigma(s\comp t)$ equals $\ngeq$, and the $\OD$ traversal for $\sigma$ contains $n_3$. Also, $\sigma(|s|-|t|\pos 0)$ equals $\ngeq$, hence the $\ODprime$ traversal for $\sigma$ contains $n_3$ too.

\end{enumerate}
\begin{figure}[t]
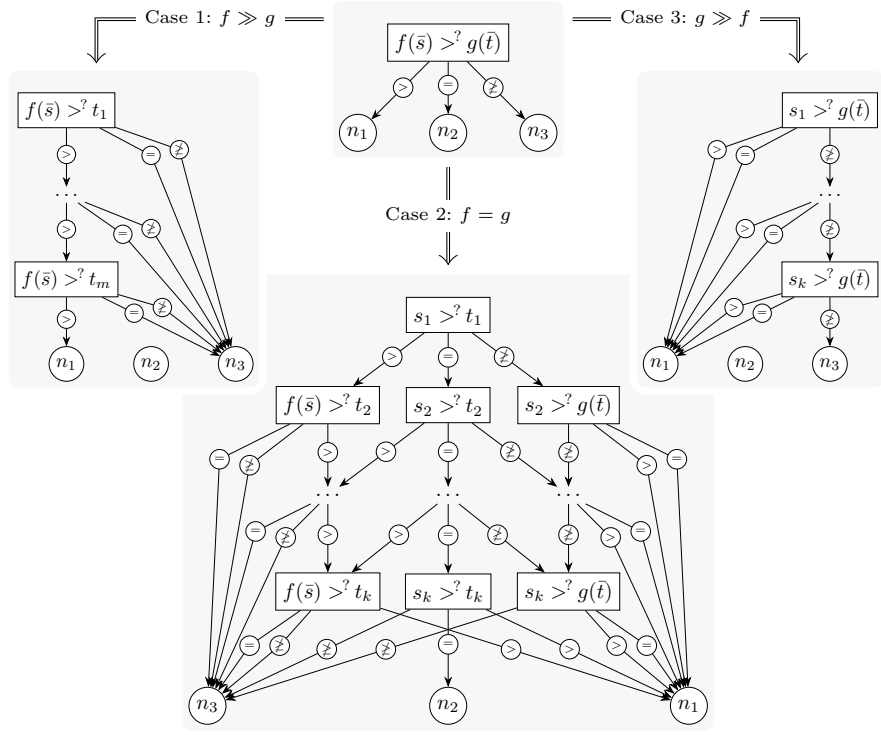

    \centering
    \usebox\lpoTransformations
    \caption{LPO transformations on \tods{} (repeated).}
    \label{fig:lpo-transformations2}
\end{figure}
For LPO, we consider the following cases. We also repeat the LPO transformations for convenience in~\Cref{fig:lpo-transformations2}.
\begin{enumerate}[label=(\arabic*),itemsep=.5em]

\item $f\gg g$. In the $m=0$ case, the $\OD$ traversal for $\sigma$ contains $n_1$, and we replace $n$ by $n_1$ in $\OD'$, so the $\OD'$ traversal for $\sigma$ contains $n_1$ too. Otherwise $m>0$. In this case, we also add $m-1$ new nodes to $\ODprime$. Let us denote $n$ in $\ODprime$ with $n'_1$ and the $m-1$ new nodes with $n'_2,\ldots,n'_m$, respectively. We have the following subcases.

\begin{enumerate}[label=(1.\arabic*),leftmargin=.7em,itemsep=.5em]
    \item For all $1\le i \le m$, we have $s\sigma\succ_\lpo t_i\sigma$. Then by~\ref{prop:lpo2}, $\sigma(s\comp t)$ equals $>$, and the $\OD$ traversal for $\sigma$ contains $n_1$. For each $1\le i\le m$, $n'_i$ is labeled with $s\comp t_i$, and $\sigma(s\comp t_i)$ equals $>$ by assumption. Hence, we get that the $\ODprime$ traversal for $\sigma$ contains $n_1$ too.

    \item There is some $1\le i\le m$ s.t. $s\sigma\succ_\lpo t_j\sigma$ for all $1\le j \le i$ and $s\sigma\nsucc_\lpo t_i\sigma$. \ref{prop:lpo1} and~\ref{prop:lpo2} do not apply. Further, for all $1\le j\le m$, we have $s\sigma\succ_\lpo s_j\sigma$ and $t\sigma\succ_\lpo t_i\sigma$, which means $s_j\sigma\succeq_\lpo t\sigma$ would contradict $s\sigma\nsucc_\lpo t_i\sigma$ and hence~\ref{prop:lpo3} does not apply either. Hence, $\sigma(s\comp t)$ equals $\ngeq$ and the $\OD$ traversal for $\sigma$ contains $n_3$.

    For each $1\le j< i$, $n'_j$ is labeled with $s\comp t_j$, and $\sigma(s\comp t_j)$ equals $>$ by assumption. Also, $n'_i$ is labeled with $s\comp t_i$ and $\sigma(s\comp t_i)$ equals either $=$ or $\ngeq$. In both cases, we get that the $\ODprime$ traversal for $\sigma$ contains $n_3$ too.
\end{enumerate}

\item $g\gg f$. In the $k=0$ case, the $\OD$ traversal for $\sigma$ contains $n_3$, and we replace $n$ by $n_3$ in $\OD'$, so the $\OD'$ traversal for $\sigma$ contains $n_3$ too. Otherwise $k>0$. In this case, we also add $k-1$ new nodes to $\ODprime$. Let us denote $n$ in $\ODprime$ with $n'_1$ and the $k-1$ new nodes with $n'_2,\ldots,n'_k$, respectively. We have the following subcases.

\begin{enumerate}[label=(2.\arabic*),leftmargin=.7em,itemsep=.5em]

    \item There is some $1\le i\le k$ s.t. $s_j\sigma\nsucceq_\lpo t$ for all $1\le j < i$ and $s_i\sigma\succeq_\lpo t\sigma$. Then by~\ref{prop:lpo3} $\sigma(s\comp t)$ equals $>$ and the $\OD$ traversal for $\sigma$ contains $n_1$. For each $1\le j< i$, $n'_j$ is labeled with $s_j\comp t$, and $\sigma(s_j\comp t)$ equals $\ngeq$ by assumption. Also, $n'_i$ is labeled with $s_i\comp t$ and $\sigma(s_i\comp t)$ equals either $>$ or $=$. In both cases, we get that the $\ODprime$ traversal for $\sigma$ contains $n_1$ too.

    \item For all $1\le i\le k$, $s_i\sigma\nsucceq_\lpo t\sigma$. Then none of~\ref{prop:lpo1}--\ref{prop:lpo3} applies, so $\sigma(s\comp t)$ equals $\ngeq$. For each $1\le i\le k$, $n'_i$ is labeled with $s_i\comp t$, and $\sigma(s_i\comp t)$ equals $\ngeq$ by assumption. Hence, we get that the $\ODprime$ traversal for $\sigma$ contains $n_3$ too.

\end{enumerate}

\item $f=g$. Also, $m=k$. Then, $m=k$. In the $k=0$ case, the $\OD$ traversal for $\sigma$ contains $n_2$, and we replace $n$ by $n_2$ in $\OD'$, so the $\OD'$ traversal for $\sigma$ contains $n_2$ too. Otherwise $k>0$. In this case, we also add $3(k-1)$ new nodes to $\ODprime$, shown as $k-1$ rows in Case 2 of Figure~\ref{fig:lpo-transformations}. We will denote the nodes from left to right in the $i$th row of this diagram with $n_{i+1,1}$, $n_{i+1,2}$ and $n_{i+1,3}$, respectively. We will also denote $n$ with $n_{1,2}$. We have the following subcases.

\begin{enumerate}[label=(3.\arabic*),leftmargin=.7em,itemsep=.5em]

    \item\label{case:success-in-greater} There exists $1\le i\le k$ s.t. $s_j\sigma=t_j\sigma$ for all $1\le j < i$, $s_i\sigma\succ_\lpo t_i\sigma$ and $s\sigma\succ_\lpo t_l\sigma$ for all $i < l \le k$. Then, \ref{prop:lpo1} applies, so $\sigma(s\comp t)$ equals $>$ and the $\OD$ traversal for $\sigma$ contains $n_1$. For each $1\le j < i$, $n_{j,2}$ is labeled $s_j\comp t_j$ and $\sigma(s_j\comp t_j)$ equals $=$ by assumption. Hence, the $\ODprime$ traversal for $\sigma$ contains $n_{i,2}$, which is labeled $s_i\comp t_i$. Then, $\sigma(s_i\comp t_i)$ equals $>$, so the $\ODprime$ traversal for $\sigma$ contains $n_1$, if $i=k$, otherwise it contains $n_{i+1,1}$. Similarly, for all $i<l\leq k$, $n_{l,1}$ is labeled $s\comp t_l$ and $\sigma(s\comp t_l)$ equals $>$ by assumption. Hence, the $\ODprime$ traversal for $\sigma$ contains $n_1$ too.

    \item\label{case:fail-in-greater} There exists $1\le i\le k$ and $i<j\le k$ s.t. $s_l\sigma=t_l\sigma$ for all $1\le l<i$, $s_i\sigma\succ_\lpo t_i\sigma$, $s\sigma\succ_\lpo t_l\sigma$ for all $i<l<j$ and $s\sigma\nsucc_\lpo t_j\sigma$. \ref{prop:lpo1} and \ref{prop:lpo2} do not apply. Furthermore, for all $1\le l\le k$, we have $s\sigma\succ_\lpo s_l\sigma$, and we also have $t\sigma\succ_\lpo t_j\sigma$. This means $s_l\sigma\succeq_\lpo t\sigma$ would contradict $s\sigma\nsucc_\lpo t_j\sigma$, so \ref{prop:lpo3} does not apply either. Hence, $\sigma(s\comp t)$ equals $\ngeq$ and the $\OD$ traversal for $\sigma$ contains $n_3$. Similarly to the previous case, we can conclude that the $\ODprime$ traversal for $\sigma$ contains $n_{j,1}$, which is labeled $s\comp t_j$. We have $\sigma(s\comp t_j)$ equals $=$ or $\ngeq$ by assumption. In both cases, the $\ODprime$ traversal for $\sigma$ contains $n_3$ too.

    \item For all $1\le i\le k$, we have $s_i\sigma=t_i\sigma$. \ref{prop:lpo1} and \ref{prop:lpo2} do not apply and for all $1\le i\le k$, $t\sigma\succ_\lpo t_i\sigma=s_i\sigma$ implies $s_i\sigma\nsucceq t\sigma$, so \ref{prop:lpo3} does not apply either. This means $\sigma(s\comp t)$ equals $=$ and the $\OD$ traversal for $\sigma$ contains $n_2$. For every $1\le i\le k$, $n_{i,2}$ is labeled $s_i\comp t_i$ and $\sigma(s_i\comp t_i)$ equals $=$ by assumption. Hence, the $\ODprime$ traversal for $\sigma$ contains $n_2$ too.

    \item There exists $1\le i\le k$ and $i<j\le k$ s.t. $s_l\sigma=t_l\sigma$ for all $1\le l<i$, $s_i\sigma\nsucceq_\lpo t_i\sigma$, $s_l\sigma\nsucceq_\lpo t\sigma$ for all $i<l<j$ and $s_j\sigma\succeq_\lpo t\sigma$. Since \ref{prop:lpo3} applies, $\sigma(s\comp t)$ equals $>$ and the $\OD$ traversal for $\sigma$ contains $n_1$. This case is symmetric to \ref{case:fail-in-greater}.

    \item There exists $1\le i\le k$ s.t. $s_j\sigma=t_j\sigma$ for all $1\le j<i$, $s_i\sigma\nsucceq_\lpo t_i\sigma$ and $s_l\sigma\nsucceq_\lpo t\sigma$ for all $i<l\le k$. \ref{prop:lpo1} and \ref{prop:lpo2} do not apply. Since (a) for all $1\le j<i$, $t\sigma\succ_\lpo t_j\sigma=s_j\sigma$ implies $s_j\sigma\nsucceq_\lpo t\sigma$, 
    (b) $s_i\sigma\nsucceq_\lpo t_i\sigma$ and $t\sigma\succ_\lpo t_i\sigma$ implies $s_i\sigma\nsucceq_\lpo t\sigma$ and (c) for all $i<l\le k$, $s_l\sigma\nsucceq_\lpo t\sigma$, \ref{prop:lpo3} does not apply either. This means $\sigma(s\comp t)$ equals $\ngeq$ and the $\OD$ traversal for $\sigma$ contains $n_3$. This case is symmetric to \ref{case:success-in-greater}.\QED

\end{enumerate}
\end{enumerate}
\end{proof}

\newpage 
\newcommand\TR{\mathcal{TR}}

\section{Forcing Function Implementation}
\label{sec:redundant-node-removal}
In this appendix, we describe the algorithm used as our forcing function $F$ in retrieval (Section~\ref{sec:retrieval}). Let $\OD$ be a \tod{} and $n$ a node in $\OD$. Moreover, let $\pi=n_0,\ldots,n_k,n$ be a path in $\OD$. We assume $\OD$, $n$ and $\pi$ arbitrary but fixed in this section. Our main goal is to determine whether the node $n$ forces a label using information within the path $\pi$. To achieve this, we have to translate this information into formulas. We build these formulas out of constraints, as defined next.
\begin{Definition}[Term Constraint]
Let $s$ and $t$ be terms. We call a \emph{term constraint} any expression of the form $s>t$ or $s=t$. Let $\sigma$ be a substitution. We say that $\sigma$ \emph{satisfies} the term constraint $s>t$ if $s\sigma\succ t\sigma$. We say that $\sigma$ \emph{satisfies} the term constraint $s=t$ if $s\sigma=t\sigma$.\QED
\end{Definition}
We denote $s>t$ also by $t<s$, denote $s>t\lor s=t$ by $s\geqslant t$, denote $t\geqslant s$ by $s\leqslant t$, and denote $\neg(s\geqslant t)$ by $s\ngeqslant t$.
%
Our algorithm works on formulas $F\to G$ where $F$ and $G$ are quantifier-free formulas containing term constraints.
The formula $F$ will contain information about $\pi$, while the formula $G$ will express that $n$ forces a specific label. Our reasoning requires the following two assumptions about $n$ and $\OD$:
\begin{enumerate}[label=(\arabic*)]
\item The node $n$ is reachable from the root node of $\OD$ via exactly one path.\label{one-path-assumption}
\item There is a substitution $\sigma$ such that the $\OD$ traversal for $\sigma$ contains $n$.\label{satisfiability-assumption}
\end{enumerate}
The main advantage of using assumption~\ref{one-path-assumption} is that there are no paths other than $\pi$ ending in $n$ and our reasoning simplifies to just a conjunction of constraints. Assumption~\ref{satisfiability-assumption} ensures that $n$ is reachable during retrievals and that the formula $F$ is satisfiable.

As described in Section~\ref{sec:retrieval} and as a consequence of Lemma~\ref{lem:visited}, these two assumptions are enforced by marking $n$ visited, after applying node replication to it if it has multiple incoming edges.

Our algorithm determines whether a term comparison node forces a label, using term constraints. Efficient implementation of this part is crucial, as KBO mostly uses term comparison nodes and LPO uses exclusively term comparison nodes.

First, we define a formula collecting term constraints from $\pi$ as follows.
\begin{Definition}[Term Formula]
Let $\pi=n_0,\ldots,n_k,n$ be a path in $\OD$. The set of \emph{top-level terms in $\pi$} is defined as:
$$\{s,t\mid s\comp t\text{ is a term comparison label in }\pi\}$$
The \emph{term formula for $\pi$} is the conjunction of term constraints containing
\begin{enumerate}[label=(\arabic*)]
\item $s>t$, $s=t$, or $s\ngeqslant t$, for every term comparison node (except for $n$) labeled $s\comp t$ in $\pi$ followed by an edge label $>$, $=$, or $\ngeq$, respectively, and
\item $s>t$, for all terms $s$ and $t$ in the set of top-level terms in $\pi$ with $s\succ t$.\QED
\end{enumerate}
\end{Definition}
It is easy to see that a query substitution $\sigma$ satisfies the term formula for $\pi$ if $\pi$ is a prefix of the $\OD$ traversal for $\sigma$.

Let us denote the conjunction of the following set of transitivity axioms with $\TR$:
\begin{align}
\forall x,y,z.(x=y\land y=z\to x=z)\tag{tr1}\label{eq:tr1}\\
\forall x,y,z.(x\leqslant y\land y<z\to x<z)\tag{tr2}\label{eq:tr2}\\
\forall x,y,z.(x<y\land y\leqslant z\to x<z)\tag{tr3}\label{eq:tr3}\\
\forall x,y,z.(x\ngeqslant y\land y\leqslant z\to x\ngeqslant z)\tag{tr4}\label{eq:tr4}\\
\forall x,y,z.(x\leqslant y\land y\ngeqslant z\to x\ngeqslant z)\tag{tr5}\label{eq:tr5}
\end{align}
The validity of the formula $\TR$ is a consequence of the properties of the simplification ordering $\succ$ underlying term constraint semantics. 

Let $n$ be a term comparison node labeled $s\comp t$ and $F$ the term formula for $n$. Clearly, if the $s>t$, $s=t$, or $s\ngeqslant t$ is implied by $\TR\land F$, then $n$ forces label $>$, $=$, or $\ngeq$, respectively.

We decide this using $\TR$ and $F$ as follows. We build a data structure for the node $n$, called a \emph{term partial ordering} (TPO for short), which contains the entire transitive closure of $F$ w.r.t. the transitivity axioms $\TR$. Then, we simply query if $s>t$, $s=t$ or $s\ngeqslant t$ is in the transitive closure.

\begin{example}
Let $f$ and $g$ be unary function symbols and consider the following path in a TOD:
\begin{center}
\begin{tikzpicture}[canvassed]
\node[inner sep=0pt] (source) at (0,.6) {};
\node[termnode] (c1) at (0,0) {$f(x)\comp y$};
\node[termnode] (c2) at (4,0) {$g(y)\comp z$};
\node[termnode] (c3) at (8,0) {$x\comp z$};

\draw[->] (source) -- (c1);
\draw[->] (c1) -- node[ngeqedgenode]{} (c2);
\draw[->] (c2) -- node[eqedgenode]{} (c3);
\end{tikzpicture}
\end{center}
We check whether the node labeled $x\comp z$ forces a label. The top-level terms in the path are $\{x,y,z,f(x),g(y)\}$. The term formula $F$ for the node labeled $x\comp z$ is the following:
$$f(x)\ngeqslant y\land g(y)=z\land f(x)>x\land g(y)>y$$
The first two conjuncts are added due to the term comparison nodes labeled $f(x)\comp y$ and $g(y)\comp z$, respectively. The third and fourth conjuncts are added because $f(x)$ and $x$, respectively $g(y)$ and $y$ are top-level terms such that $f(x)\succ x$, respectively $g(y)\succ y$. We derive the following new term constraints from $\TR\land F$:
\begin{enumerate}
    \item $x\ngeqslant y$ due to $x < f(x)$ and $f(x)\ngeqslant y$, using axiom~\eqref{eq:tr5},
    \item $z > y$ due to $y<g(y)$ and $g(y)=z$, using axiom~\eqref{eq:tr3},
    \item $x\ngeqslant z$ due to $x\ngeqslant y$ and $y < z$, using axiom~\eqref{eq:tr4}.
\end{enumerate}
Since $x\ngeqslant z$ is implied by $\TR\land F$, we conclude that the node labeled $x\comp z$ forces label $\ngeq$.\QED
\end{example}
\begin{remark}
It is possible to further strengthen reasoning over term comparison nodes. For example, a positivity check node labeled $2+|x|-|y|\pos 0$ followed by the edge label $>$ can be converted into a conjunction of term constraints of the form $s>t$, where $s$ contains exactly one more occurrence of $x$ than $t$, $t$ contains exactly one more occurrence of $y$ than $s$, and the difference in symbol weights between $s$ and $t$ add up to exactly 2. Adding such term constraints could lead to significant improvements in our algorithm. However, depending on the function signature $\Funs$, it may require reasoning over infinite sets of term constraints. We consider such improvements outside the scope of our work and leave it as future work.
\end{remark}

\paragraph{Implementation of term partial orderings.} We build TPOs efficiently as follows. The root node of the TOD $\OD$ is associated with the empty TPO. When marking $n$ visited, we copy the TPO from the predecessor of $n$, we extend it with new term constraints if needed, and then we update the transitive closure, assuming that the transitivity axioms only have to be used on new term constraints.

With this approach, TPOs are built iteratively and we do not need to recompute a TPO from scratch, which decreases our computational costs. However, TPOs need to be permanently stored for each visited node in a TOD, which can significantly increase our memory footprint. To reduce the memory cost as well, we introduce a level of abstraction, using so-called \emph{partial orderings}.

Similarly to TPOs, partial orderings store the transitive closure of equality and inequality constraints over the elements of any set $\mathcal{A}$. They store the transitive closure in a \emph{triangular array} and they are \emph{perfectly shared}: when adding new constraints to a partial ordering, we calculate the transitive closure, and check whether we already store an instance of this partial ordering. If yes, we use the already existing instance for further computations. This check can be simply done by hashing and comparing over the triangular arrays of partial orderings. By using partial orderings, we break symmetries in TPO computations, which can result in memory savings.

\newpage
\section{Making Term Comparisons Faster}

As our experiments have already shown, performing term comparisons can take up a significant part of a theorem prover's time. We have demonstrated that instances of the post-ordering problem can be solved efficiently using term ordering diagrams.

In this section, we focus on a slightly different problem related to term comparisons, defined below.

\begin{mdframed}[frametitle={\fcolorbox{black}{white}{~The Substituted Term Comparison Problem~}}, innertopmargin=0pt, innerbottommargin=8pt, frametitleaboveskip=-5pt, frametitlealignment=\center, frametitlefont=\scshape]
Given terms $s$ and $t$, and substitutions $\sigma$ and $\theta$, check if $s\sigma\succ t\theta$ holds.
\end{mdframed}
For example, evaluation of term comparison nodes during retrieval from TODs are instances of the above problem. In practice, the problem is solved by performing the following steps:
\begin{enumerate}
\item apply $\sigma$ on $s$ and $\theta$ on $t$, resulting in the terms $s'$ and $t'$, respectively,
\item compare $s'$ and $t'$ using $\succ$, that is, determine whether $s'\succ t'$, $s'=t'$ or $s'\prec t'$ holds.
\end{enumerate}
We noted earlier that KBO and LPO can be implemented in linear, respectively quadratic time in the size of the compared terms~\cite{ThingsToKnowWhenImplementingKBO,ThingsToKnowWhenImplementingLPO}. However, we can still improve the above steps using the following two observations:
\paragraph{Observation 1.} We do not have to compute $s'$ and $t'$. In other words, it is possible to solve the problem without actually applying $\sigma$ on $s$ and $\theta$ on $t$. We can exploit this observation by extending simplification orderings to so-called \emph{closure terms}.

\paragraph{Observation 2.} If $s'\succ t'$ cannot hold, it is unnecessary to check $s'=t'$ and $s'\prec t'$, as in the above problem formulation, we only need to know if $s'\succ t'$ holds. To exploit this, we implement and use \emph{unidirectional} comparisons for simplification orderings instead of using \emph{bidirectional} comparisons.

\subsection{Simplification Orderings over Closure Terms}

Observation 1 can be exploited for the following reasons. Let $s$ be a term and $\sigma$ a substitution. The main computational cost in applying $\sigma$ on $s$ resulting in $s'$ is having to \emph{construct} subterms of $s'$ and $s'$ itself. This operation takes linear time in the size of $s'$, since $s'$ has to be traversed and all its subterms have to be constructed anew.

The computational cost can be partly reduced by \emph{perfectly sharing terms}, that is, by representing instances of the same terms (up to syntactic equality) at most once in the search space. In that case, we can reuse ground subterms and the terms that are substituted by $\sigma$ for variables in $s$, resulting in a linear time algorithm in the size of $s$. But even perfectly sharing terms has an overhead, since for each subterm of $s$, it has to be determined whether $s\sigma$ is shared, that is, whether it is already represented in the search space. In the following, we assume that terms are perfectly shared.

While comparing terms, both KBO and LPO use recursive definitions that traverse the terms top-down: it is enough to inspect the top-most symbol of the terms, and then recurse into one of the subcases accordingly. Therefore, when we compare $s'$ with another term, we have to traverse it for a second time.

Our aim is now to define and implement KBO and LPO such that the comparison is performed \emph{upon first traversal of $s'$ without actually building $s'$}. This can be achieved by:
\begin{enumerate}
    \item traversing all non-variable subterms of $s$ as if they were subterms of $s'$,
    \item applying the substitution only when reaching a variable $x$ in $s$, in which case we continue traversal in $x\sigma$.
\end{enumerate}
For this, we have to work on pairs of terms and substitutions, as defined below.
\begin{Definition}[Closure Term]
Let $s$ be a term and $\sigma$ a substitution. We call the pair $s\cdot\sigma$ a \emph{closure term}.\QED
\end{Definition}
Since the point of using a closure term $s\cdot\sigma$ is to avoid applying $\sigma$ on $s$, we need to define some basic operations on closure terms.

For example, using perfectly shared terms makes checking syntactic equality ($=$) in practice a constant time operation: two terms are equal if they are represented by the same term in the search space, hence it is sufficient to check this representation\footnote{This is usually a pointer comparison taking constant time.} It would be much more costly to perfectly share closure terms. Instead, we use a recursive version of syntactic equality which takes linear time in the terms sizes, corresponding to the following inductive definition. Let us denote the empty substitution by $\varepsilon$.
\begin{Definition}[Equality for Closure Terms]
\label{def:closure-term-equality}
Let $s\cdot\sigma$ and $t\cdot\theta$ be closure terms. We say that $s\cdot\sigma$ and $t\cdot\theta$ are \emph{equal} (and write $s\cdot\sigma=t\cdot\theta$) if either:
\begin{enumerate}[label=(\arabic*)]
    \item $\sigma=\theta=\varepsilon$ or $s$ and $t$ are ground, and $s=t$,
    \item $\sigma\neq\varepsilon$, $s$ is a variable $x$, and $x\sigma\cdot\varepsilon=t\cdot\theta$,
    \item $\theta\neq\varepsilon$, $t$ is a variable $x$, and $s\cdot\sigma=x\theta\cdot\varepsilon$,
    \item $s=f(s_1,\ldots,s_k)$, $t=f(t_1,\ldots,t_k)$ and $s_i\cdot\sigma=t_i\cdot\theta$ for all $1\le i\le k$.\QED
\end{enumerate}
\end{Definition}
\begin{example}
Let $f$ be a binary symbol, $a$ and $b$ constants, $\sigma=\{x\mapsto a\}$ and $\theta=\{y\mapsto b\}$. We verify that $f(x,b)\cdot\sigma=f(a,y)\cdot\theta$ as follows. Case (4) of Definition~\ref{def:closure-term-equality} applies, so we recurse into the arguments:
\begin{enumerate}[leftmargin=1.2em]
\item For $x\cdot\sigma=a\cdot\theta$, case (2) applies. We get $a\cdot\varepsilon=a\cdot\theta$, which is true by case (1) since $a$ is ground.
\item For $b\cdot\sigma=y\cdot\theta$, case (3) applies. We get $b\cdot\sigma=b\cdot\varepsilon$, which is true by case (1) since $b$ is ground.\QED
\end{enumerate}
\end{example}
For KBO, we define linear expressions for closure terms as follows.
\begin{Definition}[Linear Expressions for Closure Terms]
Let $s\cdot\sigma$ be a closure term. The weight of $s\cdot\sigma$, denoted by $|s\cdot\sigma|$, is defined as:
\begin{enumerate}[label=(\arabic*)]
\item $|x\sigma|$ if $s$ is a variable $x$,
\item $w(f) + |s_1\cdot\sigma| + \ldots + |s_k\cdot\sigma|$ if $s=f(s_1,\ldots,s_k)$.\QED
\end{enumerate}
\end{Definition}
\begin{example}
Let $\sigma=\{x\mapsto f(y,z)\}$. We compute $|f(z,a)\cdot\sigma|$ as follows:
\begin{align*}
|f(z,a)\cdot\sigma|&=w(f)+|z\cdot\sigma|+|b\cdot\sigma|\\
&=1+|f(y,z)|+w(b)\\
&=2+w(f)+|y|+|z|\\
&=3+y+z\tag*{\QED}
\end{align*}
\end{example}
We are now ready to extend KBO and LPO to closure terms. These definitions follow the original definitions very closely.

For closure terms $s\cdot\sigma$ and $t\cdot\theta$, we write $s\cdot\sigma \succ_\kbo t\cdot\theta$ if either:
\begin{enumerate}[label=(K\arabic*),itemsep=2pt,leftmargin=2.5em]
\item $\sigma=\theta=\varepsilon$ or $s$ and $t$ are ground, and $s\succ_\kbo t$,
\item $\sigma\neq\varepsilon$, $s$ is a variable $x$, and $x\sigma\cdot\varepsilon\succ_\kbo t\cdot\theta$,
\item $\theta\neq\varepsilon$, $t$ is a variable $x$, and $s\cdot\sigma\succ_\kbo x\theta\cdot\varepsilon$,
\item $|s\cdot\sigma|-|t\cdot\theta|>0$,
\item $|s\cdot\sigma|-|t\cdot\theta|\gtrsim 0$, $s = f(s_1,...,s_n)$, $t = g(t_1,...,t_m)$ and $f\gg g$,
\item $|s\cdot\sigma|-|t\cdot\theta|\gtrsim 0$, $s = f(s_1,...,s_n)$, $t = f(t_1,...,t_n)$ and there exists $1\le i\le n$ such that
$s_j\cdot\sigma=t_j\cdot\theta$ for all $1\le j<i$ and $s_i\cdot\sigma\succ_\kbo t_i\cdot\theta$.
\end{enumerate}
Next, we extend the lexicographic-path ordering to closure terms as follows. Let $s\cdot\sigma$ and $t\cdot\theta$ be closure terms. We write  $s\cdot\sigma \succ_\lpo t\cdot\theta$ if either:
\begin{enumerate}[label=(L\arabic*),itemsep=2pt,leftmargin=2.5em]
\item $\sigma=\theta=\varepsilon$ or $s$ and $t$ are ground, and $s\succ_\lpo t$,
\item $\sigma\neq\varepsilon$, $s$ is a variable $x$, and $x\sigma\cdot\varepsilon\succ_\lpo t\cdot\theta$,
\item $\theta\neq\varepsilon$, $t$ is a variable $x$, and $s\cdot\sigma\succ_\lpo x\theta\cdot\varepsilon$,
\item $s=f(s_1,\ldots,s_n)$, $t=f(t_1,...,t_n)$ and there exists $1\le i\le n$ s.t. $s_j\cdot\sigma=t_j\cdot\theta$ for  $1\le j<i$, $s_i\cdot\sigma\succ_\lpo t_i\cdot\theta$, and $s\cdot\sigma\succ_\lpo t_k\cdot\theta$ for  $i<k\le n$,
\item $s=f(s_1,\ldots,s_n)$, $t=g(t_1,...,t_m)$, $f\gg g$ and $s\cdot\sigma\succ_\lpo t_i\cdot\theta$ for  $1\le i\le m$,
\item $s=f(s_1,\ldots,s_n)$, $s_i\cdot\sigma\succeq_\lpo t\cdot\theta$ for some $1\le i\le n$.
\end{enumerate}
Based on these definitions, it is straightforward to adapt any implementation of KBO or LPO to closure terms.
\begin{example}
Suppose we use an LPO $\succ$, and consider the term $f(x,y)$ and substitutions $\sigma=\{x\mapsto f(z,u), y\mapsto z\}$ and $\theta=\{x\mapsto z, y\mapsto f(u,z)\}$. We check $f(x,y)\cdot\sigma\succ_\lpo f(x,y)\cdot\theta$ as follows. (L4) applies, and we have the following subcomparisons:
\begin{enumerate} 
\item check $x\cdot\sigma\succ_\lpo x\cdot\theta$: (L2) applies, so we compare $f(z,u)\cdot\varepsilon$ and $x\cdot\theta$, then (L3) applies, so we compare $f(z,u)\cdot\varepsilon$ and $z\cdot\varepsilon$. Finally, $f(z,u)\cdot\varepsilon\succ_\lpo z\cdot\varepsilon$ holds due to (L1) and $f(z,u)\succ_\lpo z$.
\item check $f(x,y)\cdot\sigma\succ_\lpo y\cdot\theta$: (L3) applies, so we compare $f(x,y)\cdot\sigma\succ_\lpo f(u,z)\cdot\varepsilon$, then (L4) applies, so we perform the following subcomparisons:
\begin{enumerate}
\item check $x\cdot\sigma\succ_\lpo u\cdot\varepsilon$, which by (L2) is equivalent to checking $f(z,u)\cdot\varepsilon\succ_\lpo u\cdot\varepsilon$. This holds by (L1) and $f(z,u)\succ_\lpo u$.
\item check $f(x,y)\cdot\sigma\succ_\lpo z\cdot\varepsilon$: (L6) applies, so we check whether $x\cdot\sigma\succeq_\lpo z\cdot\varepsilon$ or $y\cdot\sigma\succeq_\lpo z\cdot\varepsilon$ holds. The first one holds, since (L2) applies, and the equivalent $f(z,u)\cdot\varepsilon\succ_\lpo z\cdot\varepsilon$ holds by (L1) and $f(z,u)\succ_\lpo z$.\QED
\end{enumerate}
\end{enumerate}
\end{example}

\subsection{Unidirectional Term Comparisons}

We now focus on exploiting Observation 2, by developing \emph{unidirectional} variants of the KBO and LPO implementations. We emphasize that KBO and LPO have linear, respectively quadratic time algorithms~\cite{ThingsToKnowWhenImplementingKBO,ThingsToKnowWhenImplementingLPO} for general, bidirectional comparisons, and these algorithms are known to be asymptotically optimal.

The unidirectional variants of these algorithms are very similar to their bidirectional counterparts: we compare terms $s$ and $t$ as usual, but when the results cannot be $s\succ t$ anymore, we fail.
\footnote{We can also check $s=t$ ``for free'' while checking $s\succ t$. Since both KBO and LPO compare subterms lexicographically in case of $s=t$, we cannot fail earlier than checking that $s$ and $t$ are equal.}
The complexity of the unidirectional variants is the same as the bidirectional ones, so we cannot improve them asymptotically.

Our LPO algorithm is essentially the same as the unidirectional variant described in~\cite{ThingsToKnowWhenImplementingLPO}. This variant can yield some improvements in time efficiency, however it still remains quadratic in nature. We have not yet found a way to improve this algorithm, so we will not discuss it here in detail.

For KBO, modifying the linear-time bidirectional algorithm to fail early yields only minor improvements. This can be explained by inspecting the definition of KBO: for two terms $s$ and $t$, in all three cases ~\ref{prop:kbo1}--\ref{prop:kbo3}, a unidirectional algorithm has to traverse at least $s$ fully to conclude that the linear expression $|s|-|t|$ cannot be non-negative, and hence all cases~\ref{prop:kbo1}--\ref{prop:kbo3} fail.

Instead of using the optimal unidirectional algorithm for KBO, we introduce a different algorithm which has worst-case quadratic time complexity, but often works much better in practice than the optimal, linear algorithm. Let $s$ and $t$ be terms. The main idea when checking $s\succ t$ is that we follow the definition of KBO:
\begin{enumerate}
    \item We compute the linear expression $|s|-|t|$.
    \item If either~\ref{prop:kbo1} or~\ref{prop:kbo2} applies, we exit early with success.
    \item If none of~\ref{prop:kbo1}--\ref{prop:kbo3} applies, we exit early with failure.
    \item Otherwise, we check if~\ref{prop:kbo3} applies recursively.
\end{enumerate}
This algorithm clearly takes quadratic time in the sizes of $s$ and $t$. However, instead of always computing the linear expressions from scratch, we utilize \emph{term sharing} to memoize the linear expressions for terms that we have already compared.

In particular, when seeing a term $s$ for the first time in the comparison algorithm, we compute the linear expression $|s|$, and memoize $|s|$ in the shared term representation of $s$. When computing a linear expression  $|s|-|t|$ for two terms $s$ and $t$ that we have already seen, we use the linear expressions $|s|$ and $|t|$ stored in the shared term representations of $s$ and $t$, respectively.

\subsection{Experimental results}

We have implemented term ordering checks with closure terms, and unidirectional comparisons for both KBO and LPO in \vampire{}. We now evaluate these variants, using them for the post-ordering check in forward demodulation. Our experimental setup is similar to that of Section~\ref{sec:evaluation}, using the option {\tt -sa} with values {\tt discount} and {\tt otter}, and the option {\tt -to} with values {\tt kbo} and {\tt lpo}. Note however that we do not use TODs in these experiments. We compare the following four configurations:
\begin{enumerate}
\item {\tt baseline} uses {\bf bidirectional} comparison {\bf without} closure terms,
\item {\tt conf1} uses {\bf bidirectional} comparison {\bf with} closure terms,
\item {\tt conf2} uses {\bf unidirectional} comparison {\bf without} closure terms,
\item {\tt conf3} uses {\bf unidirectional} comparison {\bf with} closure terms.
\end{enumerate}
Our experiments were run with a 60-second timeout. The results are shown in \Cref{tab:closure-unidirectional1,tab:closure-unidirectional2}. The tables show that both unidirectional comparisons and comparisons with closure terms improve the running times and also the number of solved problems.

\begin{table}[t]
    \caption{Number of problems solved by \vampire{} within 60s. The numbers in parentheses show the number of problems lost and won by {\tt conf1}, {\tt conf2} and {\tt conf3} compared to {\tt baseline}.}
    \centering
    \setlength{\tabcolsep}{.3em}
    \begin{tabular}{|l l| p{6.5em} | p{6.5em} | p{6.5em} | p{6.5em} |}
        \hline
                 &     & {\tt baseline} & {\tt conf1} & {\tt conf2} & {\tt conf3} \\
        \hline
        Otter    & KBO & 3235 (+0, -0) & 3246 (+13, -2) & 3236 (+2, -1) & 3252 (+19, -2) \\
                 & LPO & 3029 (+0, -0) & 3046 (+18, -1) & 3039 (+10, -0) & 3063 (+36, -2) \\\hdashline
        Discount & KBO & 3165 (+0, -0) & 3171 (+8, -2) & 3170 (+6, -1) & 3183 (+20, -2\\
                 & LPO & 2979 (+0, -0) & 2991 (+13, -1) & 2987 (+9, -1) & 3018 (+39, -0)\\
        \hline
    \end{tabular}
    \label{tab:closure-unidirectional1}
\end{table}

\begin{table}[t]
    \caption{Total time to solve problems by \vampire{} within 60s. Each line only considers problems solved by all configurations {\tt baseline/conf1/conf2/conf3}.}
    \centering
    \setlength{\tabcolsep}{.3em}
    \begin{tabular}{|l l| p{6em} | p{5em} | p{5em} | p{5em} | p{5em} |}
        \hline
                 &      & Solved by all & {\tt baseline} & {\tt conf1} & {\tt conf2} & {\tt conf3} \\
        \hline
        Otter    & KBO & 3232 & 3h44 & 3h33 & 3h43 & 3h28 \\
                 & LPO & 3027 & 3h17 & 3h03 & 3h09 & 2h51 \\\hdashline
        Discount & KBO & 3163 & 3h23 & 3h11 & 3h21 & 3h05 \\
                 & LPO & 2978 & 2h47 & 2h36 & 2h38 & 2h24 \\
        \hline
    \end{tabular}
    \label{tab:closure-unidirectional2}
\end{table}


We also found that for 682, 515, 536 and 438 benchmarks in Otter KBO, Otter LPO, Discount KBO, and Discount LPO runs, respectively, proof search in the configurations not using closure terms ({\tt baseline} and {\tt conf2}) diverged from the configurations using closure terms ({\tt conf1} and {\tt conf3}). This means that in these benchmarks, the proof in the former configurations was different in at least one step from the proof in the former configurations.

This is a side effect closure terms have on term sharing: without closure terms, some shared terms are instantiated unnecessarily before a post-ordering check failure, whereas with closure terms, we only instantiate terms after the post-ordering check already succeeded. This yields differences in the term sharing structure due to, for example, hash table collisions happening in different order, and the runs eventually diverge.

To conclude this section, our evaluation clearly favors using closure terms over terms in term orderings, and unidirectional comparisons over bidirectional comparisons. Therefore, our experimental setup on TODs (\Cref{sec:evaluation} and \Cref{app:emprical-eval}) uses {\tt conf3} as default.

\newpage
\section{Empirical Evaluation of TODs}
\label{app:emprical-eval}

In this appendix, we elaborate on our measurements already discussed in \Cref{sec:evaluation}. We performed two different experimental runs with (i) standard \vampire{} and (ii) profiled \vampire{}. In the former, we used the default measurement tools to benchmark overall solving time and instruction counts with a 60-second timeout. In the latter, we profiled various parts of \vampire{} related to TOD retrieval and maintenance, using more disruptive profiling measurements, which required a greater timeout. We discuss our profiling tooling in \Cref{sec:app-instr-count-measure}.

\subsection{Resource Usage}
\label{app:runtimes}
\Cref{fig:scatter-kbo,fig:scatter-lpo} compare the runtime of solved problems between the different values of {\tt -fdtod} using standard \vampire{}. For some benchmarks, using these different option values resulted in either different proofs, or no proof at all.
In these figures, we filtered out problems with different proofs between the two compared methods; blue crosses always have the same proof. Naturally, when only one algorithm succeeds, this criterion cannot be applied. These problems are shown with red triangles and green circles.

For example, upon close inspection of the problem {\tt PLA047\_1}, shown with a red triangle in the \Cref{fig:scatter-lpo} {\tt -sa discount -to lpo} subfigure, we could establish that the proof diverges independently of the performance of TODs.

Every dot below the diagonal line favors the newer approach. Further, those plots highlight that TODs are able to solve some problems significantly faster than the previous approach. We observe that the algorithms introduced in this paper improve the solving performance on almost all benchmarks. Further, comparing \Cref{fig:scatter-kbo,fig:scatter-lpo} highlights that LPO benefits the most from TODs.

\begin{remark}
    A keen reader might notice that the number of problems solved in \Cref{tab:vampire-vanilla} do not match the problems solved in \Cref{fig:scatter-kbo,fig:scatter-lpo}. This is due to the filtering applied for the generation of the figures. In \Cref{tab:vampire-vanilla}, we do not remove problems with different proofs.
\end{remark}

\paragraph{Time to solve.} \Cref{tab:runtimes} displays the cumulative solving time of all problems solved by \off, \on, \shared{} on different configurations. It helps us conclude that \vampire{} not only solved more problems (see \Cref{sec:evaluation}), but also solves them faster.
We reduce the overall runtime by up to 12\% with LPO.

\paragraph{Memory requirements.} In \Cref{sec:evaluation}, we mentioned that TODs seemed to consume more memory overall. We also hinted that there is a bias introduced because TODs allow the generation of more clauses before the timeout is reached. In \Cref{tab:memory}, we only considered the problems solved with the same proof (and number of clauses) and show that TODs indeed do not consume a significant amount of memory.

\begin{remark}
    Once again, there is a subtle difference in the filtering process between \Cref{tab:runtimes,tab:memory} and \Cref{fig:scatter-kbo,fig:scatter-lpo}. In the tables, we filter the problems for which one of the algorithms had a different proof. However, in the scatter plots, we only consider the proofs of the considered methods. For example, if the proof of a problem is identical with \off{} and \on{} but not \shared{}, it will appear on the plot comparing \off{} and \on{}, but not in the table.
\end{remark}

\begin{table}[t]
    \caption{Time for solving problems by \vampire{} vanilla. The first column shows the configuration considered, the second displays the number of problems that are solved by all three TOD configurations (\texttt{off/on/shared}). The next three columns show the total time to solve the aforementioned problems with their respective TOD setting.}
    \centering
    \setlength{\tabcolsep}{1em}
    \begin{tabular}{|l l| c | c | c | c |}
        \hline
                 &      & Solved by all & \off & \on & \shared \\
        \hline
        Otter    & KBO & 3159 & 3h22    & 3h17   & 3h15\\
                 & LPO & 2977 & 2h48    & 2h41   & 2h35\\\hdashline
        Discount & KBO & 3161 & 3h10    & 3h03   & 3h01\\
                 & LPO & 2958 & 2h46    & 2h32   & 2h27\\
        \hline
    \end{tabular}
    \label{tab:runtimes}
\end{table}

\begin{table}[t]
    \caption{Total memory used by \vampire{} to solve TPTP problems. The first column shows the configuration considered, the second displays the number of problems that are solved by all three TOD configurations (\texttt{off/on/shared}). The next three columns show the total memory used to solve the aforementioned problems with their respective TOD setting.}
    \centering
    \setlength{\tabcolsep}{1em}
    \begin{tabular}{|l l| c | c | c | c |}
        \hline
                 &      & Solved by all & \off & \on & \shared \\
        \hline
        Otter    & KBO & 3159 & 488 GB & 485 GB & 482 GB \\
                 & LPO & 2977 & 433 GB & 431 GB & 438 GB \\
        \hdashline
        Discount & KBO & 3161 & 530 GB & 523 GB & 530 GB \\
                 & LPO & 2958 & 460 GB & 460 GB & 474 GB \\
        \hline
    \end{tabular}
    \label{tab:memory}
\end{table}

\label{sec:scatter-plots}
\begin{figure}
    \begin{minipage}{.49\linewidth}
        \centering
        \includegraphics[width=\textwidth]{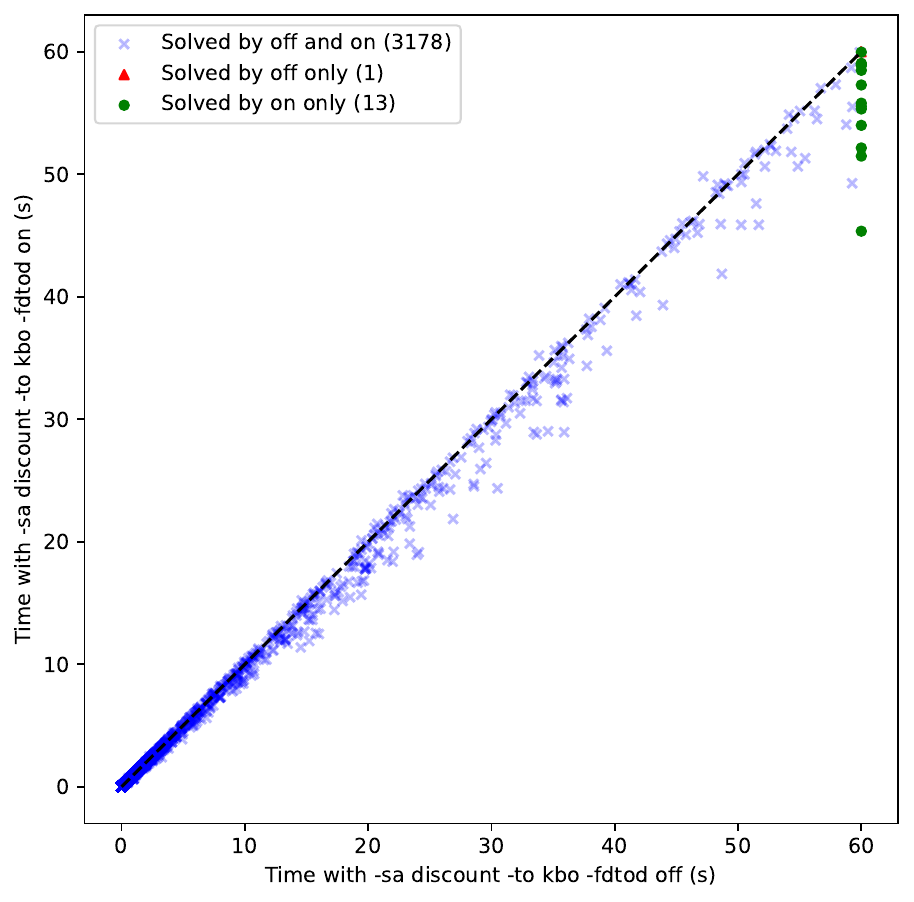}    
        \includegraphics[width=\textwidth]{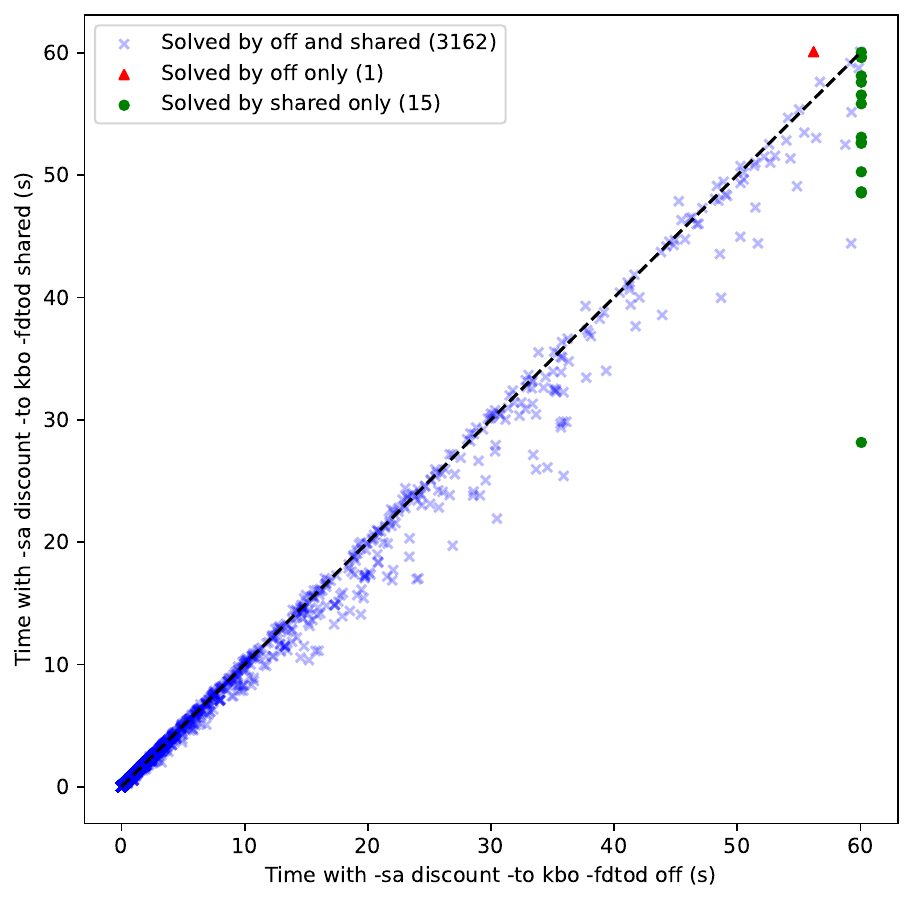}
        \includegraphics[width=\textwidth]{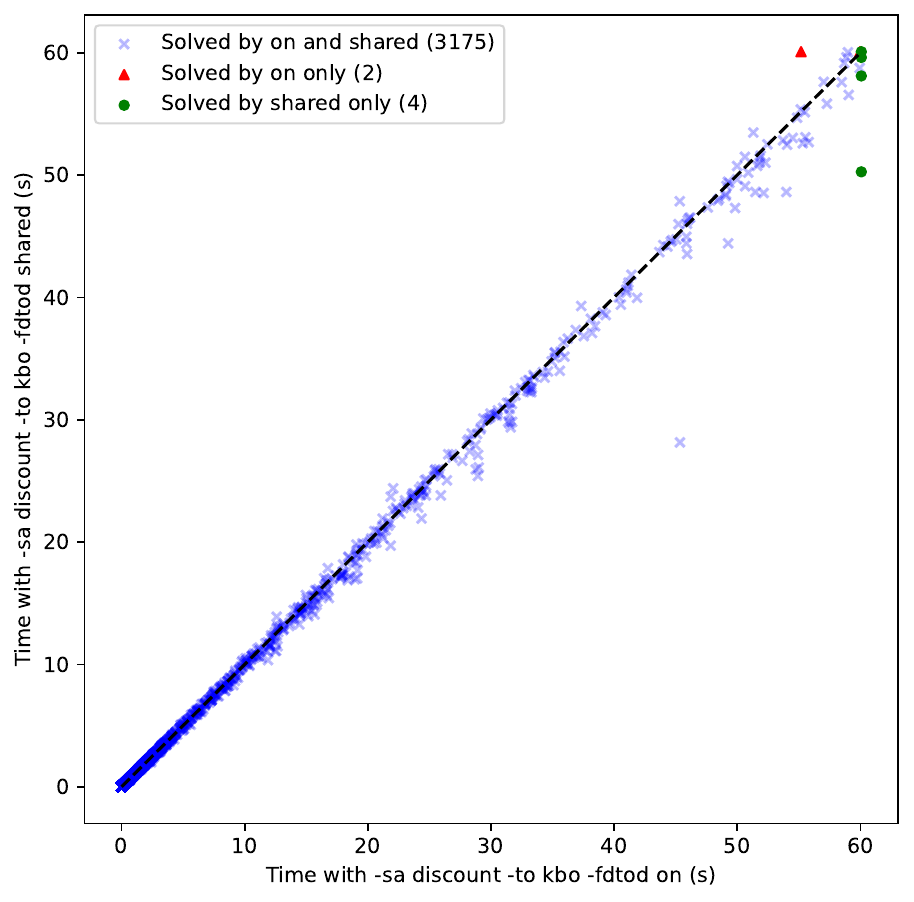}
    \end{minipage}
    \hfill
    \begin{minipage}{.49\linewidth}
        \centering
        \includegraphics[width=\textwidth]{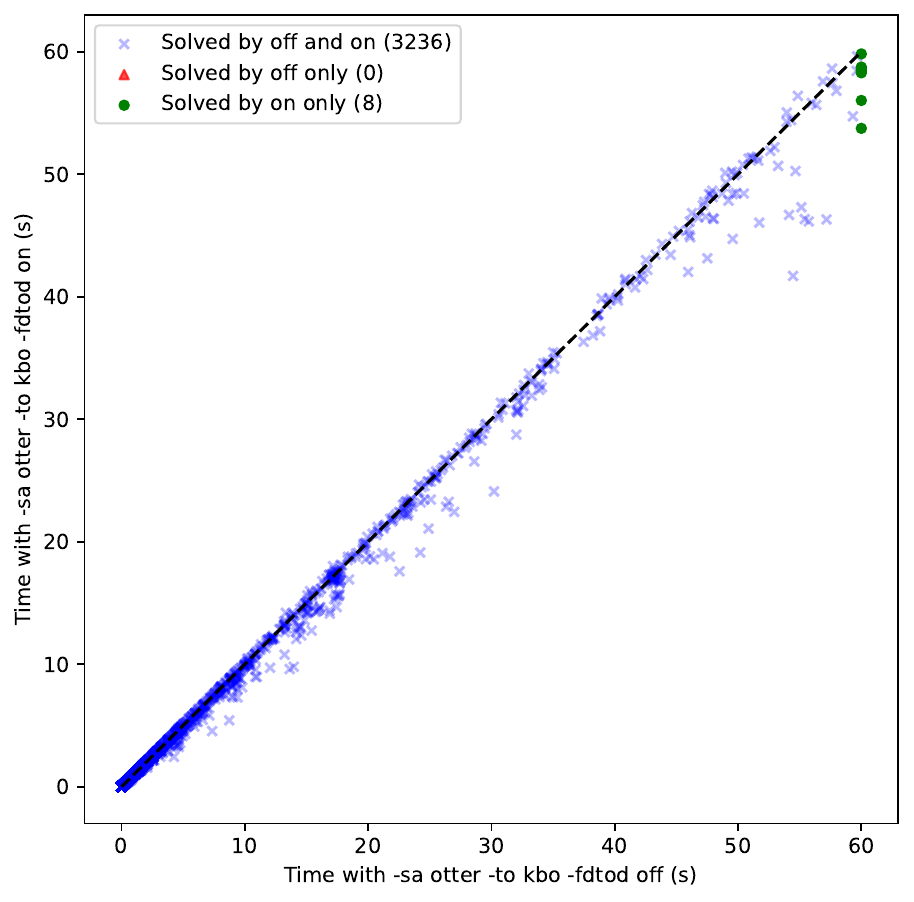}
        \includegraphics[width=\textwidth]{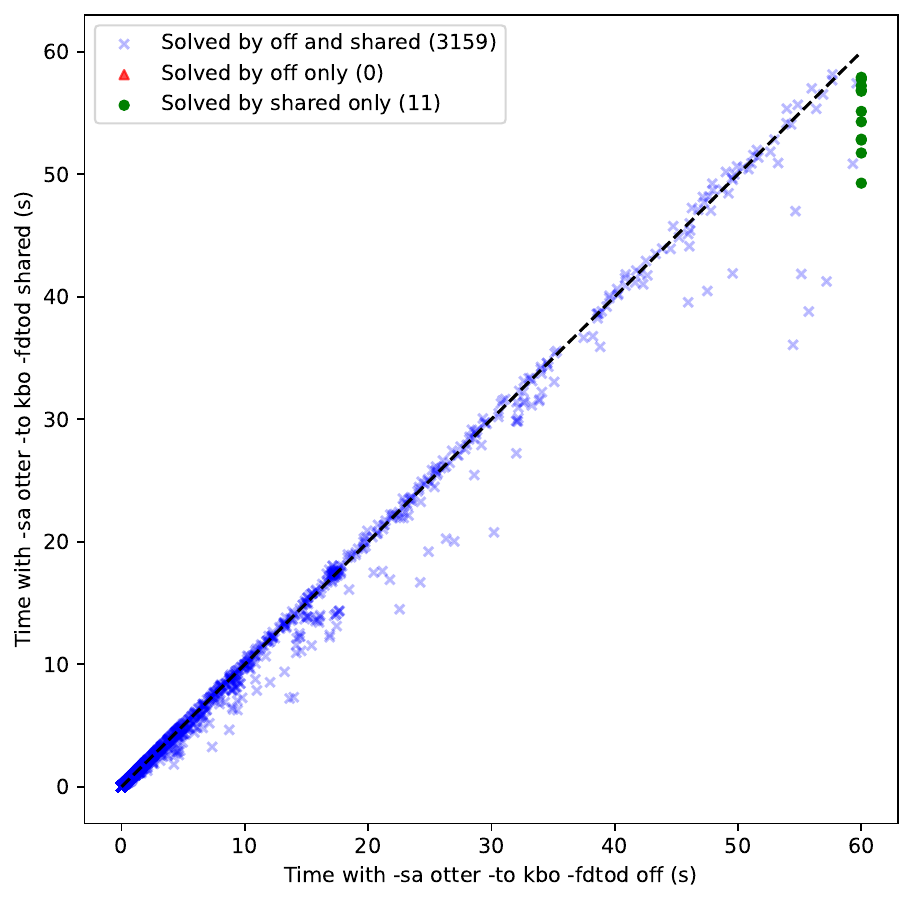}
        \includegraphics[width=\textwidth]{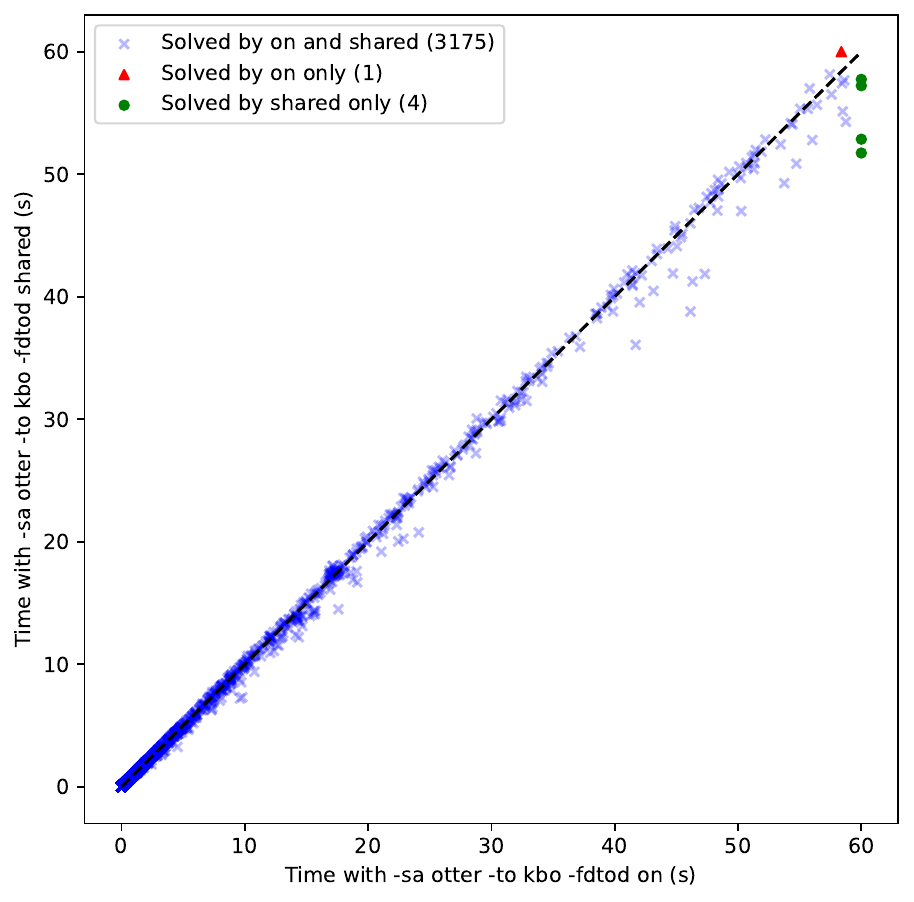}
    \end{minipage}
    \caption{\vampire{} run times with \texttt{-to kbo}. We show {\tt -sa discount} in the left column and {\tt -sa otter} in the right column. In each column, we compare {\tt -fdtod} values \off{} with \on{} (top), \off{} with \shared{} (middle), and \on{} with \shared{} (bottom).}
    \label{fig:scatter-kbo}
\end{figure}
\begin{figure}
    \begin{minipage}{.49\linewidth}
        \centering
        \includegraphics[width=\textwidth]{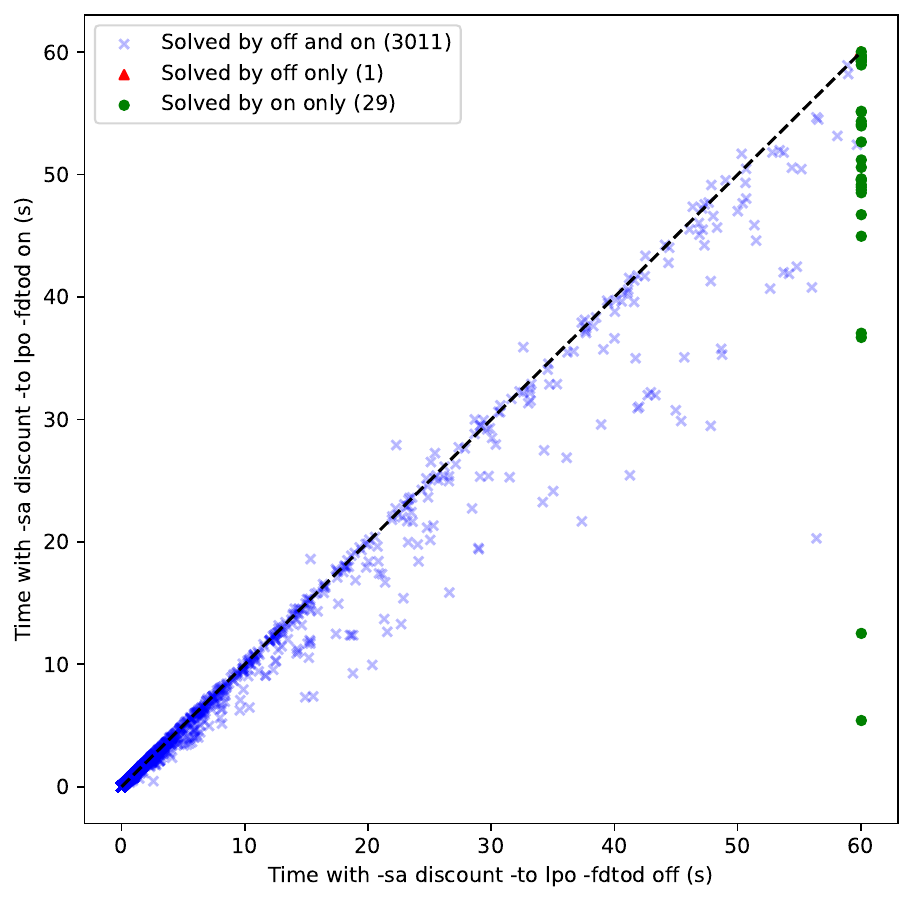}    
        \includegraphics[width=\textwidth]{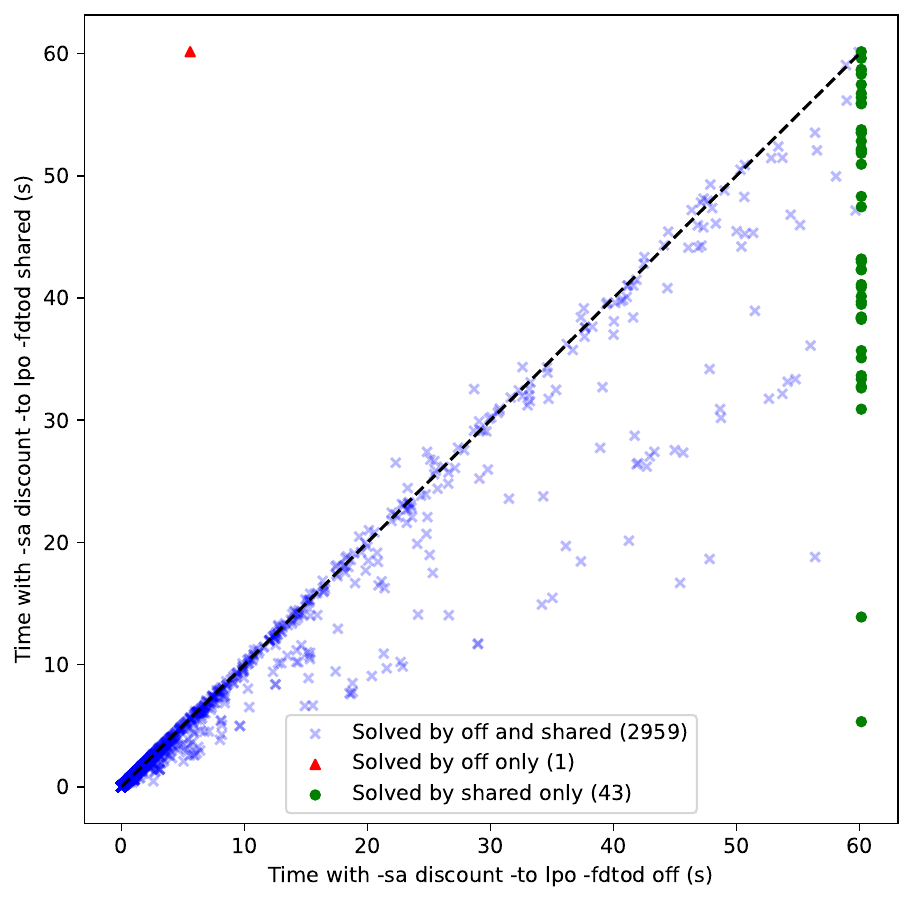}
        \includegraphics[width=\textwidth]{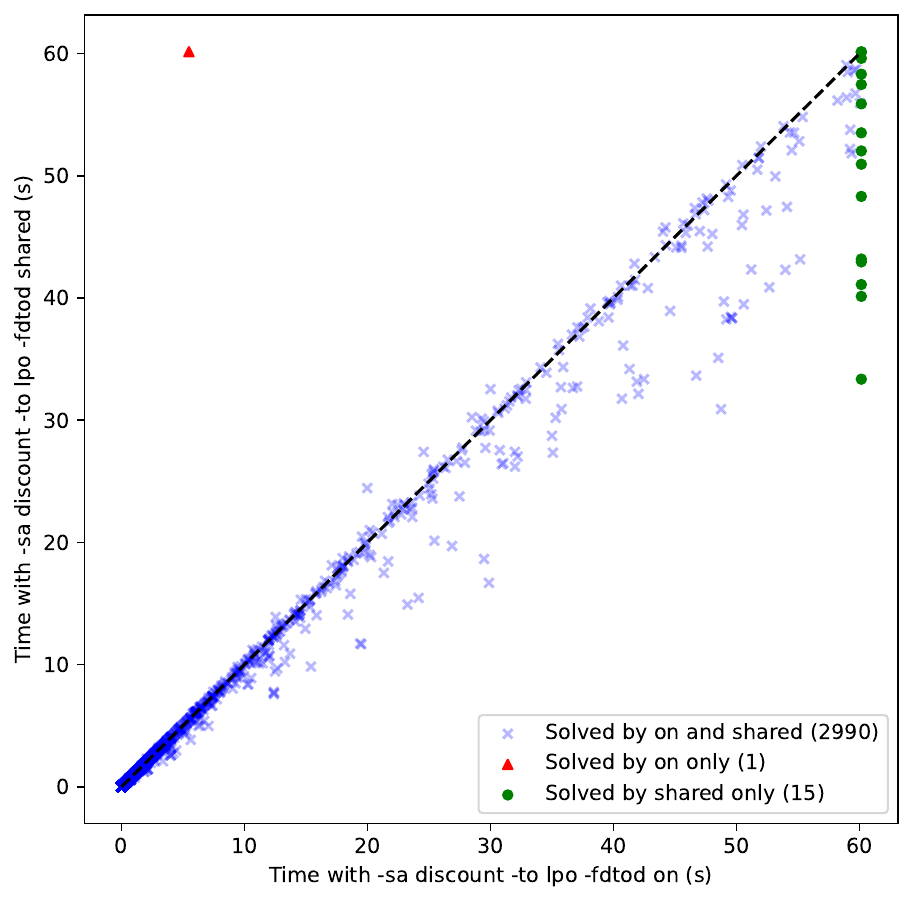}
    \end{minipage}
    \hfill
    \begin{minipage}{.49\linewidth}
        \centering
        \includegraphics[width=\textwidth]{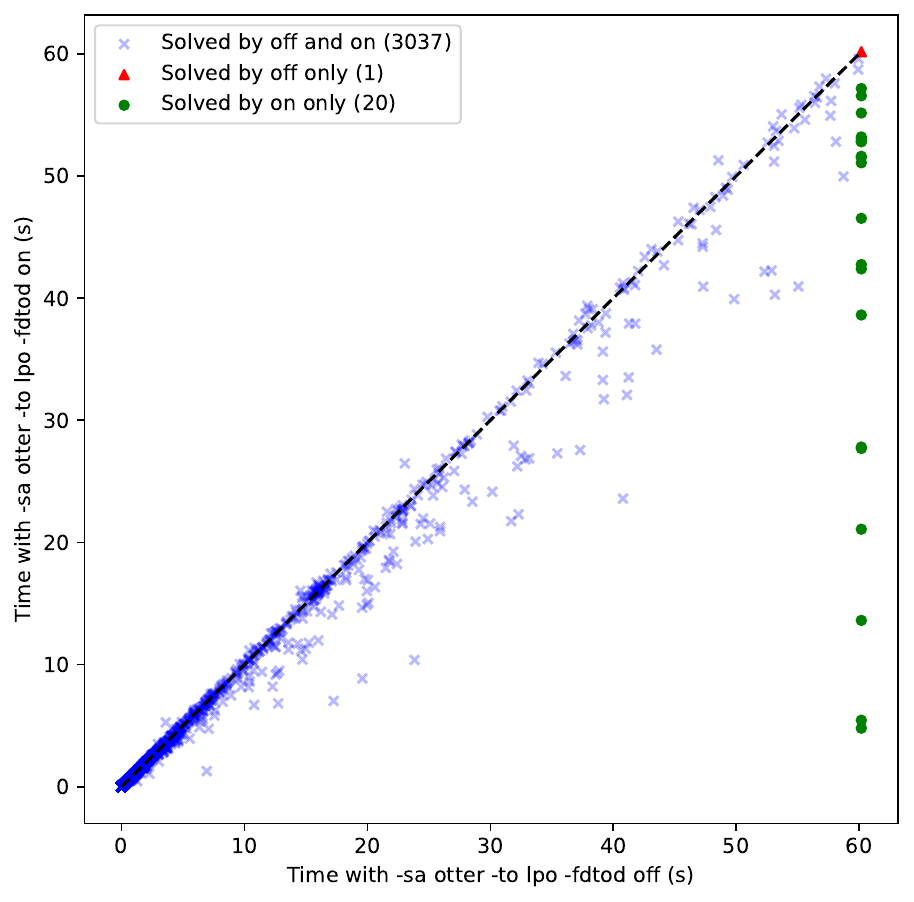}    
        \includegraphics[width=\textwidth]{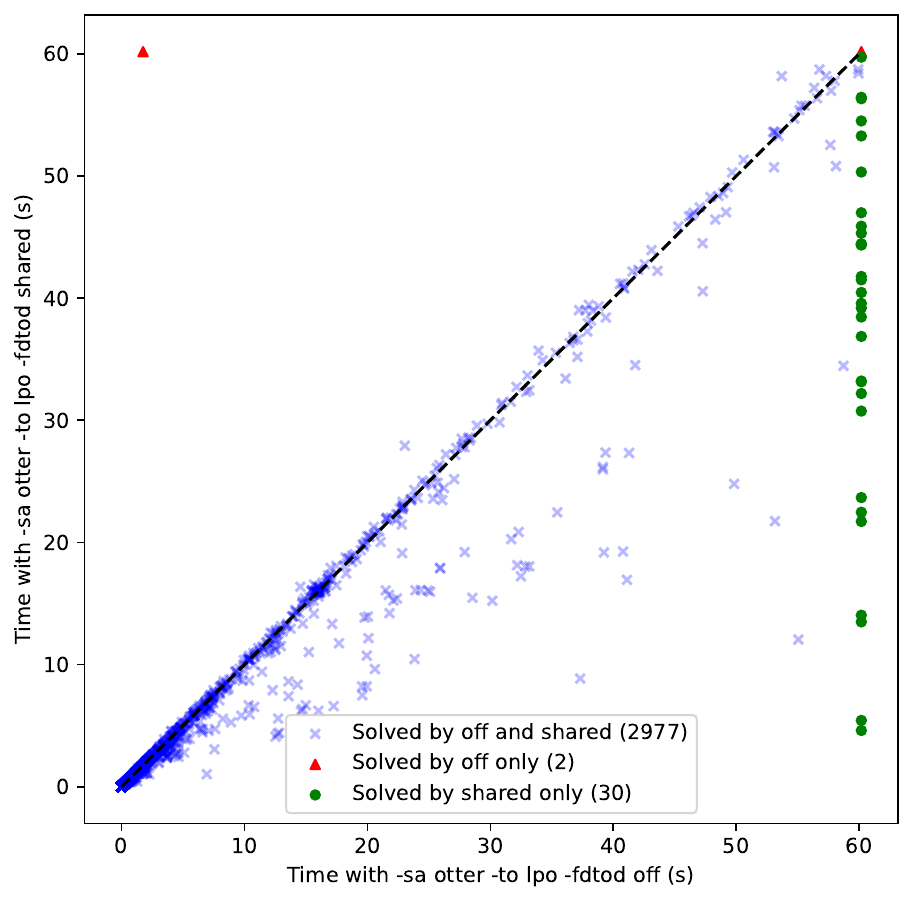}
        \includegraphics[width=\textwidth]{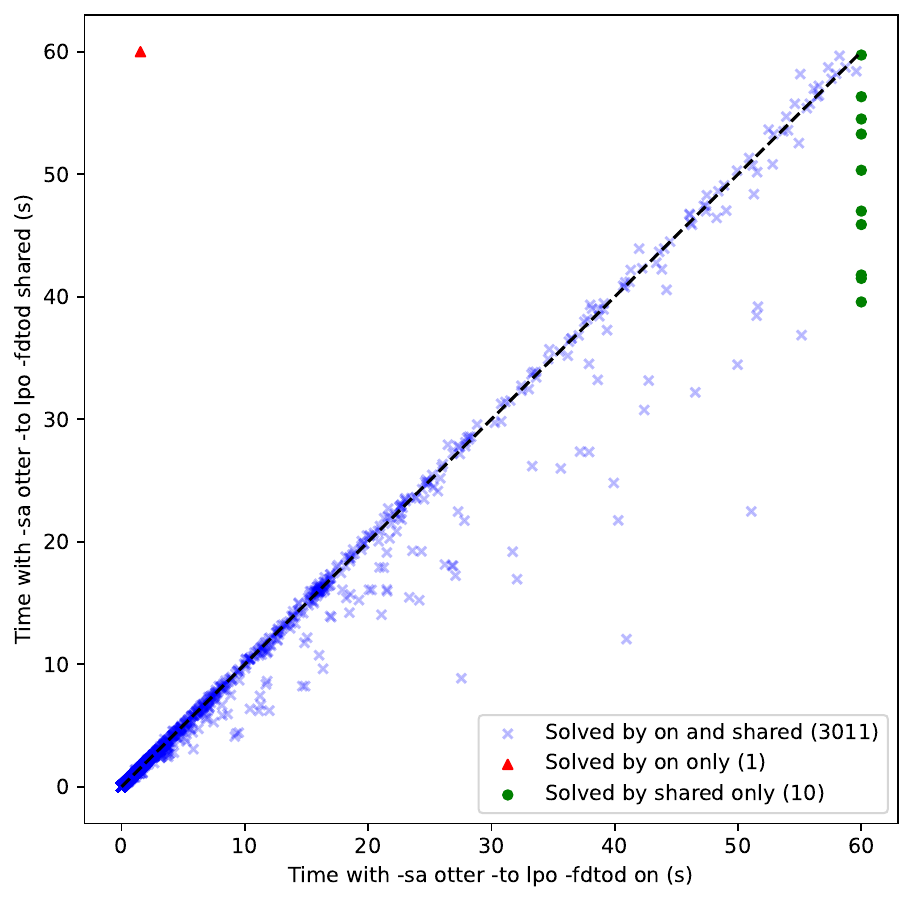}
    \end{minipage}
    \caption{\vampire{} run times with \texttt{-to lpo}. We show {\tt -sa discount} in the left column and {\tt -sa otter} in the right column. In each column, we compare {\tt -fdtod} values \off{} with \on{} (top), \off{} with \shared{} (middle), and \on{} with \shared{} (bottom).}
    \label{fig:scatter-lpo}
\end{figure}

\subsection{Distribution of Computation Time}
\Cref{fig:boxplots,fig:boxplots-post-order} show the portion of time spent in forward demodulation compared to an entire profiled \vampire{} run within $200 \times 10^9$ instructions.
Those figures are a more detailed representation of the data presented in \Cref{tab:instr-count} of \Cref{sec:evaluation}. In addition to showing that TODs reduce the number of instructions spent on post-ordering checks on average, we can see that the number of problems with the highest demodulation cost decreases significantly when TODs are introduced.
For example, this value decreased for problem {\tt NUN133-1} with configuration {\tt -sa discount -to lpo} from 99.65\% when using {\tt -fdtod off} to only 62.89\% when using {\tt -fdtod shared}. This effectively multiplied the number of instructions spent outside forward demodulation by an order of 100.

\begin{remark}
    There are differences between the proportion of instructions displayed in \Cref{tab:instr-count} and \Cref{fig:boxplots,fig:boxplots-post-order}. They are not errors but differences of calculations methods. In \Cref{tab:instr-count}, all instruction counts were summed before the ratios were computed. In \Cref{fig:boxplots,fig:boxplots-post-order}, the ratios are computed for each problem before averaging them.
\end{remark}

\begin{figure}
    \centering
    \includegraphics[width=0.49\textwidth]{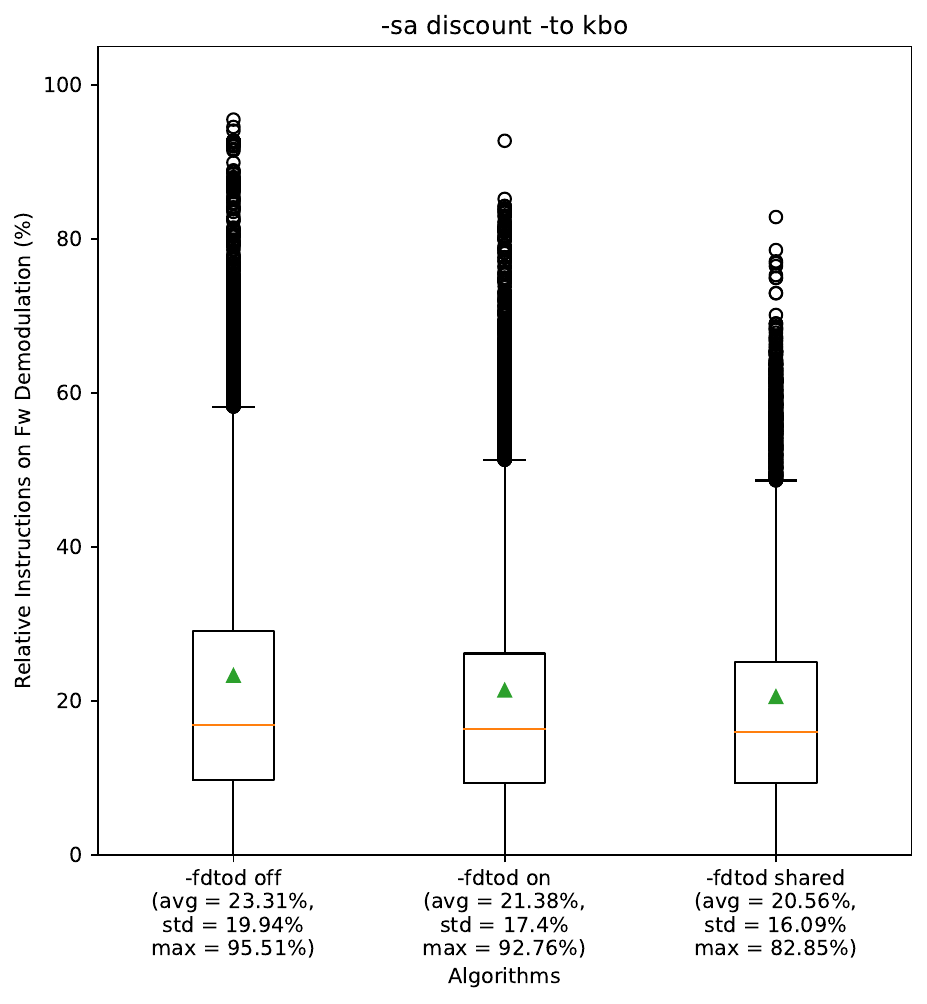}
    \includegraphics[width=0.49\textwidth]{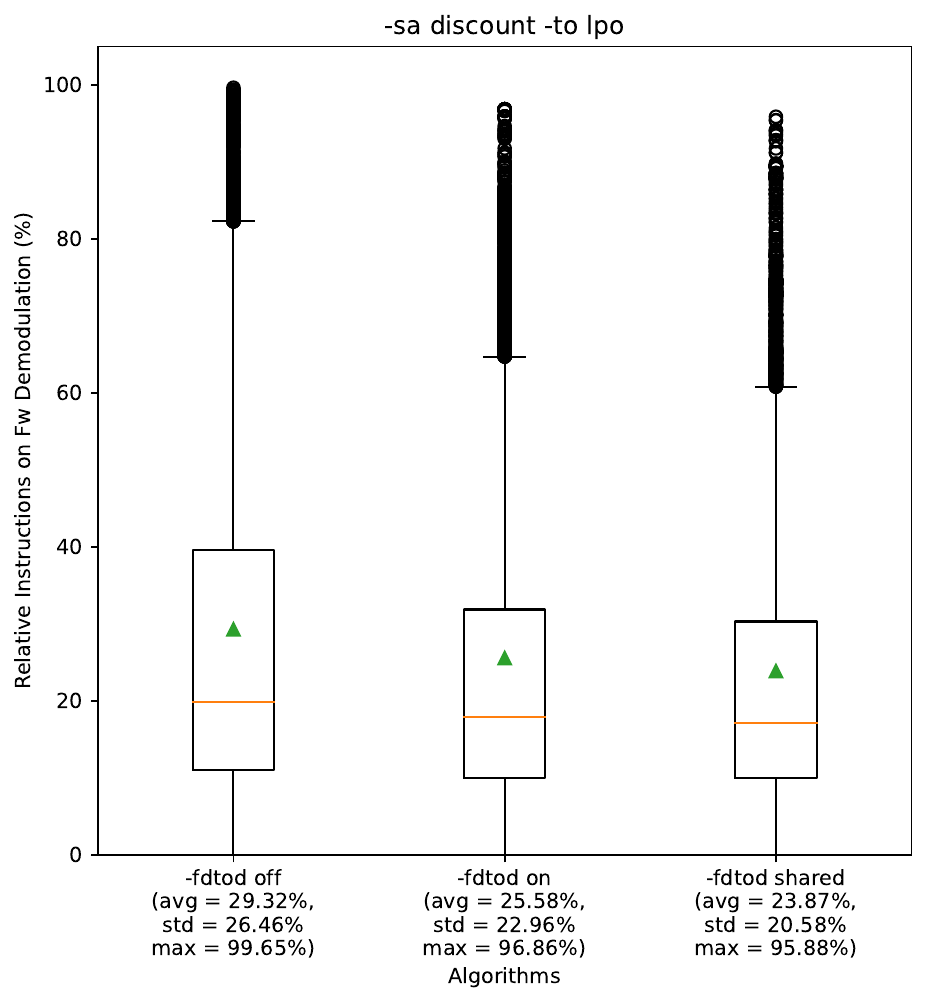}
    \includegraphics[width=0.49\textwidth]{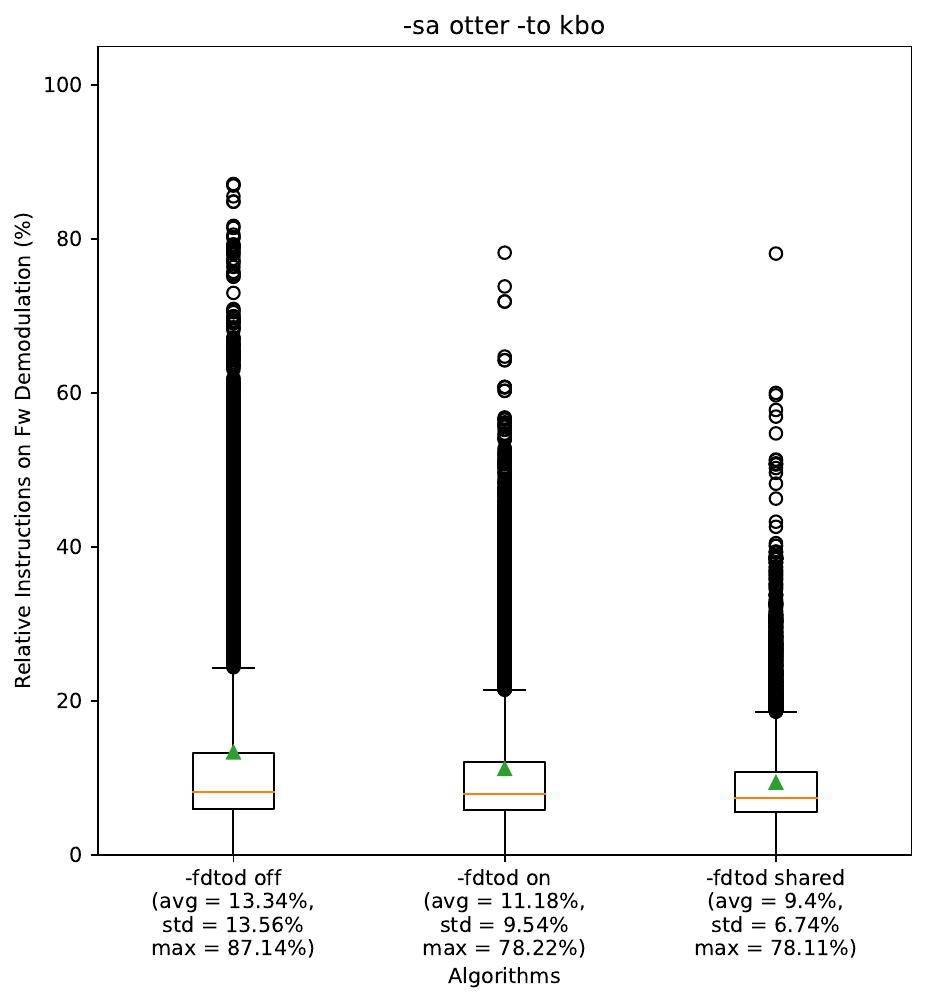}
    \includegraphics[width=0.49\textwidth]{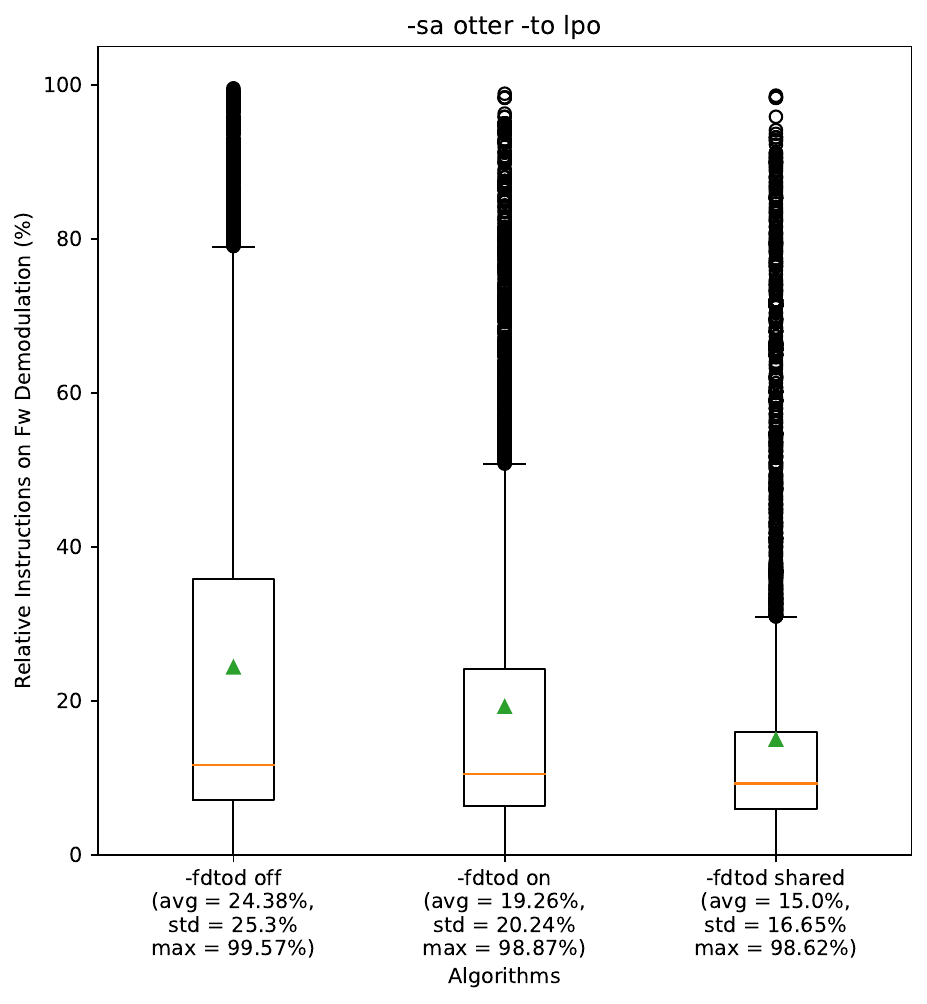}
    \caption{Proportion of number of instructions spent in forward demodulation over the overall number of instructions. We show {\tt -to kbo} in the left column, {\tt -to lpo} in the right column; {\tt -sa discount} in the top row, and {\tt -sa otter} in the bottom row. All runs were limited to $200\times 10^9$ instructions.}
    \label{fig:boxplots}
\end{figure}

\begin{figure}
    \centering
    \includegraphics[width=0.49\textwidth]{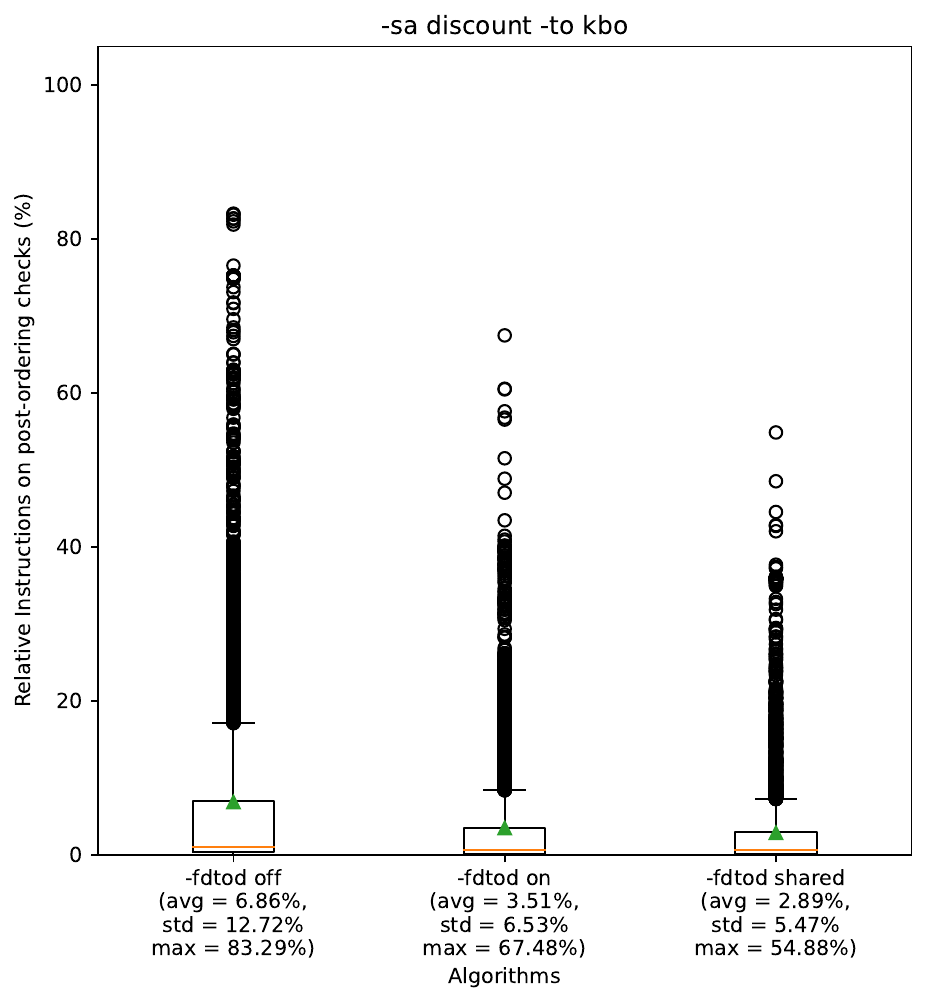}
    \includegraphics[width=0.49\textwidth]{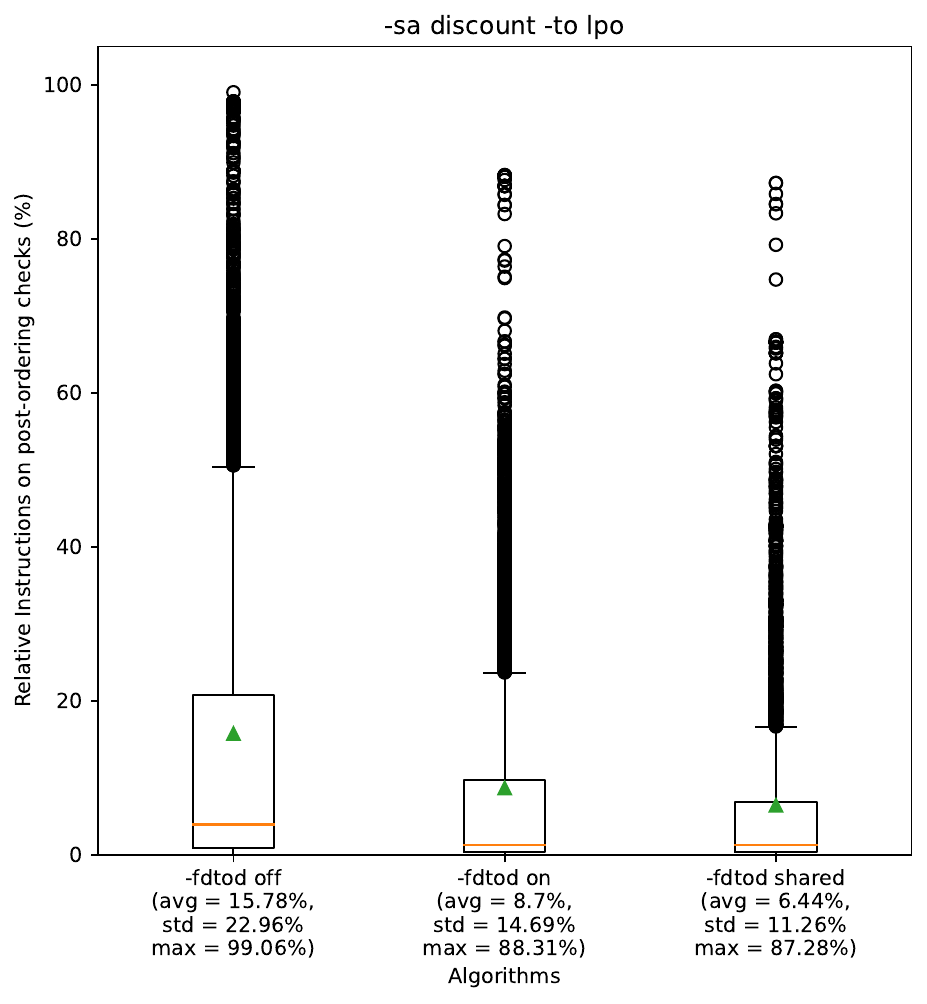}
    \includegraphics[width=0.49\textwidth]{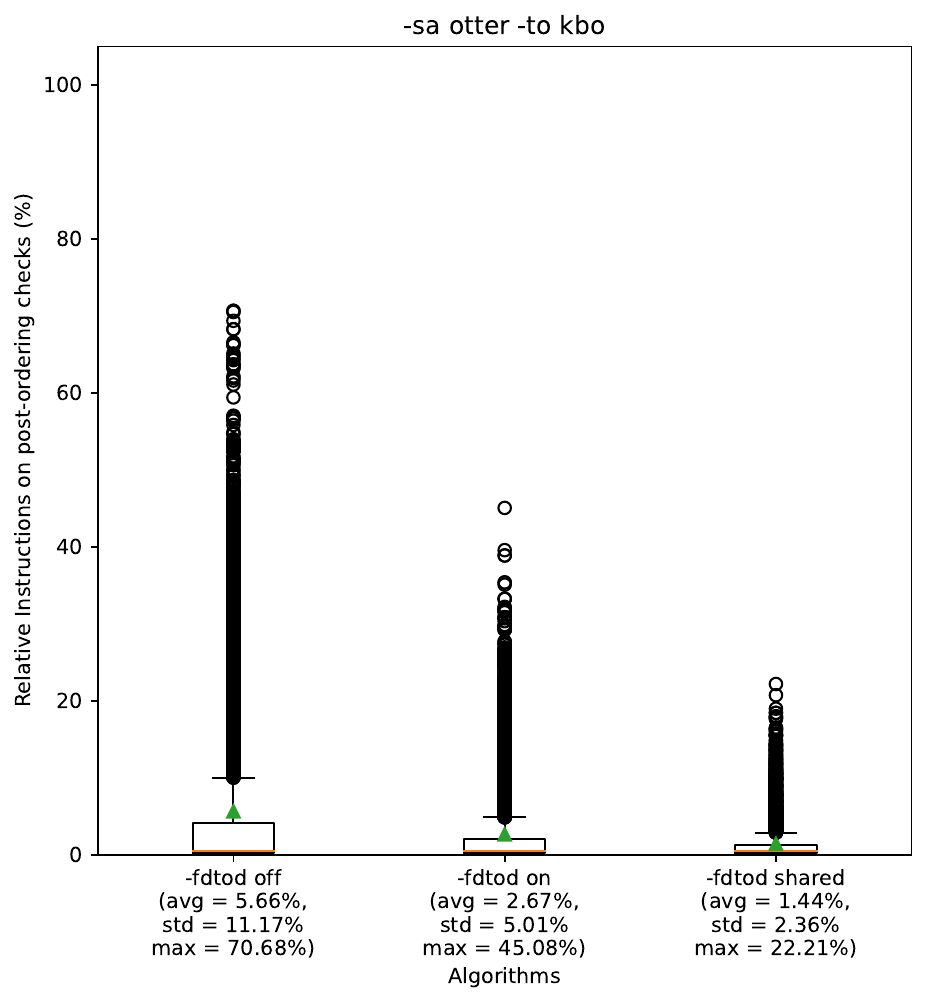}
    \includegraphics[width=0.49\textwidth]{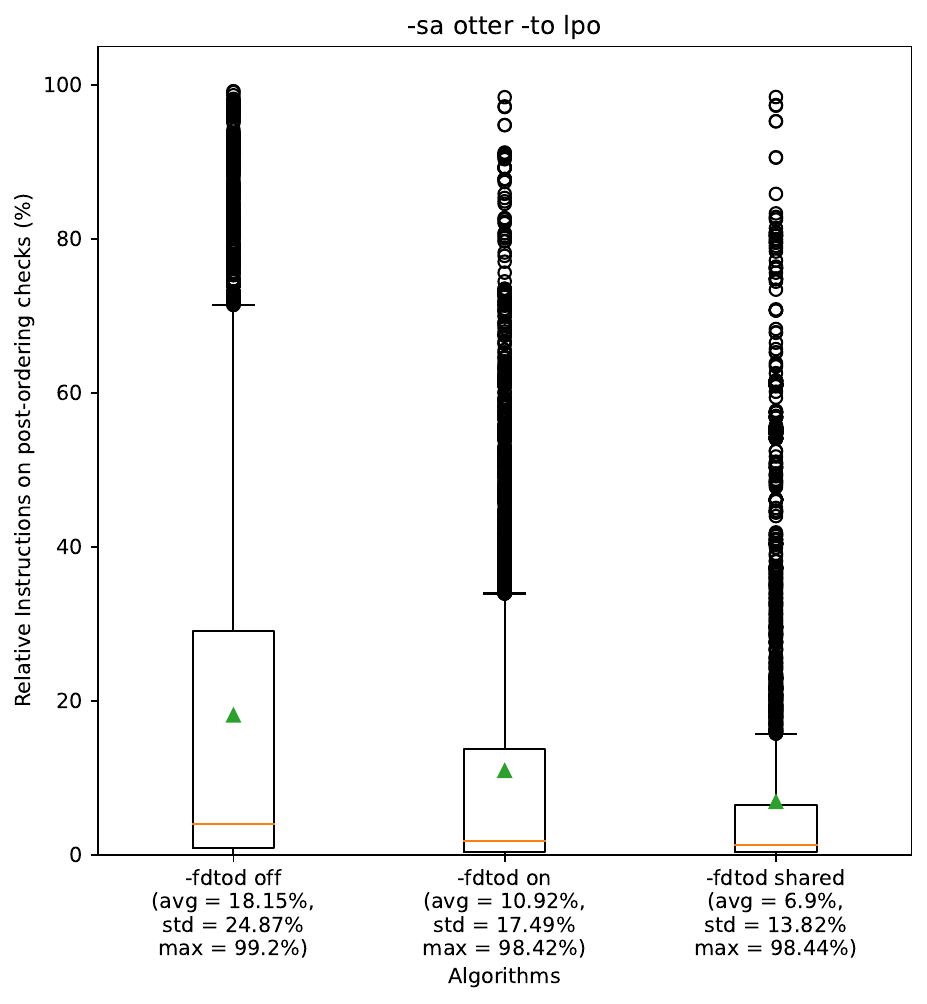}
    \caption{Proportion of number of instructions spent in post-ordering checks over the overall number of instructions. We show {\tt -to kbo} in the left column, {\tt -to lpo} in the right column; {\tt -sa discount} in the top row, and {\tt -sa otter} in the bottom row. All runs were limited to $200\times 10^9$ instructions.}
    \label{fig:boxplots-post-order}
\end{figure}

\begin{table}[t]
    \caption{Sharing of TODs for demodulators. The first column shows the \vampire{} options. The next two columns show the number of demodulators and TODs inserted. with {\tt -fdtod on}, since there is one TOD per demodulator, those two values are identical. The fourth column shows the sharing ratio of TODs. That is the average number of demodulators per TOD. The last two columns display the number of TOD retrievals performed: ``queries'' represent the number of retrievals for TODs, whereas ``answers'' is the number of times a success or exit node was traversed during retrievals.}
    \centering
    \renewcommand*\arraystretch{1.1}
    \setlength{\tabcolsep}{0.3em}
    \begin{tabular}{|l l l| r r | r | r r |}
        \hline
        \multicolumn{3}{|c|}{Options} & \multicolumn{2}{c|}{\# count $\times 10^{3}$} & Sharing ratio & \multicolumn{2}{c|}{Retrieval $\times 10^{6}$}\\
        \hdashline
                 &     &              & demodulators & TODs  &  & queries & answers \\
        \hline
\multirow{4}{0.5cm}{\rotatebox{90}{Otter}}    & \multirow{2}{0.8cm}{KBO}
              & \on     & 436 755 &  436 755 & 1.000 & 70 150 & 70 150 \\
         &    & \shared & 443 798 &  402 096 & 1.104 & 17 776 & 17 854 \\
\cdashline{2-8}
         & \multirow{2}{0.7cm}{LPO}
              & \on     & 296 168 &  296 168 & 1.000 &  100 241 &  100 241 \\
         &    & \shared & 306 878 &  274 271 & 1.119 & 23 922 & 23 964 \\
\cdashline{1-8}
\multirow{4}{0.5cm}{\rotatebox{90}{Discount}} & \multirow{2}{0.8cm}{KBO}
              & \on     & 14 328 & 14 328 & 1.000 &  126 773 &  126 773 \\
         &    & \shared & 14 396 & 13 907 & 1.035 & 50 396 & 50 432 \\
\cdashline{2-8}
         & \multirow{2}{0.7cm}{LPO}
              & \on     & 12 894 & 12 894 & 1.000 &  156 601 &  156 601 \\
         &    & \shared & 13 020 & 12 532 & 1.039 & 67 480 & 67 523 \\
        \hline
    \end{tabular}
    \label{tab:demodulator-stats}
\end{table}

\subsection{Distribution of Demodulators, TODs, and Retrievals}
\Cref{tab:demodulator-stats} shows the total number of demodulators, the overall number of TODs and the relative ratio of the latter two. Moreover, it shows the number of times TODs were queried and the number of answers (including failures) they provided during retrievals. Note that there are three orders of magnitude fewer demodulators, and hence TODs, in Discount configurations compared to in Otter configurations, which is due to the fact that Discount only uses a small subset of the clauses (so-called \emph{active clauses}) for forward demodulation, whereas Otter uses all generated and non-redundant clauses.

The third column of \Cref{tab:demodulator-stats} shows the ratio between the number of demodulators and TODs, which corresponds to the average number of distinct demodulators per TOD. For {\tt shared}, the column shows that most of the demodulators are not shared. The ratio is also smaller than the relative improvement of {\tt shared} over {\tt on} in all configurations (see Table~\ref{tab:runtimes}). This suggests that the demodulators that are actually shared appear more often when checking forward demodulation.

The fourth column of \Cref{tab:demodulator-stats} shows, first, that there is an at least three-order-of-magnitude difference between the number of TODs and number of retrievals performed on them. Second, when performing a retrieval, we rarely get multiple answers from a TOD retrieval, as most of the time we either fail or succeed with the first demodulator. This is not surprising as extra checks after the post-ordering check rarely have to be performed.

\begin{table}[t]
    \caption{Distribution of TOD nodes. The table shows the total number of term/data/poly nodes in TODs and their proportion of the total. The data originates from profiled \vampire{} runs within $200\times 10^9$ instructions.}
    \centering
    \renewcommand*\arraystretch{1.1}
    \setlength{\tabcolsep}{0.6em}
    \begin{tabular}{|l l l| r | r r r | r r r |}
        \hline
        \multicolumn{3}{|c|}{Options} & \multicolumn{4}{c|}{\# nodes $\times 10^{6}$} &  \multicolumn{3}{c|}{Ratio of nodes}\\
        \hdashline
                 &     &              & total & term & success & pos & term & success & pos \\
        \hline
\multirow{4}{0.5cm}{\rotatebox{90}{Otter}}    & \multirow{2}{0.8cm}{KBO}
              & \on     & 475.3 & 169.3 & 284.8 & 21.3 & 35.6\% & 59.9\% & 4.5\% \\
         &    & \shared & 1704.2 & 163.2 & 1501.8 & 39.2 & 9.6\% & 88.1\% & 2.3\% \\
\cdashline{2-9}
         & \multirow{2}{0.7cm}{LPO}
              & \on     & 760.1 & 568.9 & 191.2 & - & 74.8\% & 25.2\% & - \\
         &    & \shared & 2034.7 & 988.8 & 1046.0 & - & 48.6\% & 51.4\% & - \\
\cdashline{1-9}
\multirow{4}{0.5cm}{\rotatebox{90}{Discount}} & \multirow{2}{0.8cm}{KBO}
              & \on     & 15.1 & 5.0 & 9.7 & 0.5 & 33.1\% & 63.8\% & 3.1\% \\
         &    & \shared & 50.7 & 2.2 & 47.6 & 0.9 & 4.3\% & 93.9\% & 1.8\% \\
\cdashline{2-9}
         & \multirow{2}{0.7cm}{LPO}
              & \on     & 21.8 & 13.0 & 8.8 & - & 59.8\% & 40.2\% & - \\
         &    & \shared & 61.1 & 17.6 & 43.5 & - & 28.7\% & 71.3\% & - \\
        \hline
    \end{tabular}
    \label{tab:type-of-nodes}
\end{table}

\subsection{Distribution of Nodes}

\Cref{tab:type-of-nodes} shows the distribution of term comparison, positivity check, and success nodes at the end of the execution. Positivity nodes for the KBO runs constitute only a few percent of all nodes. The relative distribution of term and success nodes varies slightly, but in most configurations, success nodes dominate TODs. This relative difference is also greater in {\tt shared} variants than in {\tt on} variants, which is an indicator of how much redundancy we can avoid using, for example, redundant node removal, when sharing multiple demodulators.

In \Cref{tab:node-type-visits-stats}, we can observe that the number of created nodes is much lower than the number of processed nodes. This motivates our lazy approach to applying expensive processing steps. In addition, the number of node traversals is orders of magnitude higher than the number of processed nodes. This means that investing some computation budget to make later traversals faster will prove beneficial in the long run.

\begin{table}[t]
    \caption{Comparison of created and processed nodes classified by their type. The last three columns display the number of times those nodes are traversed during retrieval. The data originates from profiled \vampire{} runs within $200\times 10^9$ instructions.}
    \centering
    \renewcommand*\arraystretch{1.1}
    \setlength{\tabcolsep}{0.3em}
    \begin{tabular}{|l l l| r r r | r r r | r r r |}
        \hline
        \multicolumn{3}{|c|}{Options} & \multicolumn{3}{c|}{\# created $\times 10^{6}$} & \multicolumn{3}{c|}{\# processed $\times 10^{6}$} & \multicolumn{3}{c|}{\# traversed $\times 10^{9}$}\\
        \hdashline
                 &     &              & term & succ. & pos. & term & succ. & pos. & term & succ. & pos. \\
        \hline
\multirow{4}{0.5cm}{\rotatebox{90}{Otter}}    & \multirow{2}{0.8cm}{KBO}
              & \on     & 812.8 & 2164.9 & 36.1 & 17.0 & 2.1 & 31.9 & 63.2 & 68.3 & 7.8 \\
         &    & \shared & 210.4 & 2272.6 & 41.5 & 10.1 & 53.2 & 19.8 & 27.3 & 17.9 & 3.9 \\
\cdashline{2-12}
         & \multirow{2}{0.7cm}{LPO}
              & \on     & 1520.9 & 1478.8 & - & 141.9 & 2.4 & - & 205.7 & 98.8 & - \\
         &    & \shared & 1154.7 & 1584.4 & - & 93.4 & 29.3 & - & 61.9 & 24.0 & - \\
\cdashline{1-12}
\multirow{4}{0.5cm}{\rotatebox{90}{Discount}} & \multirow{2}{0.8cm}{KBO}
              & \on     & 25.3 & 70.3 & 0.9 & 0.4 & 0.4 & 0.8 & 113.7 & 124.2 & 12.5 \\
         &    & \shared & 3.7 & 71.5 & 1.0 & 0.3 & 7.8 & 0.7 & 80.0 & 50.4 & 10.0 \\
\cdashline{2-12}
         & \multirow{2}{0.7cm}{LPO}
              & \on     & 48.6 & 63.8 & - & 3.6 & 0.6 & - & 248.7 & 154.3 & - \\
         &    & \shared & 26.7 & 65.4 & - & 3.3 & 6.7 & - & 148.4 & 67.5 & - \\
        \hline
    \end{tabular}
    \label{tab:node-type-visits-stats}
\end{table}

\subsection{Unit Equality Results}
Our approach shines especially well in the UEQ (unit equality) category of TPTP.
\Cref{tab:runtimes-ueq} shows that our approach shaves up to 30\% off the solving time when LPO is combined with Discount.
\Cref{tab:ueq-res} displays the speed up of the forward demodulation part of \vampire.

\begin{table}[t]
    \caption{Time for solving UEQ problems by \vampire{}. Similar to \Cref{tab:runtimes}, this table displays the total solve time of UEQ problems.}
    \centering
    \renewcommand*\arraystretch{1.1}
    \setlength{\tabcolsep}{1em}
    \begin{tabular}{|l l| c | c | c | c |}
        \hline
                 &      & Solved by all & \off & \on & \shared  \\
        \hline
        \multirow{2}{1cm}{Otter}
                 & KBO & 516 & 29m 28s    & 27m 35s   & 26m 35s\\
                 & LPO & 453 & 20m 20s    & 17m 18s   & 14m 58s\\\hdashline
        \multirow{2}{1cm}{Discount}
                 & KBO & 465 & 29m 33s    & 27m 26s   & 26m 41s\\
                 & LPO & 413 & 21m 27s    & 16m 31s   & 15m 00s\\
        \hline
    \end{tabular}
    \label{tab:runtimes-ueq}
\end{table}

\subsection{Exact Instruction Measurement}
\label{sec:app-instr-count-measure}
As is usual in fine measurements, the act of measuring itself can introduce errors. If the computer is programmed to {\tt start} and {\tt stop} a clock without doing anything in between, the measurements will not be 0. Further, time measurements are very noisy. Experiments on our cluster showed that starting a clock on average takes $28ns$, but the variance was $340 ns^2$ (and depending on the execution of the program). This might be negligible when measuring expensive functions but becomes problematic when profiling small but regularly invoked methods such as ordering checks.

On the other hand, instruction counting is more reproducible and can be exactly computed. It is possible to calibrate an instruction counter that corrects its own impact on the measurements.

\paragraph{Calibration of instruction counters.}
When calibrating the counters, we have to take into account the inner and outer impact of the counter. When running a {\tt start/stop} cycle, the inner overhead is the amount that will be measured on top of the instructions between {\tt start} and {\tt stop}. The outer overhead is introduced by nested counters. Both overheads can be computed with two nested {\tt start/stop} cycles. The inner overhead is the number of instructions measured by the inner counter. The outer overhead is the difference between the number of instructions measured by the outer and inner counters. On our experimental setup, the inner and outer overheads were measured at 18 and 37 instructions respectively. The calibration is performed at runtime, adapting to the hardware architecture.

\paragraph{Correcting the overhead.}
Now that we know by how much we overshoot our measurements, we can offset them. To do so, we define a global variable to track the total overhead. When starting a new counter, it records this total. When the counter stops, it corrects the increment by removing the difference between the current and recorded value of the total overhead. The increment is further reduced by subtracting the inner overhead. Finally, the global overhead is incremented with the outer overhead.

With this method, we can measure precisely the number of instructions necessary to run even small parts of the program.

\paragraph{Tradeoff.}
For precise measurements, we pay two prices. i) Measuring the number of instructions invokes expensive system calls and is approximately one hundred times slower than time measurements (3 $\mu s$ per {\tt start/stop} cycle on our machine). Therefore, the computation budget is largely increased. ii) The number of instructions is not always a good metric of the efficiency of a program. Indeed, it does not consider cache behaviors, branch predictions, varying instruction costs, e.g. floating-point division or integer addition have different impacts.

This approach is, therefore, best complemented with coarser time measurements for a clearer picture of the performance gains.

\end{document}